\documentclass[10pt,usenames,dvipsnames]{article}
\usepackage{enumerate}
\usepackage{amsmath,amssymb}
\usepackage{natbib}
\usepackage{caption}
\usepackage{subcaption}
\usepackage[usenames]{color}
\usepackage{bm}
\usepackage{ying}
\usepackage{multirow}
\usepackage{rotating}
\usepackage{fullpage}
\usepackage{authblk}
\usepackage{enumerate}
\usepackage{geometry}
\usepackage{float}

\usepackage{mathtools}
\mathtoolsset{showonlyrefs}

\newcolumntype{g}{>{\columncolor{red}}c}

\usepackage{graphicx,amssymb}
\usepackage{xcolor}

\usepackage[colorlinks,
            linkcolor=red,
            anchorcolor=blue,
            citecolor=blue
            ]{hyperref}
\usepackage{algorithm}
\usepackage{algorithmic}

\let\hat\widehat
\let\tilde\widetilde

\def \iid {\stackrel{\textnormal{i.i.d.}}{\sim}}

\def \calib {\textnormal{calib}}
\def \train {\textnormal{train}}
\def \test {\textnormal{test}}
\def \obs {\textnormal{obs}}
\def \TPP {\tilde{\PP}}
\def \quant {\textnormal{Quantile}}
\def \defn {\,:=\,}
\def \super {\textnormal{sup}}
\DeclareMathOperator*{\esssup}{ess\,sup}
\def \nc {{n}}
\def \nt {{n_\textnormal{train}}}

\usepackage{paralist}

\usepackage{hyperref}
\def\##1\#{\begin{align}#1\end{align}}
\def\$#1\${\begin{align*}#1\end{align*}}
\newcommand{\revise}[1]{{\color{black} #1}}
\newcommand{\revisea}[1]{{\color{black} #1}}
\newcommand{\reviseb}[1]{{\color{black} #1}}

\makeatletter
\newcommand{\printfnsymbol}[1]{%
  \textsuperscript{\@fnsymbol{#1}}%
}
\makeatother

\usepackage{xcolor}

\title{Sensitivity Analysis of Individual Treatment Effects: \\
A Robust Conformal Inference Approach}
\author[1]{Ying Jin\thanks{Author names listed alphabetically.}}
\author[2]{Zhimei Ren\printfnsymbol{1}}
\author[1,3]{Emmanuel J. Cand\`es}
\affil[1]{Department of Statistics, Stanford University}
\affil[2]{Department of Statistics, University of Chicago}
\affil[3]{Department of Mathematics, Stanford University}
\date{}

\begin{document}

\maketitle

\begin{abstract}
  We propose a \revise{model-free} framework 
for sensitivity analysis of 
individual treatment effects (ITEs),  
building upon ideas
from conformal inference. 
For any unit, our procedure reports the {\em $\Gamma$-value}, a number
which quantifies the minimum strength of confounding needed to explain
away the evidence for ITE. 

Our approach rests on the reliable predictive inference of
counterfactuals and ITEs in situations where the training data is
confounded.  
Under the marginal sensitivity model
of~\citet{tan2006distributional}, we characterize the shift between
the distribution of the observations and that of the
counterfactuals. 
\reviseb{We first 
develop a general method 
for predictive inference of  
test samples 
from a shifted distribution;}
we then leverage this to construct
covariate-dependent 
prediction sets for counterfactuals. No matter the
value of the shift, these prediction sets (resp.~approximately)
achieve marginal coverage if the propensity score is known exactly
(resp.~estimated).  We describe a distinct procedure also attaining
coverage, however, conditional on the training data. In the latter
case, we prove a sharpness result showing that for certain classes of
prediction problems, the prediction intervals cannot possibly be
tightened.  
We verify the validity and performance of the new
methods via simulation studies and apply them to analyze real
datasets.

\end{abstract}
 
\section{Introduction}


Understanding the effect of a treatment is arguably one of the main
research lines in causal inference.  Over the past few decades, there
has been a rich literature in identifying, estimating and conducting
inference on the mean value of causal effects; parameters of interest
include the average treatment effect (ATE) or the conditional average
treatment effect (CATE).  These quantities, however, might fail to
provide reliable uncertainty quantification for individual responses:
the knowledge that a drug might be effective for a whole population
`on average' does not imply that it is effective on a particular
patient.
Taking the intrinsic variability of the responses into account,
inference on the individual treatment effect (ITE) may be better
suited for reliable decision-making. 
To  
quantify the uncertainty in 
individual treatment effects,~\cite{lei2020conformal} offered a novel
viewpoint. Rather than constructing confidence intervals for
parameters---e.g. the ATE---they proposed designing prediction
intervals for potential outcomes, namely, for counterfactuals and
ITEs. Briefly, Lei and Cand\`es constructed well calibrated prediction
intervals by building upon the conformal inference
framework~\citep{vovk2005algorithmic,vovk2009line}. In that work, the
typical mismatch between the counterfactuals and the observations due
to the selection mechanism is resolved with the strong ignorability
assumption~\citep{imbens2015causal}; that is to say, the treatment
assignment mechanism is independent of the potential outcomes
conditional on a set of observed 
covariates.  The strong ignorability assumption is automatically
satisfied in randomized experiments and commonly used in the causal
inference
literature~\citep{rubin1978bayesian,rosenbaum1983central,imbens2015causal}.
In observational studies, however, the strong ignorability assumption
is not testable~\citep[Chapter 12]{imbens2015causal} and hard to
justify in general.  In practice, failing to account for possible
confounding in observational data can yield misleading conclusions
~\citep{rutter2007identifying,fewell2007impact,gaudino2018unmeasured}. 

\subsection{$\Gamma$-values}  \label{subsec:intro_gval}

In this paper, we seek to understand the robustness of causal
conclusions on ITEs against potential unmeasured confounding.  To this
end, the procedure we propose starts from a sequence of hypothesized
confounding strengths whose precise meaning will be made clear
shortly; for each hypothesized strength, a prediction interval is
constructed for the ITE. The procedure then screens the prediction
intervals, and for each unit, reports the smallest confounding
strength with which the prediction interval contains zero:  
we call this the {\em $\Gamma$-value}. Informally, the
\revise{$\Gamma$-value} describes the strength of unmeasured
confounding necessary to explain away the predicted effect.

Imagine we want to test for a positive ITE on \revise{a treated unit}. 
For a range of hypothesized confounding strengths, our procedure
constructs one-sided \revise{prediction} intervals for \revise{the ITE} at the $1-\alpha = 0.9$
confidence level. The $\Gamma$-value is then the smallest confounding
strength with which the lower bound of the prediction interval is
smaller than zero. 
\revise{Figure~\ref{fig:intro_suv} shows the survival function of
the $\Gamma$-values calculated on a real dataset measuring 
the academic performance of students subject or not to mindset interventions 
(all the details are 
in Section~\ref{subsec:real}).}
\revise{We see that $20\%$ of the \revisea{students} 
have $\Gamma$-values greater than $1.05$; roughly
    speaking, for these \revisea{students} the confounding strength needs to be as
  large as $1.05$ to explain away the evidence for positive ITEs. We
  also see that $7.20\%$ of the \revisea{students} have $\Gamma$-values greater
  than $2$ and some \revisea{students} have
$\Gamma$-values as large as $5$, showing strong evidence for
positive ITEs.}

\begin{figure}[ht]
    \centering
    \includegraphics[width=0.56\textwidth]{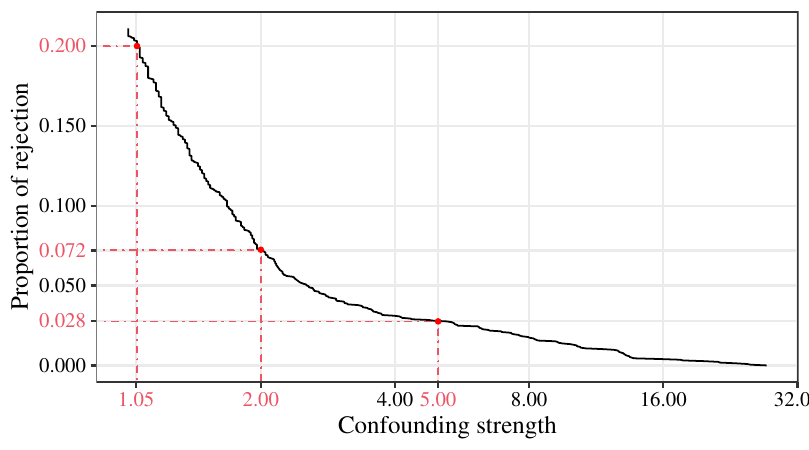}
    \caption{Proportion of test samples in the real dataset from
      Section~\ref{subsec:real} identified as positive ITE at each
      hypothesized confounding strength.  The confidence level is set
      to $1-\alpha=0.9$.}
    \label{fig:intro_suv}
\end{figure}

Formally, the $\Gamma$-value can be used to draw conclusions
on the ITE with a pre-specified confounding strength. For example,
if we believe the confounding strength is no larger than $2$, we can
claim an individual has positive ITE as long as its $\Gamma$-value
is greater than $2$. Our method guarantees that if the magnitude of
confounding is at most $2$, the probability of incorrectly
`classifying' a unit as having a positive ITE is at most
$\alpha=0.1$ (or any other fraction). 

Figure~\ref{fig:intro_xplot} plots the $\Gamma$-values as a
function of the achievement levels of the schools the students
belong to. Once more, our procedure provides valid inference on a
single ITE. This means that if the strength of confounding is at
most $2$, the chance that an individual \reviseb{with a negative 
ITE} has a $\Gamma$-value larger than $2$ is at most $\alpha$.  
Taking a step further, one might be
interested in the inference on a set of selected units (e.g., the
red points in Figure~\ref{fig:intro_xplot}), for which evidence on
multiple ITEs needs to be combined.  In a companion paper, we shall
consider simultaneous inference on multiple ITEs with proper error
control.

\begin{figure}[ht]
    \centering
    \includegraphics[width=\textwidth]{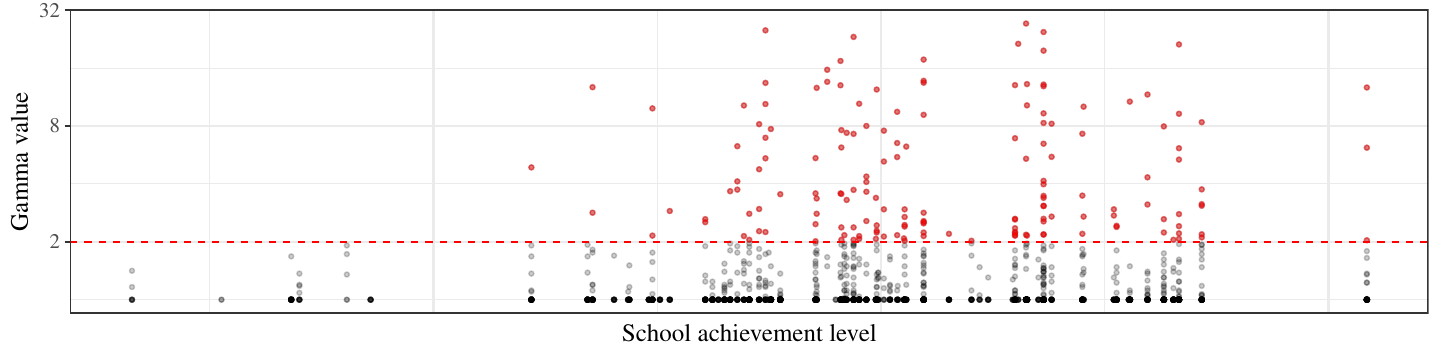}
    \caption{$\Gamma$-values versus school achievement level 
      for test samples. The red points correspond to the students 
      whose $\Gamma$-values are greater than $2$.} 	 
    \label{fig:intro_xplot}
\end{figure}

\subsection{Problem setup}
Throughout, we work under the potential outcome 
framework~\citep{neyman1923applications,imbens2015causal}.
Let $X \in \cX$ denote the observed 
covariates, $T\in\{0,1\}$
the assigned treatment, and $Y(1),
Y(0) \in \RR$ the potential outcomes. 
We assume there is an unobserved confounder $U\in\cU$ satisfying
\#\label{eq:u_ign}
\big(Y(1),Y(0)\big) \indep T \biggiven X, U,
\#
in which $U$ can be a random vector. 
As pointed out by~\citet{yadlowsky2018bounds},
a confounder satisfying~\eqref{eq:u_ign} always 
exists since one can take $U=(Y(1),Y(0))$.
In contrast, the strong ignorability assumption
states
\#\label{eq:ign}
\big(Y(1),Y(0)\big)\indep T \biggiven X.
\#
As well known, the latter assumes that we have measured sufficiently
many features so that the potential outcomes are independent of the treatment conditional on the covariates $X$. 

Suppose there are i.i.d.~samples
$\{(X_i, U_i, T_i, Y_i(0), Y_i(1))\}_{i\in\cD}$ from some distribution
$\PP$.  We adopt the commonly used stable unit treatment value
assumption (SUTVA) (see
e.g.,~\citet{cox1958planning,rubin1978bayesian,rubin1990comment,imbens2015causal}),
so that we observe
$Y_i=Y_i(T_i) = T_i \cdot Y_i(1) + (1-T_i) \cdot Y_i(0)$.  This means
that the
observed dataset consists of the random variables 
$(X_i,T_i,Y_i)_{i\in \cD}$.

Without further assumptions, the potential outcomes and the treatment
assignment mechanism can arbitrarily depend on $U$, making \revise{the
  estimation of treatment effects} impossible.  For example, imagine
we would like to 
assess the effect of a drug on patients. 
We are interested in $Y(1)$ (e.g., the survival time of the patient if
the drug is taken), and have available observational data recording
the treatment assignment $T$ and outcome $Y$.  Consider a confounded
setting, where the drug is assigned to patients based on an
undocumented factor $U$, namely, the patient's condition when admitted
to the hospital, so that only those in critical condition get treated.
Since $U$ is highly correlated with $Y(1)$, the survival times of
treated patients will likely be smaller than those in the whole
population (which is our inferential target), making the task of
identifying the effectiveness of the drug extremely difficult. 
As such, we shall work with confounders that only have limited effect
on the treatment assignment mechanism; the concept of ``limited
effect'' is formalized by the sensitivity models introduced below.

\subsection{Sensitivity models}
\label{subsec:sens_mdl}
A sensivity model characterizes the degree 
to which the data distribution
violates the strong ignorability assumption.
There has been a rich literature in 
designing different types of sensitivity models 
(see e.g.~\citet{zhao2017sensitivity} and the
references therein). 
In this paper, we work under 
the marginal sensitivity model 
\citep{tan2006distributional,zhao2017sensitivity}
on the unidentifiable super-population, 
characterized by the following 
marginal $\Gamma $-selection condition:

\begin{definition}[Marginal $\Gamma $-selection]
A distribution $\PP$ over $(X,U,T,Y(0),Y(1))$ satisfies
the marginal $\Gamma $-selection condition if for 
$\PP$-almost all $x\in \cX$ and $u\in \cU$, 
\#\label{eq:gamma_sel_mgn}
\frac{1}{\Gamma } \leq \frac{\PP(T=1\given X=x,U=u)/\PP(T=0\given X=x,U=u)}{\PP(T=1\given X=x )/\PP(T=0\given X=x )} \leq \Gamma .
\#
\end{definition}
Under the marginal $\Gamma $-selection condition, 
no matter how one changes the value of the confounder,
the odds of being treated conditional on the covariates
and the confounder will at most be off by a factor of $\Gamma $
when compared to the odds of being treated conditional 
{\em only} on the covariate. Therefore, the effect 
of confounders on the selection bias is bounded.   

The above marginal sensitivity model is  
closely related to Rosenbaum's sensitivity
model~\citep{rosenbaum1987sensitivity} 
\revise{and its generalizations~\citep{yadlowsky2018bounds}}, characterized
by the following $\Gamma $-selection condition:

\begin{definition}[$\Gamma $-selection]
A distribution $\PP$ over $(X,U,T,Y(0),Y(1))$ satisfies
the $\Gamma $-selection condition if for 
$\PP$-almost all $x\in \cX$ and $u,u'\in \cU$, 
\#\label{eq:gamma_sel}
\frac{1}{\Gamma } \leq \frac{\PP(T=1\given X=x,U=u)/\PP(T=0\given X=x,U=u)}
{\PP(T=1\given X=x,U=u' )/\PP(T=0\given X=x,U=u')} \leq \Gamma .
\#
\end{definition}
As pointed out by~\citet[Prop. 7.1]{zhao2017sensitivity}, the 
$\Gamma $-selection condition is stronger 
than the marginal $\Gamma $-selection condition in the
sense that a distribution $\PP$ over $(X,U,T,Y(0),Y(1))$
satisfying the $\Gamma $-selection condition
must also satisfy the marginal $\Gamma $-selection condition.

In the following, we let $\PP^\super$ denote the unknown 
super-population over $(X,U,T,Y(0),Y(1))$ 
that generates the partial observations $\cD$. 
For any $\Gamma \geq 1$,  
$\cP(\Gamma  )$ is the set of super-populations
that satisfy the marginal $\Gamma $-selection condition.

\subsection{Prediction intervals for counterfactuals}
The crux of our 
approach is to construct reliable \revise{prediction} intervals for
counterfactuals when the observations satisfy the marginal sensitivity
model.  
\reviseb{In Section~\ref{sec:mgn_interval},
  we propose a generic robust weighted conformal procedure, which is applied to
 counterfactual prediction
 in Section~\ref{sec:dist_shift}.}
Suppose we are interested in $Y(1)$ and $(X_{n+1},Y_{n+1}(1))$
is a test sample from the super-population (the results apply to other
types of counterfactuals as well).  
Given a nominal level $1-\alpha$ and a fixed confounding level
$\Gamma\geq 1$, the prediction interval $\hat{C}(X_{n+1},\Gamma )$ we
construct from the confounded data ensures \$ \PP\big(Y_{n+1}(1)
\in \hat{C}(X_{n+1},\Gamma )\big) \ge 1- \alpha - \hat{\Delta} \$ as
long as $\PP^\super \in \cP(\Gamma)$; here, the probability is over
the confounded observations $(X_i,T_i,Y_i)_{i\in\cD}$ and the test
sample $(X_{n+1},Y_{n+1}(1))$.  The error term $\hat\Delta = 0$
if the propensity score $e(x) = \PP(T = 1 \given X = x)$ of the
observed data is known exactly, and otherwise depends on the
\revise{estimation of the propensity score}.  

In practice, researchers may want to control the risk 
of falsely rejecting a hypothesis on an individual 
treatment effect \emph{given} the data at hand, $\cD$.  
Section~\ref{sec:cal_interval} offers a sister procedure with {\em
  probably approximately correct} (PAC)-type guarantee. Once more,
suppose $\PP^\super \in \cP(\Gamma)$. Then given any
$\delta,\alpha >0$ and any $\Gamma\geq 1$, we can construct a
prediction interval $\hat{C}(X_{n+1},\Gamma )$ such that 
\$ \PP\big(Y_{n+1}(1) \in \hat{C}(X_{n+1},\Gamma ) \given \cD \big)
\ge 1 - \alpha - \hat\Delta \$ holds with probability at least
$1-\delta$ over the randomness of $\cD$.  As before, the error term
$\hat{\Delta}$ depends on the estimation of the propensity score
$e(x)$. 
Since any distribution $\PP$ satisfying the $\Gamma $-selection
condition must also satisfy the marginal $\Gamma $-selection
condition, our methods also provide valid prediction intervals for
counterfactuals under Rosenbaum's sensitivity model.

\subsection{Related work}
The idea of sensitivity analysis dates back
to~\citet{cornfield1959smoking} who studied the causal effect of
smoking on developing lung cancer.  The authors concluded that if an
unmeasured confounder---hormone in their example---were to rule out
the causal association between smoking and lung cancer, it needed to
be so strongly associated with smoking that no such factors could
reasonably exist.  The approach of~\citet{cornfield1959smoking}
requires that both the outcome and the confounder are binary and that
there are no covariates. Whereas
\citet{bross1966spurious,bross1967pertinency} used the same conditions
later on, subsequent works substantially relaxed these assumptions.
\citet{rosenbaum1983assessing} proposed a sensitivity model to work
with categorical covariates.  Later, under Rosenbaum's sensitivity
model, a series of
works~\citep{rosenbaum1987sensitivity,gastwirth1998dual,rosenbaum2002attributing,rosenbaum2002observational}
further extended sensitivity analysis to broader settings by studying
samples with matching
covariates.~\citet{imbens2003sensitivity},~\citet{ding2016sensitivity}
and~\citet{vanderweele2017sensitivity} also considered unmeasured
confounders with `limited' effect, but under different sensitivity
models.  More recently,~\cite{tan2006distributional} proposed the
marginal sensitivity model, and~\citet{zhao2017sensitivity} proposed a
construction of bounds and confidence intervals for the
ATE under this model.  Their result was recently sharpened
by~\citet{dorn2021sharp}.  Bringing a distributionally robust
optimization perspective to the sensitivity analysis
problem,~\citet{yadlowsky2018bounds} studied the estimation of CATE
under Rosenbaum's sensitivity model.

Our contribution is to provide robustness for the inference procedure
against a proper level of confounding.  This bears similarity with
several works conducting `safe' policy evaluation and policy
learning under certain sensitivity models (see
e.g.,~\citet{namkoong2020off,kallus2021minimax}).  In contrast to the
estimation and learning tasks, we provide well-calibrated uncertainty
quantification for counterfactuals, which calls for a different set of
techniques.

Another closely related line of work is conformal inference, which is
the tool we employ (and improve) towards robust quantification of
uncertainty.  Conformal inference was pioneered and developed by
Vladmir Vovk and his collaborators in a series of papers (see
e.g.,~\citet{vovk2005algorithmic,vovk2009line,gammerman2007hedging,shafer2008tutorial,vovk2012conditional,vovk2013transductive}). 
In recent years, the technique has been broadly used for 
establishing statistical guarantees for learning algorithms 
(see e.g.,~\citet{lei2014distribution,lei2018distribution,lei2019fast,romano2020classification,cauchois2021knowing}). 
In particular,~\citet{lei2020conformal} studied the counterfactual 
prediction problem with conformal inference tools 
under the strong ignorability condition. 
The PAC-type guarantee for conformal prediction sets 
is also studied in~\citet{vovk2012conditional} 
without distributional shifts. 

In addition, our robust prediction perspective is related
to~\citet{cauchois2020robust}, which also studies the construction of
robust prediction sets.  That said, the setting considered there is
substantially different from ours; we will expand on this in
Section~\ref{sec:mgn_interval}. 
\reviseb{~\cite{park2021pac} 
constructs PAC-type
prediction sets under 
an identifiable covariate shift,
with some robustness features. 
We provide in Section~\ref{sec:cal_interval} a robust PAC-type
procedure as well; however, our methods apply to partially
identifiable distributional shifts and are distinct from
the rejection sampling strategy used in \cite{park2021pac}.
}

Finally, we note that 
in an independent and concurrent paper,
\citet{yin2021conformal} also develops sensitivity analysis for the ITE 
under the marginal sensitivity model; 
they define the marginal sensitivity model without positing a latent confounder 
and provide an alternative derivation of our Lemma~\ref{lem:shift}. 
They propose a robust weighted conformal inference procedure 
that is equivalent to Algorithm~\ref{alg:mgn}, while 
the analyses are complementary and offer different perspectives. 
Our current work additionally presents a distinct algorithm achieving 
the PAC-type coverage, and establishes the sharpness of our procedure in certain cases. 
We also interpret sensitivity analysis as a multiple testing procedure 
and 
prove that the type-I error is simultaneously controlled over all super-populations.

\subsection{Outline of the paper}
The rest of the paper is organized as follows:
\begin{itemize}

\item Sections~\ref{sec:mgn_interval} to~\ref{sec:cal_interval} concern
  the development of robust counterfactual inference procedures. 
 In Section~\ref{sec:mgn_interval}, we develop a general 
robust weighted conformal inference procedure; 
we show  
in Section~\ref{sec:dist_shift}  
how to apply it to 
construct valid counterfactual prediction sets.  
We propose a distinct procedure in
Section~\ref{sec:cal_interval} with 
PAC-type coverage, and establish a sharpness result.  

\item Section~\ref{sec:sens_analysis} expresses sensitivity analysis
  as a sequence of hypotheses testing problems, and gives a
  statistical interpretation of the $\Gamma$-value.
  Simulation studies explain how the $\Gamma$-value relates
    to the true effect size and actual confounding level.

\item  
  Section~\ref{sec:synthetic} evaluates the proposed method on a
  semi-synthetic dataset to examine its validity and
  applicability. Finally, our sensitivity analysis framework is used
  to draw causal conclusions on a real dataset.
\end{itemize}

\section{Robust weighted conformal inference}
\label{sec:mgn_interval}
 
\revisea{We begin by considering a generic predictive inference
  problem under distributional shift. We will connect this  to counterfactual
  inference under a marginal sensitivity model in 
  Section~\ref{sec:dist_shift}.

Suppose we have i.i.d.~training data $(X_i,Y_i)_{i\in \cD}$ 
from some distribution $\PP$, and 
an independent test sample $(X_{n+1},Y_{n+1})$  
from some
possibly different distribution $\TPP$.  
We consider a general setting where 
$\TPP$ is ``within bounded distance'' from 
$\PP$,
in the sense that, for some fixed
functions $\ell(\cdot)$ and $u(\cdot)$,
it belongs to the  identification set 
defined as
\#\label{eq:id_set}
\cP(\PP,\ell,u) = \Big\{  ~\tilde\PP\colon \ell(x) \le  
\frac{\ud \tilde\PP}{\ud \PP}(x,y) \le
u(x) ~~\PP\textrm{-almost surely}  \Big\}.
\#
The task is to provide a calibrated prediction interval $\hat{C}(X_{n+1})$ 
for $Y_{n+1}$.}

Equation~\eqref{eq:id_set} actually identifies a new class of
distributional robustness problems.   
As we shall see later
\reviseb{in Section~\ref{subsec:dist_shift},} 
our model~\eqref{eq:id_set} is motivated by
sensitivity analysis 
and is quite distinct from other models in the
literature. 
For instance, \citet{cauchois2020robust} considers a
setting in which 
the target distribution $\TPP$ is assumed to be within an
$f$-divergence ball with radius $\rho$ centered around $\PP$, so that
the identification set is 
\#\label{eq:id_set_fdiv} 
\cQ(\PP,\rho) =
\Big\{~\TPP: D_f\big( \TPP \,\|\, \PP \big) \le \rho \Big\}.  
\# 
Instead of bounding the overall shift
  in~\eqref{eq:id_set_fdiv}, the constraint in~\eqref{eq:id_set}
  actually allows freedom in the shift of $X$; to be sure, the
  set~\eqref{eq:id_set} can be small as long as $\ell(x)$ and $u(x)$
  are close.
  For counterfactual prediction under the strong ignorability
  condition in~\cite{lei2020conformal}, the identification
  set~\eqref{eq:id_set} happens to be a singleton even if $\PP_X$ and
  $\TPP_X$ are drastically different, while \eqref{eq:id_set_fdiv}
  would require a large $\rho$ to hold.  More generally, when there is
  a (approximately) known large shift in the marginal distribution
  $\PP_X$ but a relatively small shift in the conditional
  $\PP_{Y\given X}$, \eqref{eq:id_set} provides a tighter range of the
  target distributions.  
  Finally, the pointwise constraint
  in~\eqref{eq:id_set} 
  (as opposed to the average form) makes it naturally compatible with
  the weighted conformal inference framework
  of~\cite{TibshiraniBCR19}.

\subsection{Warm up: weighted conformal inference}
Before introducing our method, it is best to start by a brief recap of
the weighted (split) conformal inference procedure. 
Assume the likelihood ratio
$w(x,y) = \frac{\ud \tilde\PP}{\ud \PP}(x,y)$ is \emph{known exactly}.
The dataset $\cD$ is first randomly split into a training fold
$ \cD_\train$ of cardinality $\nt$ and a calibration fold
$ \cD_\calib $ of cardinality $\nc$.  We use $\cD_\train$ to train any
nonconformity score function $V\colon \cX\times \cY\to \RR$ which
measures how well $(x,y)$ ``conforms'' to the calibration samples: the
smaller $V(x,y)$, the better $(x,y)$ conforms to the
\revise{calibration samples}; see e.g.,~\citet{gupta2019nested} for
examples of nonconformity scores. We then define $V_i = V(X_i,Y_i)$
for $i\in \cD_\calib$.  For any hypothetical value
$(x,y)\in \cX\times\cY$ of the new data point, we assign weights to
the samples as \$ &{p}_i^w(x,y) :=
\dfrac{w(X_i, Y_i)}{\sum_{j=1}^{\nc} w(X_j,Y_j) + w(x,y) },~i=1,\dots,\nc,\\
&p_{n+1}^w(x,y) := \dfrac{w(x,y)}{\sum_{j=1}^\nc w(X_j,Y_j) + w(x,y)
}.  \$ \revise{ For any random variable $Z$, define the quantile
function as
$\text{Quantile}(q,Z) = \inf\{z\colon \PP(Z\leq z)\geq q\}$, and let
$\delta_a$ denote a point mass at $a$.}  Then a level $(1-\alpha)$
prediction interval is given by \# &\hat{C}(X_{n+1}) = \big\{y:
V(X_{n+1},y) \le
\hat{V}_{1-\alpha}(y)\big\},\label{eq:qt_wcp}\\
\textrm{where}~~ &\hat{V}_{1-\alpha} (y) = \quant \Big(1-\alpha,~
\sum_{i=1}^\nc p_i^w ( X_{n+1},y )\cdot\delta_{V_i} + p_{n+1}^w
(X_{n+1},y) \cdot\delta_{\infty}\Big).  \notag \# \reviseb{The prediction
interval~\eqref{eq:qt_wcp} is shown by~\cite{TibshiraniBCR19} to obey
$\TPP(Y_{n+1} \in \hat{C}(X_{n+1})) \ge 1-\alpha$, and it is computable
when $w(x,y)$ is a known function of $x$ only. In our context,
$w(x,y)$ depends on $x$ only when $\cP(\PP,\ell,u)=\{\tilde\PP\}$ is a
singleton---this is exactly the case in~\cite{lei2020conformal}, where
the strong ignorability condition~\eqref{eq:ign} is assumed.}

\subsection{Robust weighted conformal inference procedure}
\label{subsec:robust_wcp}

\revisea{Now suppose we have a pair of functions $\hat{\ell}$ and
  $\hat{u}\colon \cX\to \RR^+$ with $\hat\ell(x)\leq \hat u(x)$ for
  all $x\in \cX$.  In general, we expect $\hat\ell(x)$
  (resp. $\hat{u}(x)$) to serve as a pointwise lower (resp. upper)
  bound on the unknown likelihood ratio $w(x,y)$, although our
  theoretical guarantee does not depend on this. }

Proceeding as before and denoting  
$\cD_\calib = \{1,\dots,\nc\}$, 
we train a nonconformity score 
function $V\colon \cX\times \cY \mapsto \RR$ 
on $\cD_\train$. 
\revisea{We also allow $\hat{\ell}(\cdot)$ and $\hat{u}(\cdot)$ 
to be obtained from $\cD_\train$.} 
Let  
$[1], [2],\dots,[n]$ be a permutation of $\{1,2,\ldots,n\}$ 
such that $V_{[1]}\leq V_{[2]}\leq \cdots \leq V_{[n]}$. 
Defining  $\ell_i =  \hat\ell(X_i)$ and $u_i = \hat u(X_i)$
for $1\leq i\leq n$, and $u_{n+1} = \hat u(X_{n+1})$,
we construct the prediction interval  
\#\label{eq:ci_mgn}
  \hat{C}(X_{n+1}) = \big\{y: V(X_{n+1}, y) \leq V_{[k^*]}\big\},
\#
where 
\#\label{eq:mgn_qtl}
k^* = \min\big\{k: \hat{F}(k) \geq 1-\alpha\big\},\quad \hat{F}(k) =\frac{\sum^k_{i=1} \ell_{[i]}  }{\sum^k_{i=1} \ell_{[i]}
+ \sum^\nc_{i=k+1}u_{[i]} + u_{n+1}}.
\# 
The thresholding function $\hat{F}(k)$ in~\eqref{eq:mgn_qtl} 
is monotone in $k$, hence a linear search suffices 
to find $k^*$. 
We summarize the procedure in Algorithm~\ref{alg:mgn}.

\begin{algorithm}[htbp]
\caption{Robust conformal prediction: the  marginal procedure}\label{alg:mgn}
\begin{algorithmic}[1]
\REQUIRE Calibration data $\cD_\calib$, bounds $\hat\ell(\cdot)$, $\hat{u}(\cdot)$, 
non-conformity score function $V\colon \cX\times \cY\to \RR$, test covariate $x$,
target level $\alpha\in(0,1)$. 
\vspace{0.05in}
\STATE For each $i\in \cD_\calib$, 
compute $V_i = V(X_i,Y_i)$
\STATE For each $i\in \cD_\calib$, 
compute $\ell_i = \hat{\ell}(X_i)$ and $u_i = \hat{u}(X_i)$.
\STATE Compute $u_{n+1} = \hat{u}(x)$. 
\STATE For each $1\leq k\leq \nc$, 
compute $\hat{F}(k)$ as in~\eqref{eq:mgn_qtl}. 
\STATE Compute $k^* = \min\{k\colon \hat{F}(k) \geq 1-\alpha\}$. 
\vspace{0.05in}
\ENSURE Prediction set $\hat{C}(x) = \{y\colon V(x,y)\leq V_{[k^*]}\}$.
\end{algorithmic}
\end{algorithm}

\begin{remark}\label{rm:mgn_opt}
  \normalfont
Writing $W_i = w(X_i,Y_i)$
for $1\leq i\leq \nc$
and $W_{n+1}(y) = w(X_{n+1},y)$,
we can check for~\eqref{eq:qt_wcp} that 
$\hat{V}_{1-\alpha}(y) = V_{[k^*(y)]}$, where 
$k^*(y) = \min\{k\colon F(k,y)\geq 1-\alpha\}$ and 
\#\label{eq:mgn_k*y}
F(k,y) = \frac{\sum^k_{i=1} W_{[i]}}
{\sum^\nc_{i=1}W_i + W_{n+1}(y)} .
\#
For each $k$, the threshold $\hat{F}(k)$ 
in~\eqref{eq:mgn_qtl} is the solution to
 the following optimization problem 
\$
\textrm{minimize} \quad &\frac{\sum^k_{i=1}W_{[i]}  }{\sum^\nc_{i=1}W_i + W_{n+1}} \\
\textrm{subject to}\quad & \hat\ell(X_i) \leq W_i \leq \hat u(X_i),\quad \forall~ i\in \cD_\calib \cup \{n+1\}, 
\$
which seeks a lower bound on the unknown $F(k,y)$  
in~\eqref{eq:mgn_k*y} 
if we believe $\hat\ell(x)\leq w(x,y)\leq \hat u(x)$ for all $(x,y)\in \cX\times\cY$. 
Therefore, $\hat{V}_{[k^*]}$ is 
a conservative estimate (upper bound) of  
$\hat{V}_{1-\alpha}(y)$ in~\eqref{eq:qt_wcp}, and it is also the tightest
upper bound one could obtain given the constraints 
$\hat\ell(X_i) \leq W_i \leq \hat u(X_i)$, $\forall~ i\in \cD_\calib \cup \{n+1\}$. 
\end{remark}

\subsection{Theoretical guarantee}

To state the coverage guarantee 
of Algorithm~\ref{alg:mgn}, 
we start with some notations. 
We write $a_+ = \max(a,0)$ and $a_- = \max(-a,0)$ for any $a \in \RR$.
For any random variable $U$, define $\|U\|_r = (\EE[|U|^r])^{1/r}$
as the $L_r$ norm for any $r\geq 1$, 
and $\|U\|_\infty = \EE[\esssup|U|]$ 
as the $L_\infty$ norm. 
Throughout the paper, 
all the statements are conditional on $\cD_\train$,
so that 
$\hat{\ell}(\cdot)$, $\hat{u}(\cdot)$ 
and $V(\cdot,\cdot)$ can be viewed as fixed functions.


\begin{theorem}\label{thm:mgn}
Assume $(X_i,Y_i)_{i\in \cD_{\textnormal{calib}}} \iid \PP$, 
and the independent test point $(X_{n+1},Y_{n+1})\sim \TPP$ 
has likelihood ratio $w(x,y) = \frac{\ud \TPP}{\ud \PP}(x,y)$.  
Then for any target level $\alpha\in(0,1)$, 
the output of Algorithm~\ref{alg:mgn} satisfies
\$
\TPP\big(Y_{n+1} \in\hat{C}(X_{n+1})\big) 
   \ge 1-\alpha - \hat\Delta, 
\$ 
where the probability is over $\cD_{\textnormal{calib}}$ and $(X_{n+1},Y_{n+1})$, 
and
\#\label{eq:mgn_gap}
\hat\Delta &=\big\|1 / \hat{\ell}(X)   \big\|_{q} \cdot  \bigg(\,\Big\|  \big(\hat{\ell}(X )-w(X ,Y )  \big)_+  \Big\|_{p} +  \Big\|  \big(\hat{u}(X )-w(X ,Y )  \big)_-  \Big\|_{p} \notag 
\\
&\qquad \qquad \qquad \qquad \qquad  
+ \frac{1}{\nc} \Big\|  w(X ,Y )^{1/p}  \cdot \big(\hat{u}(X)-w(X ,Y )  \big)_-  \Big\|_{p} \,\bigg)   . 
\#
Given $q \ge 1$, $p$ is chosen such that $1/p+1/q=1$ with the convention that $p=\infty$ for $q=1$. 
The expectation in $L_p$, $L_q$ is taken over  
an independent sample
$(X,Y)\sim\PP$.

Clearly, if $\hat\ell(X) \le w(X,Y) \le \hat{u}(X)$ a.s., $\hat \Delta
= 0$ and
\$
\TPP\big(Y_{n+1} \in\hat{C}(X_{n+1})\big) 
   \ge 1-\alpha. 
\$  
\end{theorem}

The proof of Theorem~\ref{thm:mgn} is deferred to
Appendix~\ref{app:proof_thm_mgn}. 
Almost sure bounds on the
likelihood ratio imply \revisea{exact} coverage. 
Otherwise, coverage is off by
at most $\hat\Delta$. 
First, $\|1/\hat{\ell}(X)\|_q$ is bounded when 
$\hat{\ell}(\cdot)$ is bounded away from zero.  
Second, the
remaining terms (between brackets) 
are small  
if $\hat \ell(X)$ (resp.~$\hat u(X)$) is
below (resp.~above) the likelihood ratio most of the time. 
In particular, for all target distributions within~\eqref{eq:id_set},
if $\hat\ell(\cdot)$ and $\hat u(\cdot)$ are estimators for
$\ell(\cdot)$ and $u(\cdot)$, one additive term can be decomposed as
\#\label{eq:mgn_gap_decomp} \Big\| \big(\hat{\ell}(X)-w(X ,Y ) \big)_+
\Big\|_{p} =\Big\| \big\{ \underbrace{\big(\hat{\ell}(X) -
  \ell(X)\big)}_{\textrm{estimation error}} - \underbrace{\big(w(X ,Y)
  - \ell(X)\big)}_{\textrm{population gap}} \big\}_+ \Big\|_{p}, \#
which is small as long as the estimation error does not exceed the
population gap most of the time.  Our numerical experiments
demonstrate that $\hat\Delta$ is reasonably small even with
non-negligible estimation errors.

\reviseb{
\begin{remark}\label{rem:mgn_general}
  \normalfont 
  The miscoverage bound 
in Theorem~\ref{thm:mgn} 
holds for any distributional shift and 
any postulated bounds $\hat\ell(\cdot)$ and $\hat{u}(\cdot)$. 
The results also hold if the bounds take the general form
$\hat\ell(x,y)$ and $\hat u(x,y)$.  Therefore, our procedure may be
of use in a very broad range of settings.  
\end{remark}
 }

\section{Valid counterfactual inference under the sensitivity model}
\label{sec:dist_shift}
 
With the generic methodology in place, we return to
counterfactual inference under unmeasured confounding.  We shall
characterize distributional shifts of interest, 
\reviseb{translate the target distribution into 
a set like~\eqref{eq:id_set}}
and show how Algorithm~\ref{alg:mgn} can be
applied.

Recall that we have a partially revealed and 
possibly confounded dataset 
$(X_i,T_i,Y_i)_{i\in \cD}$ generated by an unknown super population $\PP^\super$. 
Given 
an independent new 
unit from $\PP^\super$  
with only the covariate $X_{n+1}$ observed,  
we would like to construct
a prediction interval that covers $Y_{n+1}(1)$ 
with probability at least $1-\alpha$; the probability is
over the randomness of the training sample  
$\cD$ and the new unit. 
The main challenge is that the  
distribution of $(X_{n+1},Y_{n+1}(1))$ 
may differ from \revise{that of} the 
observations we have access to; that is, 
the training distribution is $\PP_{\textrm{train}} = \PP_{X,Y(1) \given T=1}$, 
while the target distribution is $\PP_{\textrm{target}} = \PP_{X,Y(1)}$.

\subsection{Bounding the distributional shift} 
\label{subsec:dist_shift}

In the previous example, the likelihood ratio takes the form
  $w(x,y) = \frac{\ud \PP_{X,Y(1) }}{\ud \PP_{X,Y(1) \given
      T=1}}(x,y)$.   A key observation in~\citet{lei2020conformal} is
that, under the strong ignorability condition~\eqref{eq:ign}, $w(x,y)$
is an identifiable function of $x$, i.e., the target distribution can
be identified from the observed (training) distribution with a
covariate shift. This fact is used to construct prediction intervals
for counterfactuals with finite-sample guarantees by leveraging the
weighted conformal inference procedure.

In the presence of 
unmeasured confounding, $w(x,y)$ can no longer 
be expressed as a function of $x$. However, 
under the marginal $\Gamma$-selection condition~\eqref{eq:gamma_sel_mgn}, 
$w(x,y)$ can be bounded from above and below by functions of 
$x$; 
\reviseb{that is, the unknown target distribution 
falls within a set of the form~\eqref{eq:id_set}}. 
The following lemma is a key ingredient
for establishing the boundedness result. 

\begin{lemma}\label{lem:shift}
Suppose a distribution $\PP$ over $(X,U,T,Y(0),Y(1))$ 
satisfies the marginal $\Gamma$-selection condition. 
Then for any $t\in \{0,1\}$, it holds for $\PP$-almost all $x\in\cX$, $y\in \cY$
that       
\$
\frac{1}{\Gamma} \le \frac{\ud \PP_{Y(t)\given X,T=t}}
{\ud \PP_{Y(t)\given X,T=1-t}}(x,y) \le \Gamma. 
\$
\end{lemma}

The proof of Lemma~\ref{lem:shift} is in 
Appendix~\ref{supp:pf_lem_shift}. Returning to $w(x,y)$,
by Bayes' rule we have
\#\label{eq:bayes}
\frac{\ud \PP_{X,Y(1) }}{\ud \PP_{X,Y(1)\given T=1}}
= \PP(T=1) \cdot 
\bigg(1 + \frac{\ud \PP_{Y(1) \given X, T=1}}
{\ud \PP_{Y(1) \given X, T=0}}
\cdot \frac{1-e(X)}{e(X)}\bigg).
\#
Applying Lemma~\ref{lem:shift} to~\eqref{eq:bayes},
we obtain 
\$
\PP(T=1) \cdot \bigg(1 + \frac{1}{\Gamma} \cdot \frac{1-e(X)}{e(X)}\bigg)
\le \frac{\ud \PP_{X,Y(1)}}{\ud \PP_{X,Y(1) \given T=1}}
\le \PP(T=1) \cdot \bigg(1 + \Gamma \cdot \frac{1-e(X)}{e(X)}\bigg).
\$
We have thus bounded the likelihood ratio $w(x,y)$ 
by functions of the covariate $x$.  

The above reasoning 
broadly applies to other types of inferential 
targets: 
we might consider {\em average treatment effect on the 
treated} (ATT)-type inference on $Y(1)$; that is, 
to construct a prediction interval such that
$\PP\big(Y_{n+1}(1) \in \hat{C}(X_{n+1}) \given T=1 \big)\ge 1-\alpha$.
In this case, 
the likelihood ratio is simply
$w(x,y) = \frac{\ud \PP_{X,Y(1) \given T=1}}{\ud \PP_{X,Y(1)\given T=1}}(x,y) = 1$.
Alternatively, 
we might be interested in the {\em average treatment effect on the 
control} (ATC)-type inference on $Y(1)$;
that is, 
to construct $\hat{C}(X_{n+1})$ such that
$\PP\big(Y_{n+1}(1) \in \hat{C}(X_{n+1}) \given T=0 \big)\ge 1-\alpha$. 
In this case, 
the likelihood ratio is 
$\frac{\ud \PP_{X,Y(1) \given T=0}}{\ud \PP_{X,Y(1) \given T=1}}$ and 
the lower and upper bounds are 
\$
\ell(x) = \frac{\PP(T=1)}{\PP(T=0)} \cdot \frac{1}{\Gamma} \cdot \frac{1-e(x)}{e(x)},
\qquad u(x) = \frac{\PP(T=1)}{\PP(T=0)}\cdot \Gamma \cdot \frac{1-e(x)}{e(x)}.
\$
More generally, one might also wish to conduct inference on 
a different population, i.e., the target distribution
admits a different distribution of covariates 
whereas
the joint distribution of $(Y(0),Y(1),U,T)$ given $X$ 
stays invariant. To be specific, we assume
\$
\mbox{training distribution: } \PP_{X,U,T,Y(0),Y(1)} &= \PP_X \times \PP_{U,T,Y(0),Y(1)\given X},\\
\mbox{target distribution: } \PP_{X,U,T,Y(0),Y(1)} &= \QQ_X \times \PP_{U,T,Y(0),Y(1)\given X}.
\$
The corresponding upper and lower bounds on the likelihood ratio are
\$
 \ell(x) = \PP(T=1) \cdot \frac{\ud \QQ_{X}}{\ud \PP_X }(x)
 \cdot \Big(1+ \frac{1}{\Gamma} \cdot \frac{1-e(x)}{e(x)}\Big),
 \qquad 
 u(x) = \PP(T=1) \cdot \frac{\ud \QQ_{X}}{\ud \PP_X }(x)
 \cdot \Big(1+ \Gamma \cdot \frac{1-e(x)}{e(x)}\Big).
\$ 
The aforementioned types of inferential target can 
also be applied to $Y(0)$ following the same arguments. 
We summarize the 
bounds for various inferential targets in 
Table~\ref{tab:summary_up_lo},
which recovers Table 1 in \cite{lei2020conformal} when $\Gamma=1$.

\begin{table}[ht]\renewcommand{\arraystretch}{1.7}
\centering
\begin{tabular}{c|c|cccc}
    \toprule
    Counterfactual & Bound & ATE-type & ATT-type & ATC-type & General \\
    \midrule
    \multirow{2}{*}{$Y(1)$} & $\ell(x)$ & $p_1 \cdot \big(1+ \frac{1}{ \Gamma\cdot r(x)}\big)$ & $1$ & 
    $\frac{p_1}{p_0} \cdot \big(\frac{1}{\Gamma \cdot r(x)}\big)$  & $p_1 \cdot \frac{\ud \QQ_X}{\ud \PP_X}(x)\cdot\big(1+ \frac{1}{\Gamma \cdot r(x)}  \big)$\\
    & $u(x)$ & $p_1 \cdot \big(1+  \frac{ \Gamma }{r(x)}\big)$ & 1 & 
    $ \frac{p_1}{p_0} \cdot  \frac{ \Gamma }{r(x)} $  & $p_1\cdot \frac{\ud \QQ_X}{\ud \PP_X}(x)\cdot\big(1+ \frac{\Gamma}{r(x)}\big)$\\
    \multirow{2}{*}{$Y(0)$} & $\ell(x)$ & $p_0 \cdot \big( 1+ \frac{r(x)}{\Gamma} \big)$ & $\frac{p_0 }{p_1 } \cdot \frac{r(x)}{\Gamma}$  & 
    $1$ & $p_0 \cdot \frac{\ud \QQ_X}{\ud \PP_X}(x)\cdot\big(1+ \frac{r(x)}{ \Gamma} \big)$\\
        & $u(x)$ & $p_0 \cdot \big(1+ \Gamma \cdot r(x) \big) $ & $\frac{p_0}{p_1} \cdot  \Gamma \cdot r(x) $  & 
    1 & $p_0 \cdot \frac{\ud \QQ_X}{\ud \PP_X}(x)\cdot\big(1+ \Gamma \cdot r(x)\big)$\\
    \bottomrule
\end{tabular}
\caption{Summary of the upper and lower bounds of the 
likelihood ratio for different inferential targets.
For $t\in\{0,1\}$, $p_t = \PP(T=t)$ and $r(x)=e(x)/(1-e(x))$ 
is the odds ratio of the propensity score. The training distribution
for  $Y(t)$ is always $\PP_{X,Y(t)\given T=t}$. 
For target distributions, 
ATE-type refers to  $\PP_{X,Y(t)}$;
ATT-type refers to  $\PP_{X,Y(t)\given T=1}$;
ATC-type refers to  $\PP_{X,Y(t)\given T=0}$;
General refers to  $\QQ_X \times \PP_{Y(t)\given X}$. 
}
\label{tab:summary_up_lo}
\end{table}

\subsection{Robust counterfactual inference}

To apply Algorithm~\ref{alg:mgn} to counterfactual prediction 
with a pre-specified confounding level $\Gamma$, 
we take the upper and lower bounds as presented 
in Table~\ref{tab:summary_up_lo}.  
For instance, 
for ATE-type predictive inference for $Y(1)$, 
the likehood ratio $w(x,y)$ is bounded by $\ell(x)$
and $u(x)$ which take the form:
\$
\ell(x) = p_1 \cdot \bigg(1+\frac{1-e(x)}{\Gamma \cdot e(x)}\bigg),\quad 
u(x) = p_1 \cdot \bigg(1+\Gamma\cdot \frac{1-e(x)}{e(x)}\bigg).
\$
We then construct $\hat\ell(x)$ and $\hat{u}(x)$ 
by plugging in an estimator $\hat{e}(x)$ of $e(x)$ trained on $\cD_\train$.  
Taking $\Gamma=1$, our procedure recovers the approach
of~\cite{lei2020conformal}, where $\hat\ell(x)=\hat u(x)$ and both
equal the likelihood ratio function; our finite-sample bound
$\hat\Delta$ in~\eqref{eq:mgn_gap} reduces to a bound similar to
\citet[Theorem 3]{lei2020conformal} (the forms of the two bounds are
slightly different, hence not directly comparable in general).  

The finite-sample guarantee in 
Theorem~\ref{thm:mgn} shows that 
the accuracy of $\hat e(x)$ itself 
(hence that of $\hat \ell(x)$ and $\hat u(x)$) 
may not matter much for valid coverage;
what matters is whether or not $\hat \ell(x)$ is below $w(x,y)$ 
and $\hat u(x)$ above $w(x,y)$. 
In our simulation studies, we empirically 
evaluate $\|\hat{\ell}(X) - \ell(X)\|_1$, $\|\hat{u}(X) - u(x)\|_1$ and $\hat \Delta$ 
(see Figure~\ref{fig:mgn_gap}). 
Even if  
$\|\hat{\ell}(X)- \ell(X)\|_1$ and $\|\hat u(X) - u(X)\|_1$ 
can be large (especially for large $\Gamma$), 
the gap 
$\hat \Delta$ remains reasonably small. 

\subsection{Numerical experiments}\label{subsec:simu_mgn}

We illustrate the performance of the 
novel procedure in
a simulation setting  
similar to that in~\cite{yadlowsky2018bounds}.  
Given a sample size $n_\train = n_\calib \in \{500, 2000,5000\}$
and a covariate dimension $p\in\{4,20\}$, we generate the covariates 
and unobserved confounders with 
\$
X\sim \textrm{Unif}[0,1]^p,\quad U\given X \sim N\Big(0,1+\frac{1}{2} \cdot (2.5X_1)^2\Big).
\$
The counterfactual of interest is $Y(1)$ generated as 
\$
Y(1) = \beta^\top X + U,\quad \textrm{where}~~ \beta = (-0.531, 0.126, -0.312, 0.018,0,\dots,0)^\top \in \RR^p.
\$ 
In other words, the fluctuation of $Y(1)$ from 
its conditional mean is entirely driven by $U$. 
With i.i.d.~data $(X_i,Y_i,U_i)$ 
generated from the fixed super-population, 
treatment assignments $T_i$ are generated from 
treatment mechanisms satisfying~\eqref{eq:gamma_sel_mgn}
with different confounding levels $\Gamma\in\{1,1.5,2,2.5,3,5\}$.  
Specifically, we design the propensity scores as
\#\label{eq:simu_prop}
e(x) = \text{logit}(\beta^\top x),\quad e(x,u) = a(x) \ind\{|u|>t(x)\} + b(x) \ind\{|u|\leq t(x)\},
\#
for the same $\beta\in \RR^p$. Above,
\$
a(x) = \frac{e(x)}{e(x) + \Gamma(1-e(x))},\quad b(x) = \frac{e(x)}{e(x) +  (1-e(x))/\Gamma}
\$ 
are the lower and upper bounds on $e(x,u)$ under the marginal $\Gamma$-selection model. 
The threshold $t(x)$ is designed to ensure $\EE[e(X,U)\given X] = e(X)$.  
The training data $\cD$ are $\{(X_i,Y_i(1)\}$ for those $T_i=1$. 
By~\eqref{eq:simu_prop}, the setting is designed to be adversarial 
so as to show the performance of our method in a nearly worst case. 
 
For each configuration of $(n_\calib,p,\Gamma,\alpha)$, 
we run Algorithm~\ref{alg:mgn}  
with ground truth $\ell(\cdot)$, $u(\cdot)$ 
and with estimated $\hat{\ell}(\cdot)$, $\hat{u}(\cdot)$. 
To obtain estimated bound functions, 
we fit the propensity score $\hat{e}(x)$ on $\cD_\train$ 
using regression forests from the \texttt{grf} R-package,
and set $\hat{\ell}(x) = \hat{p}_1 (1+ (1-\hat{e}(x))/(\Gamma\cdot \hat{e}(x)))$,
$\hat{u}(x) = \hat{p}_1 (1+\Gamma\cdot (1-\hat{e}(x))/\hat{e}(x))$
with $\hat{p}_1$ being the empirical proportion of $T_i=1$.  
We compute the average coverage on one test sample 
over $N=1000$ independent runs.

The proposed approaches work for any nonconformity score function
and any training method. 
In our experiments, we follow the  
Conformalized Quantile Regression algorithm (CQR)~\citep{romano2019conformalized}
to compute the nonconformity score: 
$\cD_\train$ is employed to train a 
conditional quantile function $\hat{q}(x,\beta)$
for $Y(1)$ conditional on $X$ 
by quantile random forests~\citep{meinshausen2006quantile}.
Then for \revise{a target level} $\alpha\in(0,1)$, the 
nonconformity score is defined as
\$
V(x,y) = \max \big\{ \hat{q}\,(x,\alpha/2) - y, ~y - \hat{q}\,(x,1-\alpha/2)\big\}.
\$
We run the procedures for a target $\alpha \in \{0.1,0.2,\dots,0.9\}$
with the corresponding nonconformity scores. 

The empirical coverage 
when $\hat \ell(\cdot)$ and $\hat u(\cdot)$ are estimated 
is summarized in Figure~\ref{fig:mgn_cov_est},
with its counterpart with ground truth in Figure~\ref{fig:mgn_cov_known}
in Appendix~\ref{app:simu_predict} showing quite similar performance.

\begin{figure}[h]
\centering 
\includegraphics[width=5.5in]{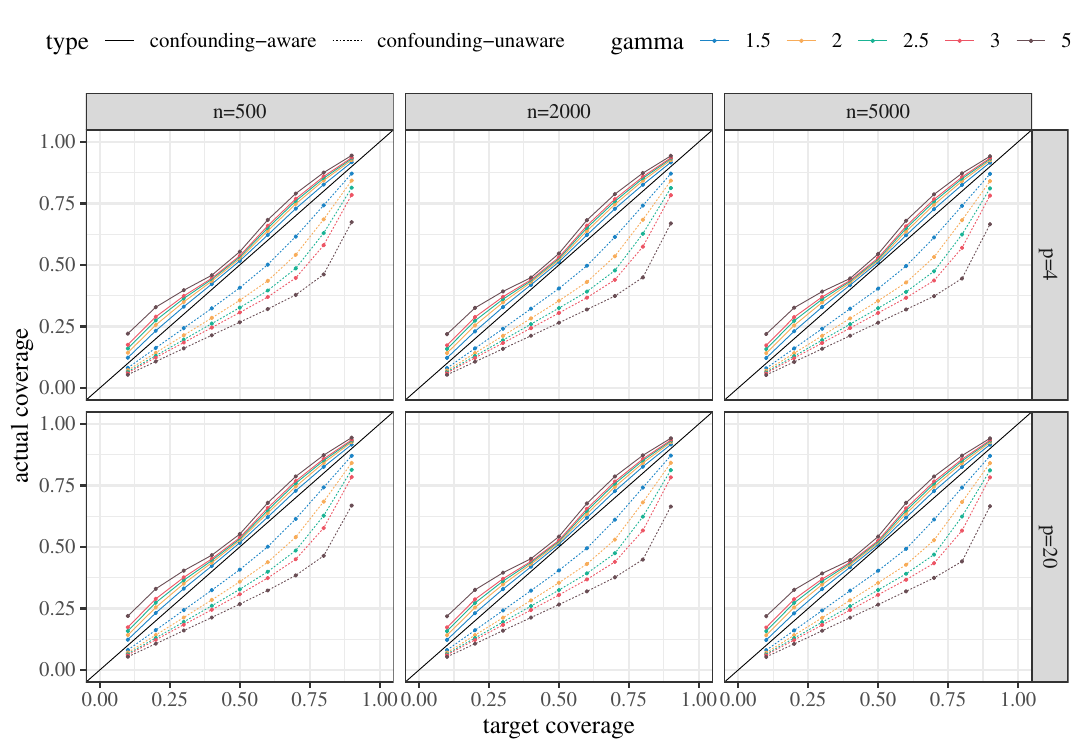}
\caption{Empirical coverage of Algorithm~\ref{alg:mgn} when $\hat \ell(\cdot)$
and $\hat u(\cdot)$ are estimated. Each column corresponds to a sample size $n = n_\calib$,
while each row corresponds to a dimension $p$. Within each subplot, 
each line corresponds to a confounding level $\Gamma$. 
The solid lines are for Algorithm~\ref{alg:mgn}.
\revise{The dashed lines assume no confounding 
and are shown for comparison.} 
}
\label{fig:mgn_cov_est}
\end{figure}

\paragraph{Validity} 
In Figure~\ref{fig:mgn_cov_est}, 
the solid lines are all above the $45^\circ$-line, 
showing the validity of the proposed procedure 
over the whole spectrum of sample sizes. We also see that confounding
must be taken into consideration to reach valid counterfactual
conclusions.

\paragraph{Sharpness}
In all the configurations, especially 
when the target coverage is close to $0.5$, 
the actual coverage is quite close to the target, 
which shows the sharpness of the prediction interval 
in this setting.

\paragraph{Robustness to estimation error} 
As illustrated in the decomposition~\eqref{eq:mgn_gap_decomp},
the gap between $\ell(x)$ and $w(x,y)$ (as well as  
that between $w(x,y)$ and $u(x)$) 
provides some buffer for the estimation error
in $\hat{\ell}(\cdot)$ and $\hat{u}(\cdot)$. 
To be concrete, we numerically evaluate 
the gap $\hat\Delta$
with $(q,p) = (\infty,1)$ 
as well as the estimation error
$\|\hat{\ell}(X)-\ell(X)\|_1$ and $\|\hat{u}(X)-u(X)\|_1$ 
in Figure~\ref{fig:mgn_gap}. 
We see that although the estimation error of the 
lower and upper bounds can sometimes be large, 
\revise{the realized gap $\hat\Delta$} is often very close to zero.

\begin{figure}[htbp]
\centering 
\includegraphics[width=1\linewidth]{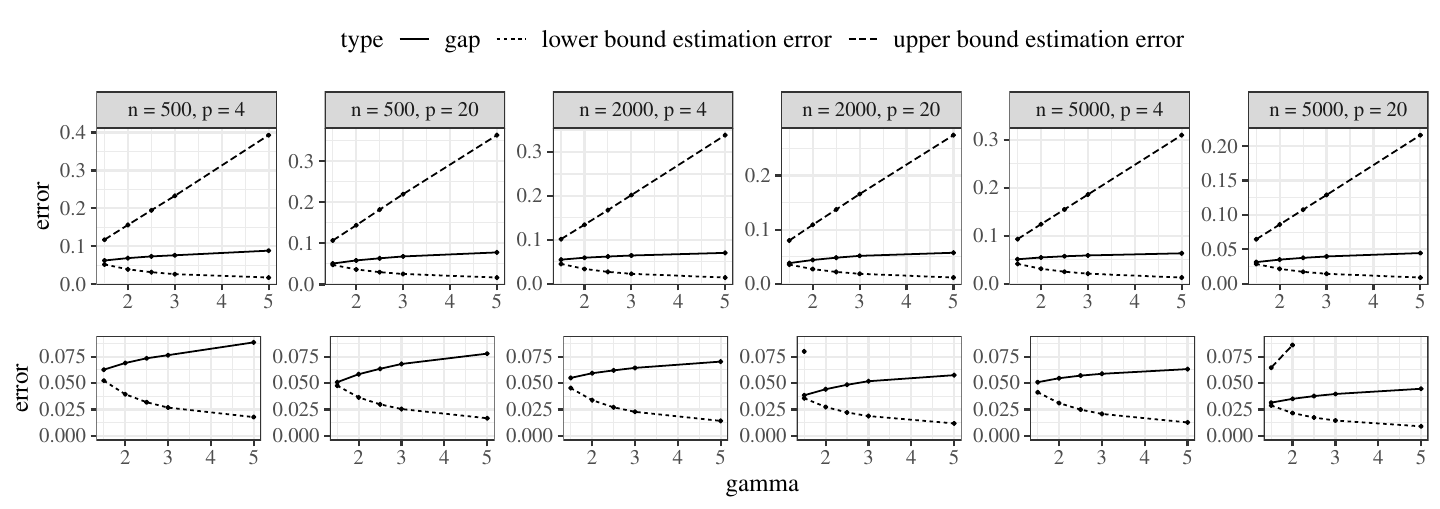}
\caption{Empirical gap and estimation errors. The plots in 
the second row zoom in on the gaps. 
Each plot corresponds to a sample size $n = n_\calib$ and a dimension $p$. 
The long-dashed lines are $\|\hat{u}(X)-u(X)\|_1$,
the short-dashed lines are $\|\hat{\ell}(X)-\ell(X)\|_1$, and
the solid lines are 
$\hat\Delta$ defined in Theorem~\ref{thm:mgn}.}
\label{fig:mgn_gap}
\end{figure}

\section{PAC-type robust conformal inference}
\label{sec:cal_interval}

In this section, 
we construct robust prediction sets  
with guaranteed coverage on the test sample 
conditional on the training data.  
Such guarantee has been considered 
by~\citet{bates2021distribution,bates2021testing} and
might be appealing to the practitioners---it ensures that unless 
one gets really unlucky with the training set $\cD$, 
the anticipated coverage of the prediction set
on the test sample {\em conditional on} $\cD$ achieves the desired
level.

\subsection{The procedure} \label{subsec:whp_procedure}

We state \reviseb{our approach in the generic setting,} starting with
some basic notations.  Throughout, we follow the sample splitting
routine as before, and all statements are conditional on $\cD_\train$.
\reviseb{ Suppose we have a pair of functions $\hat\ell$ and
  $\hat u\colon \cX\to \RR^+$ with $\hat\ell(x)\leq \hat u(x)$ for all
  $x\in \cX$.  Again, we expect them to bound the unknown likelihood
  ratio in a pointwise fashion, although this is not required for
  getting theoretical guarantees. }

For any function $f\colon \cX\times\cY\to \RR$ 
and any target distribution $\TPP$, we define the 
cumulative distribution function (c.d.f.) induced by $\TPP$ as 
\$
F(t\,; f,\tilde\PP) := \tilde\PP\big(f(X_{n+1},Y_{n+1})\leq t\big) 
=\int \ind_{\{f(x,y)\leq t\}}\ud \tilde\PP(x,y),\quad \forall ~t\in \RR. 
\$
Our approach is based on a general \reviseb{non-decreasing 
function  
$G(\cdot)\colon \RR \to [0,1]$.   
We expect $G(\cdot)$ to serve as a {\em conservative 
envelope function} for the 
unknown target distribution 
of the non-conformity score, i.e.,   
\#\label{eq:G_t}
&G(t) \leq 
F(t\,;V,\tilde\PP) \quad \text{for all }t\in \RR.
\#
We now construct $G(\cdot)$ explicitly:
\#\label{eq:def_G_t_formula}
G(t) =   \max\Big\{ \EE\big[\ind_{\{V(X,Y)\leq t\}}\hat{\ell}(X)\big],~
1- \EE\big[ \ind_{\{V(X,Y)>t\}} \hat{u}(X) \big]  \Big\},\quad t\in \RR,
\#
where $\EE$ is taken with respect to an independent copy $(X,Y)\sim \PP$.
By construction, $G(\cdot)$ satisfies~\eqref{eq:G_t} when $\hat \ell(x)$
and $\hat u(x)$ are lower and upper bounds on $w(x,y)$.
}

Given a constant $\delta\in (0,1)$, suppose we
can construct a non-decreasing confidence lower bound $\hat{G}_n(\cdot)$
for $G(\cdot)$ such that 
for any {\em fixed} $t\in \RR$ 
or random variable $t\in \sigma(\cD_\train)$, it holds that
\#\label{eq:hatG_t}
\PP_{\cD }\big( \hat{G}_n(t) \leq G(t)  \big) \geq  1- \delta,
\#
where $\PP_{\cD }$ is taken with respect to $\cD_{\calib}$. 
We then define the prediction interval as 
\$ 
\hat{C}(X_{n+1}) = \Big\{y: V(X_{n+1},y) 
\leq \inf\big\{ t: \hat{G}_n(t) \geq 1-\alpha \big\}\Big\}.
\$
The procedure is summarized in Algorithm~\ref{alg:whp}.

\begin{algorithm}[H]
\caption{Robust conformal prediction: the PAC procedure}\label{alg:whp}
\begin{algorithmic}[1]
\REQUIRE Calibration data $\cD_\calib$, bounds $\hat\ell(\cdot)$, $\hat{u}(\cdot)$, 
non-conformity score function $V\colon \cX\times \cY\to \RR$, test covariate $x$, 
target level $\alpha\in(0,1)$, confidence level $\delta\in(0,1)$.
\vspace{0.05in} 
\STATE Construct the conservative envelope
distribution function $\hat{G}_n(t)$ for $t\in \RR$.
\STATE Compute $\hat{v} = \inf\{t\colon \hat{G}_n(t) \geq 1-\alpha\}$. 
\vspace{0.05in}
\ENSURE Prediction set $\hat{C}(x) = \{y\colon V(x,y)\leq \hat{v}\}$.
\end{algorithmic}
\end{algorithm}

\paragraph{Construction of estimators}
The construction of $\hat{G}_n(\cdot)$ can be flexible, since we only require~\eqref{eq:hatG_t}
to hold for any fixed $t\in \RR$ instead of holding uniformly. 
For example, assuming $\|\hat\ell\|_\infty, \|\hat{u}\|_\infty \leq M$ for some constant $M>0$, 
for any $\delta\in (0,1)$, we may set
\#\label{eq:hat_G_hoeffding}
\hat{G}_n(t) = \max\bigg\{ \frac{1}{\nc}\sum_{i=1}^\nc \ind_{\{V_i\leq t\}}\hat{\ell}(X_i)  ,~ 1-  \frac{1}{\nc}\sum_{i=1}^\nc \ind_{\{V_i> t\}}\hat{u}(X_i)   \bigg\}- M\sqrt{\frac{\log(2/\delta)}{2\nc}}.
\#
Hoeffding's inequality then ensures 
$\PP_{\cD}(\hat{G}_n(t) \leq G(t))\geq 1-\delta$ for any fixed $t\in \RR$. 
One may also construct $G_n(\cdot)$ by the 
Waudby-Smith--Ramdas bound \citep{waudbysmith2021estimating}; 
details are deferred to
Proposition~\ref{prop:wsr_bound} in Appendix~\ref{app:pac}. 
When $\hat\ell(\cdot)$ and $\hat u(\cdot)$
are obtained from $\cD_\train$, the upper bounds on $\|\hat\ell\|_\infty$ and $\|\hat u\|_\infty$
could be 
obtained from the training process.

\subsection{Theoretical guarantee}
\label{subsec:theory_whp}

\begin{theorem}\label{thm:whp}
  Assume that $(X_i,Y_i)_{i\in \cD_{\textnormal{\calib}}} \iid \PP$
  and the independent test point $(X_{n+1},Y_{n+1})\sim \TPP$ has
  likelihood ratio $w(x,y) = \frac{\ud \TPP}{\ud \PP}(x,y)$.  Fix a
  target level $\alpha\in(0,1)$ and a confidence level
  $\delta\in(0,1)$.  Suppose $\hat{G}_n(\cdot)$
  satisfies~\eqref{eq:hatG_t} for \reviseb{$G(\cdot)$
    in~\eqref{eq:def_G_t_formula}.}  Then the output of
  Algorithm~\ref{alg:whp} satisfies
  \#\label{eq:ci_whp_plugin_coverage}
\tilde\PP\big( Y_{n+1}\in \hat{C}(X_{n+1})\biggiven \cD_{\textnormal{\calib}}\big) 
\geq 1-\alpha-\hat\Delta
\#
with probability at least $1-\delta$ over $\cD_{\textnormal{\calib}}$, 
and  
\#\label{eq:whp_gap}
\hat\Delta = \max \Big\{ \EE \big[\big(\, \hat{\ell}(X)-w(X,Y)\big)_+ \big] ,
~ \EE\big[\big( \hat{u}(X)-w(X,Y)\big)_{-} \big]     \Big\}.
\#
Here the expectations are over an independent copy $(X,Y)\sim \PP$. 

If $\hat\ell(X)\leq w(X,Y)\leq \hat{u}(X)$ a.s., then $\hat\Delta=0$ and 
\$
\tilde\PP\big( Y_{n+1}\in \hat{C}(X_{n+1})\biggiven \cD_{\textnormal{\calib}}\big) 
\geq 1-\alpha.
\$
\end{theorem}
The proof of Theorem~\ref{thm:whp} is deferred to
Appendix~\ref{app:proof_thm_whp}.  
Almost sure bounds on the likelihood ratio yields exact
coverage. Otherwise, coverage is off by at most
$\hat \Delta$.  
    Again, $\hat \Delta$ is small if
  $\hat\ell(X)$ (resp. $\hat{u}(X)$) is below (resp. above) $w(X,Y)$
  most of the time. This is demonstrated in the simulation results
  presented in Figure~\ref{fig:whp_gap}. 

  \reviseb{Just as before, the miscoverage bound in
    Theorem~\ref{thm:whp} holds for any distributional shift and any
    inputs $\hat\ell(\cdot)$ and $\hat u(\cdot)$.  The conclusion also
    applies to inputs of the form $\hat\ell(x,y)$ and $\hat{u}(x,y)$
    without modification.
}

\paragraph{Application to counterfactual prediction}
To apply Algorithm~\ref{alg:whp} 
to counterfactual prediction 
under a pre-specified confounding level $\Gamma$, 
We  
plug in the estimated $\hat e(x)$ to obtain
$\hat \ell(x) $ and $\hat u (x)$ 
according 
to Table~\ref{tab:summary_up_lo}.  As before, $\hat \Delta$ can be
small as long as $\hat \ell(x)$ (resp. $u(x)$) falls below
(resp. above) $w(x)$ most of the time, even if $\hat e(x)$ has a
non-negligible estimator error.

\subsection{Sharpness}\label{subsec:sharp}
Besides validity, 
one may also be concerned with the sharpness of the 
method---indeed, one can always construct a valid but arbitrarily 
conservative prediction interval. 
Ideally, we desire a valid prediction set whose 
coverage is not 
too much larger than the prescribed level. 

In this section, we take a close look at the sharpness of 
Algorithm~\ref{alg:whp} by identifying 
the worst-case distributional shift. 
We study the sharpness of our method in two problems:
robust predictive inference 
and counterfactual inference 
under unmeasured confounding. 
They are treated in the same way when 
developing the validity results, 
but they actually have 
distinct identification sets and 
sharpness results.

To remove the nuisance in estimation, 
we fix the nonconformity score function 
and consider the 
asymptotic setting where $\hat{G}_n(t) \to G(t)$ and $\hat{\ell}\to \ell$, $\hat{u}\to u$, 
focusing on the generic conservativeness of our method.  
In this regime, the prediction set is  
$\hat{C}(X_{n+1})=\{y\colon V(X_{n+1},y)\leq \hat{v}\}$, where 
$\hat{v} = \inf\{t\colon G(t)\geq 1-\alpha\}$. 
When $(X_{n+1},Y_{n+1})\sim \tilde\PP$, 
the coverage of $\hat{C}(X_{n+1})$ is 
\$ 
\tilde\PP\big(Y_{n+1}\in \hat{C}(X_{n+1})\biggiven \cD\big) 
=  
\tilde\PP\big( V(X_{n+1},Y_{n+1}) \leq \hat{v} \biggiven \cD   \big) 
=  
F(\hat{v}\,;\,V,\tilde\PP).
\$ 
Therefore, given some identification set $\cP$ of distributions, 
the sharpness of $G(\cdot)$ relies on  
its difference from 
the worst-case c.d.f.  defined as
\#\label{eq:whp_worst_cdf}
F(t\,;\cP):= \inf_{\TPP\in \cP } F(t \,;V,\tilde\PP) = 
\inf_{\TPP\in \cP }~ \tilde\PP\big(V(X,Y)\leq t\big)
= \inf_{\TPP\in \cP} \EE\bigg[ \ind_{\{V(X,Y)\leq t\}}\frac{\ud \TPP}{\ud \PP}(X,Y)\bigg],
\quad t\in \RR.
\#
Here the last expectation is taken with respect to an independent copy $(X,Y)\sim \PP$. 
\revise{The closer $G(t)$ to 
the worst-case c.d.f.~\eqref{eq:whp_worst_cdf}, 
the closer $\hat{v}$ to 
$\inf\{t\colon F(t\,;\cP)\geq 1-\alpha\}$  
(the worst-case quantile), 
hence the sharper the procedure.}
We provide exact worst-case characterizations of~\eqref{eq:whp_worst_cdf}
in the two problems. 

\subsubsection{Sharpness as a robust prediction problem}
\label{sec:sharpness}
In the more general robust prediction problem, 
the identification set is 
\#\label{eq:id_marginal}
\cP(\PP,\ell,u) = \Big\{~ \tilde\PP\colon \ell(x) \leq \frac{\ud \TPP}{\ud \PP}(x,y)\leq u(x)~~ \PP\textrm{-almost surely}  \Big\}.
\#  

\begin{prop}\label{prop:tight_mgn}
For each $t\in \RR$, the worst-case distribution function in~\eqref{eq:id_marginal} is 
\#\label{eq:tight_mgn}
F\big(t\,;\cP(\PP,\ell,u) \big) = 
\max\Big\{ \EE\big[\ind_{\{V(X,Y)\leq t\}} {\ell}(X) \big],~ 
1- \EE\big[ \ind_{\{V(X,Y)>t\}}  {u}(X)  \big]  \Big\},
\#
where $\EE$ is the expectation over $(X,Y)\sim \PP$.
\end{prop}
 
The conservative distribution function $G(\cdot)$ constructed in~\eqref{eq:G_t}
coincides with the actual worst-case distribution function provided in~\eqref{eq:tight_mgn}.
Therefore, ruling out the estimation errors, 
the PAC-type procedure proposed in Section~\ref{sec:cal_interval} 
is sharp.

\subsubsection{Sharpness as a counterfactual inference problem}
\label{sec:sharpness_counterfactual}
Returning to the counterfactual inference problem,
the sharpness is a bit more subtle than 
in the robust prection formulation. 
The key difference is that the covariate shift 
is identifiable while only the shift in $\PP_{Y\given X}$ varies,  
leading to a potentially smaller identification set.

For clarity, we denote the unknown super-population 
as $(X,Y(0),Y(1),U,T)\sim \PP^\super$ and 
the observed distribution as $(X,Y,T)\sim\PP^\obs$. 
For a super-population $\PP^{\sup}$ to be meaningful, 
it should agree with $\PP^\obs$
on the observable: 
\#\label{eq:data_compat}
\PP_{X,Y,T}^{\textnormal{\super}} = \PP^\obs_{X,Y,T},
\# 
which is called the \emph{data-compatibility} condition in~\cite{dorn2021sharp}.
Along with the marginal sensitivity model,~\eqref{eq:data_compat}
characterizes the set of meaningful target distributions. 
Let us consider the counterfactual inference 
of $Y(1)$ for units in the control group. 
Letting $\PP = \PP_{X,Y(1)} = \PP^\obs_{X,Y(1)\given T=1}$
be the training distribution,  
we have the following sharp characterization of the identification set.  

\begin{prop}[Identification set]\label{prop:id_causal}
For the counterfactual inference of $Y(1)\given T=0$
under confounding level $\Gamma\geq 1$, 
we define the data-compatible  identification set as
\#\label{eq:data_comp_idset}
\cP  = \Big\{  \TPP = \PP^\super_{X,Y(1)\given T=0}\colon  ~ 
\PP^\super\text{ satisfies }\eqref{eq:gamma_sel_mgn}
\text{ and }\eqref{eq:data_compat} \Big\}.
\#
Then we have $\cP = \cP(\PP, f,\ell_0,u_0)$, where 
\#\label{eq:id_set_causal}
\cP(\PP, f,\ell_0,u_0) = \bigg\{  \TPP \colon \frac{\ud\TPP_X}{\ud \PP_X}(x) = f(x),~ \ell_0(x)\leq\frac{\ud\TPP_{Y(1)\given X}}{\ud \PP_{Y(1)\given X}}(y\given x) \leq u_0(x)~~ \PP\textrm{-almost surely}   \bigg\}.
\#
Writing $e(x) = \PP^\obs(T=1\given X=x)$, $p_0 = \PP^\obs(T=0)$
and $p_1 = \PP^\obs(T=1)$, 
we have $\ell_0(x) = 1/\Gamma $, $u_0(x) = \Gamma$, and 
$
f(x) = p_1(1-e(x))/[p_0 \cdot e(x)]
$.
\end{prop}
     
Employing similar arguments, 
each inferential target mentioned in Section~\ref{subsec:dist_shift}
corresponds to an identification set 
similar to~\eqref{eq:id_set_causal}, 
where the $X$-likelihood ratio is identifiable from the training distribution, 
and the conditional likelihood
ratio can be bounded by identifiable functions $\ell_0$ and $u_0$.

In light of Proposition~\ref{prop:id_causal},
we are interested in $F(\cdot\,;\cP(\PP,f,\ell_0,u_0))$ 
with general functions $f,\ell_0,u_0$.  
For any $x\in \cX$ and any $\beta\in[0,1]$, 
we denote the $\beta$-conditional quantile function of $\PP = \PP^\obs_{X,Y(1)\given T=1}$ (up to a.s.~equivalence) as 
\$
q(\beta\,;x,\PP) = \inf\big\{z\colon \PP(Y\leq z\given X=x) \geq \beta \big\}.
\$ 

\begin{prop}\label{prop:tight_causal}
For each $t\in \RR$,  the worst-case distribution function in~\eqref{eq:id_set_causal} is 
\#\label{eq:causal_worst}
F\big(t\,;\cP(\PP,f,\ell_0,u_0)\big) = \EE\big[  \ind_{\{V(X,Y)\leq t\}} w^*(X,Y)   \big],
\#
where the expectation is with respect to generic random variables $(X,Y)\sim\PP$, and 
\#\label{eq:causal_sharp_w*}
w^*(x,y) = f(x) \cdot \big[\ell_0(x) \ind_{\{V(x,y)< q (\tau(x)\,;x,\PP )\}} + 
\gamma_0(x) \ind_{\{V(x,y)=q(\tau(x)\,;x,\PP)\}} +
u_0(x) \ind_{\{V(x,y)>q(\tau(x)\,;x,\PP)} \big].
\#
Here $\tau(x) = (u_0(x)-1)/(u_0(x)-\ell_0(x))$ and $\gamma_0(x)$ is chosen as nonzero when $\PP(V(x,Y)=q(\tau(x)\,;x,\PP)\given X=x)>0$  such that $\EE[w^*(x,Y)\given X=x] = f(x)$  for $\PP$-almost all $x\in \cX$. 
\end{prop}

As indicated by Proposition~\ref{prop:tight_causal}, 
the worst-case likelihood ratio function $w^*(x,y)$ 
is separated into two regions: 
one taking the lower bound $\ell(x)$ and the other taking the upper bound $u(x)$. 
This is similar to~\eqref{eq:id_marginal}
\revise{with the proviso that} 
the boundary $q(\tau(x)\,;x,\PP)$ is more complicated.  

It is possible to attain sharpness with a few more efforts. 
Under the marginal $\Gamma$-selection~\eqref{eq:gamma_sel_mgn}, 
one can check that  
$\tau(x)\equiv \tau$ for some constant $\tau\in(0,1)$. 
Therefore,  
letting $G(t) = F(t\,;\cP(\PP,f,\ell_0,u_0))$ as in~\eqref{eq:causal_worst} 
yields a sharper procedure. 
If  
a tight lower bound $\hat{G}_n(t)$ for $G(t)$ could be constructed, 
the worst-case coverage would equal $1-\alpha$ asymptotically. 
Such modifications are beyond the scope of our current work.

\subsection{Numerical experiments}\label{subsec:simu_whp}

We apply Algorithm~\ref{alg:whp} to 
the same settings as of Section~\ref{subsec:simu_mgn} 
with a fixed confidence level $\delta=0.05$
and the  W-S--R bound as detailed in Proposition~\ref{prop:wsr_bound}. 
The empirical coverage is evaluated on $10000$ test samples 
for  $N=1000$ independent runs. 
The $0.05$-th quantile of empirical coverage 
for various target level $1-\alpha$ 
are summarized in Figure~\ref{fig:whp_qt_est}, 
where estimated $\hat{\ell}(\cdot)$ and $\hat{u}(\cdot)$ are 
used. 
Results with 
the ground truth  $\ell(\cdot)$ and $u(\cdot)$ are
in Figure~\ref{fig:whp_qt_known} in Appendix~\ref{app:simu_predict_whp}.  
The marginal coverage 
is presented in Figures~\ref{fig:whp_cov_known} and~\ref{fig:whp_cov_est}
in Appendix~\ref{app:simu_predict_whp}.  

\begin{figure}[htbp]
\centering 
\includegraphics[width=5.5in]{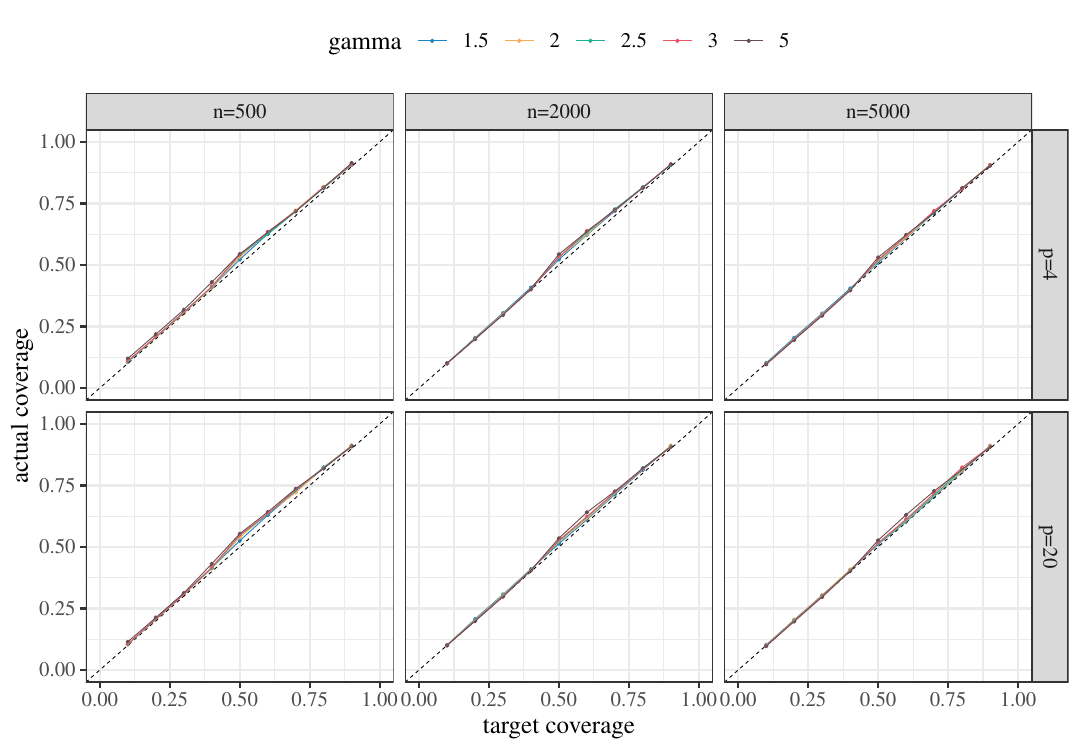}
\caption{$0.05$-th quantile of empirical coverage 
in Algorithm~\ref{alg:whp} when $\hat \ell(\cdot)$
and $\hat u(\cdot)$ are estimated. 
Each subplot corresponds to a configuration of 
sample size $n = n_\calib$
and dimension $p$. Within each subplot, 
each line corresponds to a confounding level $\Gamma$. }
\label{fig:whp_qt_est}
\end{figure}
 
\paragraph{Validity}
In Figure~\ref{fig:whp_qt_est}, 
the $0.05$-th quantiles of empirical coverage   
all lie above the $45^\circ$-line, 
which confirms 
the validity of the proposed procedure.   

\paragraph{Sharpness}
The quantile-versus-target lines in Figure~\ref{fig:whp_qt_est}
almost overlap with the $45^\circ$-line, since by design
the data generating distribution is a near-worst case. 
This shows the sharpness of our method and
is in accordance with the theoretical justification 
in Section~\ref{subsec:sharp}.

\paragraph{Robustness to estimation error}
Similar to Algorithm~\ref{alg:mgn},
the PAC-type algorithm shows robustness to the 
estimation error. 
In Figure~\ref{fig:whp_gap},
we numerically evaluate  the 
gap $\hat\Delta$ in~\eqref{eq:whp_gap} as well as 
the estimation errors 
$\|\hat{\ell}(X)-\ell(X)\|_1$ and $\|\hat{u}(X)-u(X)\|_1$. 
Even though the estimations error (dashed lines)
can be nonnegligibly off, 
the actual gap $\hat\Delta$  defined in Theorem~\ref{thm:whp}
is small.

\begin{figure}[H]
\centering 
\includegraphics[width=1\linewidth]{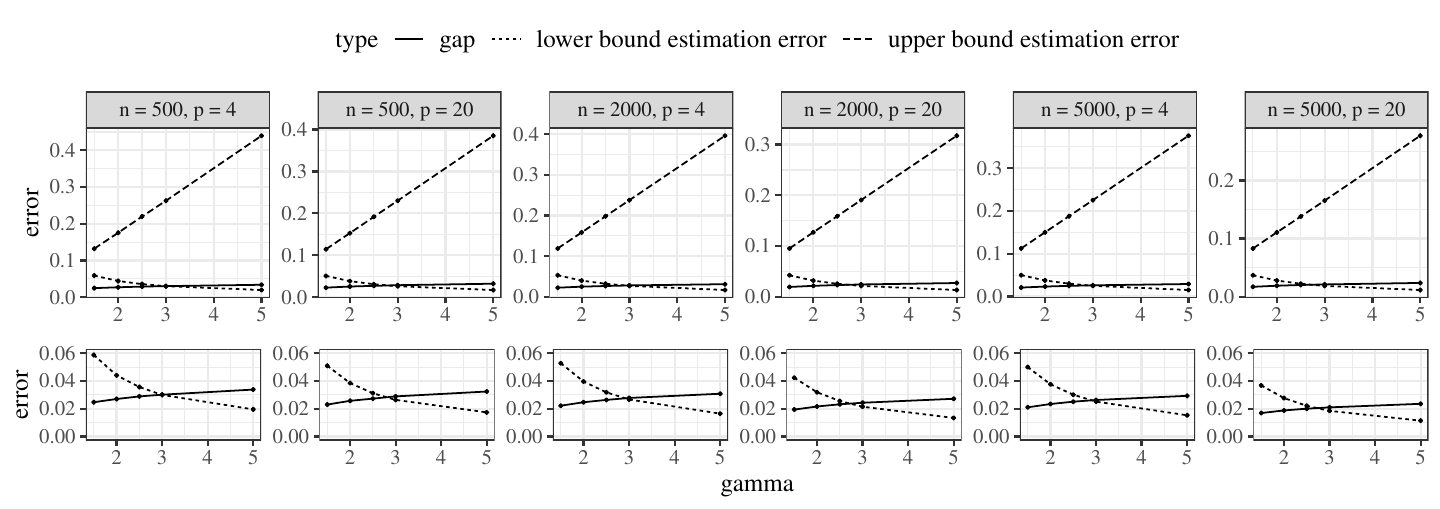}
\caption{Empirical gap and estimation errors. The plots in the 
second row zoom in on the gaps. 
Each plot corresponds to a sample size $n = n_\calib$ and a dimension $p$. 
The long-dashed lines are $\|\hat{u}(X)-u(X)\|_1$,
the short-dashed lines are $\|\hat{\ell}(X)-\ell(X)\|_1$, and
the solid lines are $\hat\Delta$ defined in Theorem~\ref{thm:whp}.}
\label{fig:whp_gap}
\end{figure}

\section{Sensitivity analysis of ITEs}
\label{sec:sens_analysis}
 
We are now ready to present our framework of sensitivity 
analysis for ITEs. 
In the following, we first show 
the construction of prediction sets for ITEs,  
and then invert the prediction sets to 
output the $\Gamma$-values. Finally we 
show the statistical meaning of the 
$\Gamma$-values from  
\revise{a hypothesis testing perspective}.

\subsection{Robust predictive inference for ITEs}
\label{subsec:pred_ite}
We consider  
two cases:  
1) if only one outcome 
is missing,
we use the prediction set 
for the counterfactual 
to form that for the ITE; 
2) if both outcomes  are missing,
we combine prediction sets for
both $Y(1)$ and $Y(0)$ to form that for the ITE. 

\paragraph{One outcome missing}
Given a test sample $X_{n+1}$ with $T_{n+1} = w$, for $w\in \{0,1\}$,
the potential outcome $Y_{n+1}(w)$ is observed while $Y_{n+1}(1-w)$ is missing. 
Using the method introduced in Algorithm~\ref{alg:mgn}
or~\ref{alg:whp}, 
we construct a prediction set 
$\hat{C}_{1-w}(X_{n+1}, \Gamma, 1-\alpha)$
for $Y_{n+1}(1-w)$ with coverage level $1-\alpha$ 
conditional on $T_{n+1} = w$ and the sensitivity parameter $\Gamma$. 
We then create the prediction set for ITE as 
\begin{equation}\label{eq:ite_set}
\hat{C}(X_{n+1},\Gamma ) = 
\begin{cases}
    Y_{n+1}(1) - \hat{C}_0(X_{n+1},\Gamma,1-\alpha),\quad \text{if }w=1,\\
    \hat{C}_1(X_{n+1},\Gamma,1-\alpha) - Y_{n+1}(0),\quad \text{if }w=0.
\end{cases}
\end{equation}

\paragraph{Both outcomes missing}
When both outcomes are missing, the task is more challenging 
and we use a Bonferroni correction to construct the preditive set for ITEs. 
Let $ X_{n+1} $ be a test sample with both $Y_{n+1}(1)$ and $Y_{n+1}(0)$ missing.
Using Algorithms~\ref{alg:mgn} or~\ref{alg:whp}, 
we construct prediction set $\hat{C}_w(X_{n+1},
\Gamma, 1-\alpha/2, \delta/2)$
at confounding level $\Gamma$, 
where $\delta/2$ is the input confidence level if Algorithm~\ref{alg:whp} is used. 
Then we let
\#\label{eq:ite_c_bonferroni}
\hat{C}(X_{n+1}, \Gamma ) = 
\Big\{y-z\colon 
y \in \hat{C}_1(X_{n+1},\Gamma, 1-\alpha/2, \delta/2)~\text{and}~
z \in \hat{C}_0(X_{n+1},\Gamma, 1-\alpha/2, \delta/2)
\Big\}. 
\#
The coverage guarantee of the prediction sets 
in the above two cases 
directly follows from 
the validity of counterfactual prediction intervals; 
Propositions~\ref{prop:ite_mgn} and~\ref{prop:ite_whp} are included
for completeness. 

In practice, though, the Bonferroni
correction~\eqref{eq:ite_c_bonferroni} might be too
conservative.~\citet[Section 4.2]{lei2020conformal} introduced a nested
method that efficiently combines the counterfactual intervals to form
the interval for the ITE; their method can also be applied here.

\subsection{The $\Gamma$-value:  inverting nested prediction sets}

For a new unit, we consider the set of hypotheses
indexed by $\Gamma  \in[1,\infty)$:
\#\label{eq:hypothesis}
H_0(\Gamma) \colon \quad 
Y_{n+1}(1)-Y_{n+1}(0) \in C~~\textrm{and}~~ \PP^\super \in \cP(\Gamma).
\#
If we reject $H_0(\Gamma )$,  
we are saying that either $Y(1)-Y(0)\notin C$ or
the observational data has at least confounding level $\Gamma $.
The pre-specified set $C$ determines the hypothesis one
wishes to test, or equivalently, the causal conclusion one 
wishes to make. For example, setting $C = \{0\}$, we are 
tesing whether or not the ITE is exactly zero; 
setting $C = (-\infty,0]$, we wish to test whether or not
the ITE is negative. 

The hypothesis~\eqref{eq:hypothesis} involves 
the random variable $Y_{n+1}(1)-Y_{n+1}(0)$, 
and there are two ways to treat such hypotheses:
we may either regard it as a \emph{deterministic} 
hypothesis, which means the condition in~\eqref{eq:hypothesis}
holds almost surely, or as a \emph{random} hypothesis such that $H_0(\Gamma)$ 
is true with some probability. In both cases, 
the type-I error is defined as  
rejecting a true hypothesis. 
That is, if we treat them as random hypotheses, 
a
type-I error is  
$H_0(\Gamma)$ being true {\em and}
rejected at the same time.  

Hereafter, for the true super population $\PP^\super$, 
we denote 
\$
\Gamma^*  = \inf \big\{\Gamma\colon \PP^\super \in \cP(\Gamma) \big\}, 
\$
and assume without loss of generality that $\PP^\super \in \cP(\Gamma^*)$. 
Our goal is to test the set of hypotheses $\{\cH_0(\Gamma)\}_{\Gamma \ge 1}$
simultaneously---a multiple testing problem.
Put
\$
\cH_0 = \{\Gamma \colon H_0(\Gamma) \text{ is true}\}.
\$ 
In the case of deterministic hypotheses, 
$\cH_0$ is either an empty set 
(if $Y_{n+1}(1)-Y_{n+1}(0)\in C$ a.s.~is false) 
or an interval $[\Gamma^*,\infty)$ 
(if $Y_{n+1}(1)-Y_{n+1}(0)\in C$ a.s.~is true). 
In the case of random hypotheses, $\cH_0$ is a random set---$\cH_0$
is an empty set if $Y_{n+1}(1)-Y_{n+1}(0)\notin C$, 
or an interval $[\Gamma^*,\infty)$ when $Y_{n+1}(1)-Y_{n+1}(0)\in C$.

Let $\hat{C}(X_{n+1},\Gamma)$ be the prediction set 
for the ITE constructed as in Section~\ref{subsec:pred_ite}
with the confounding level $\Gamma\geq 1$ and the target coverage $1-\alpha$. 
In the case of one missing outcome, $\hat{C}(X_{n+1},\Gamma)$
implicitly depends on the observed outcome as well.  
Note that the prediction sets are nested in $\Gamma$ 
in the following sense: 
for each fixed coverage level $\alpha\in(0,1)$
(and confidence level $\delta$ if necessary), it holds that 
$\hat{C}(X_{n+1},\Gamma )\subset \hat{C} (X_{n+1},\Gamma')$ 
for any $\Gamma'\geq \Gamma \geq 1$. 
Moving on, we define the rejection set as \#\label{eq:rej_set} \cR =
\big\{ \Gamma \colon ~C\cap \hat{C} (X_{n+1}, \Gamma)=\varnothing
\big\}.  \# That is, we reject all $H_0(\Gamma)$ for which
$\hat{C}(X_{n+1},\Gamma)$ does not overlap with the target set $C$.
We now consider the critical value defined as
$ \hat\Gamma \defn \sup \cR, $ with the convention that
$\hat\Gamma= 1$ if $\cR=\varnothing$.
\revise{$\hat\Gamma$ is the formal definition of the
  $\Gamma$-value we introduced in Section~\ref{subsec:intro_gval}.
  The
  $\Gamma$-value is a quantity specific to a unit (instead of a
  population quauntity).  
Due to the variability in ITE, 
it might not converge to a constant value 
as the training sample size goes to infinity. }

\begin{prop}[Simultaneous control]\label{prop:fwer}
Fix a target level $\alpha$ (and a confidence level $\delta$ if necessary). 
For any $\Gamma\geq 1$, 
let $\hat{C} (X_{n+1},\Gamma)$ be the output of Algorithm~\ref{alg:mgn}
or~\ref{alg:whp} with the confounding level $\Gamma$.  
The marginal probability of making a false rejection  
can be controlled as  
\$
\textrm{mErr}\defn \PP ( \cR\cap \cH_0 \neq \varnothing  ) 
\leq \PP \big(  Y_{n+1}(1)-Y_{n+1}(0) \notin \hat{C} (X_{n+1}, \Gamma^*) \big),
\$
where the probability $\PP$ is taken over $\cD_{\calib}$ and the new sample on both sides. 
\revise{Furthermore}, 
the $\cD_{\calib}$-conditional probability
of making an error satisfies
\$
\textrm{dErr}\defn \PP ( \cR\cap \cH_0 \neq \varnothing \given \cD_{\calib} ) 
\leq \PP \big(  Y_{n+1}(1)-Y_{n+1}(0) \notin 
\hat{C} (X_{n+1}, \Gamma^*) \biggiven \cD_{\calib} \big).
\$
\end{prop}

\revise{By Proposition~\ref{prop:fwer},}
as long as the predictive inference achieves valid coverage 
at any {\em fixed} confounding level, 
without 
any adjustment of multiple testing, 
we achieve simultaneous control over the
sequence of testing problems.

\revise{The perspective of hypothesis testing  
provides
an interpretation of the $\Gamma$-value:}
the risk of  $Y_{n+1}(1)-Y_{n+1}(0)\in C$ 
being true but rejected at $\hat\Gamma \geq \Gamma^*$
is (approximately) under $\alpha$. 
In the case of testing deterministic
hypotheses, when $Y(1) - Y(0)\in C$ \revise{almost surely},
$\hat{\Gamma}$ is a $(1-\alpha)$ lower confidence bound for $\Gamma^*$.
It is in accordance with the common practice of sensitivity analysis 
to find a critical value $\hat\Gamma$ that inverts a causal conclusion
and check whether $\hat\Gamma$ is too large to be true 
in order to assess the robustness of such conclusion.

\subsection{Types of null hypotheses}
With the general recipe of assessing robustness of causal conclusions 
on ITE, 
we now provide concrete examples of the target set $C$ and the 
corresponding forms of $\hat{C}(X_{n+1},\Gamma)$. 

\paragraph{Sharp null}
One might be interested in the sharp null, i.e., whether 
the individual treatment effect is zero. 
In this case, one could let $C=\{0\}$, and the prediction set
can take any form. 
Rejecting $H_0(\Gamma)$ is saying that $Y_{n+1}(1)\neq Y_{n+1}(0)$ 
unless the true confounding level satisfies $\Gamma^*>\Gamma$.

\paragraph{Directional null}
If we presume the stochastic nature of ITEs, 
the sharp null might be implausible. 
In this case, one may be interested in the directional null with 
$Y_{n+1}(1)-Y_{n+1}(0)\leq 0$, which is equivalent to 
choosing $C=(-\infty, 0]$. 
Rejecting $H_0(\Gamma)$ here is saying that 
the individual treatment effect is positive unless $\Gamma^*>\Gamma$. 
More generally, one might consider $C=(-\infty, a]$ for some 
$a\in \RR$ to test 
whether the ITE is above a certain value. 
To make sense of the multiple testing procedure, 
one-sided prediction intervals are constructed for ITE, i.e., 
$\hat{C}(X_{n+1},\Gamma) = [\hat{I}(X_{n+1},\Gamma),\infty)$ for 
some $\hat{I}(X_{n+1},\Gamma)\in \RR$. 
It 
can be achieved by one-sided prediction intervals 
$\hat{C}_0(X_{n+1},\Gamma,1-\alpha,1-\delta ) = (-\infty, \hat{Y}_0(X_{n+1},\Gamma)]$
if $Y_{n+1}(0)$ is missing 
and 
$\hat{C}_1(X_{n+1},\Gamma,1-\alpha,1-\delta ) = [\hat{Y}_1(X_{n+1},\Gamma),\infty)$
if $Y_{n+1}(1)$ is missing, 
with a Bonferroni correction as 
introduced in Section~\ref{subsec:pred_ite} 
if both outcomes are missing. 
The $\Gamma$-value is hence the smallest $\Gamma$ such that 
the two prediction intervals overlap.

\subsection{\revisea{Numerical experiments}}
\label{subsec:simu_sens}
 
We focus on the ATT-type inference for the directional null 
\$
H_0(\Gamma)\colon Y(1)-Y(0)\leq 0~~\text{and}~~ \PP^\super\in \cP(\Gamma).
\$
The test sample is from $\PP_{X,Y(0),Y(1)\given T=1}$,
\revise{for which we observe $(X_{n+1},Y_{n+1}(1))$ 
and would like to predict $Y_{n+1}(0)$.} 

We fix $n_\train = n_\calib = 2000$ and $p=4$. 
The covariates $X$, unobserved confounders $U$
and counterfactual $Y(0)$ (instead of $Y(1)$) are generated 
in
the same way as in \revisea{Section~\ref{subsec:simu_mgn}}. 
The treatment mechanism $e(x,u)$ is also 
the same as in~\eqref{eq:simu_prop}
with confounding level $\Gamma \in \{1.2, 1.4, \dots, 2\}$. 
The training data are $(X_i,Y_i(0))$ for those $T_i=0$. 
We generate $Y(1)$ in two ways: 1) $Y (1)-Y (0)\equiv a$ (fixed ITE) and 
2) $Y(1) - Y(0) = a\cdot U$ (random ITE). Here  $a$ ranges in  
$\{-1,-0.5, 0, 0.5, 1\}$. 

Fixing level $\alpha = 0.1$ and $\delta=0.05$ (for Algorithm~\ref{alg:whp}), 
for each fixed  $\Gamma\geq 1$, we construct a
one-sided ATT-type prediction interval for $Y(0)$, 
which takes the form $\hat{C}(X_{n+1},\Gamma) = (-\infty, \hat{Y}(X_{n+1},\Gamma)]$.
The prediction interval is obtained 
via the non-conformity score 
\#\label{eq:score_oneside_cqr}
V(x,y) = y - \hat{q} (x,1-\alpha).
\# 
We reject no hypotheses if $Y_{n+1}(1) \leq \hat{Y}(X_{n+1},1)$;
otherwise, we
reject all $H_0(\Gamma)$ such that $Y_{n+1}(1) > \hat{Y}(X_{n+1},\Gamma)$,  
hence the rejection set is $\cR = [1, \hat\Gamma)$ for some $\hat\Gamma>1$, 
which we define as the $\Gamma$-value. 

\subsubsection{FWER control} 

We evaluate the empirical FWER, which is the proportion of making a false rejection
among $\{H_0(\Gamma)\}_{\Gamma\geq 1}$, 
averaged over all test samples in all $N=1000$ independent runs.  
The results with estimated $\hat{\ell}(\cdot)$ and $\hat{u}(\cdot)$ are presented
in \reviseb{Figures~\ref{fig:simu_ss_fwer} 
and~\ref{fig:simu_ss_fwer_pac}, 
showing control of the FWER 
even when the likelihood ratio bounds are estimated}; 
we omit the case of fixed ITE 
at $a > 0$ since the hypotheses $H_0(\Gamma)$ are always false.
The results with the ground truth $\ell(\cdot)$ and $u(\cdot)$ 
are in Appendix~\ref{app:simu_sens}. 


\begin{figure}[htbp]
\centering
\begin{subfigure}[t]{0.46\linewidth}
    \centering
    \includegraphics[width=0.8\linewidth]{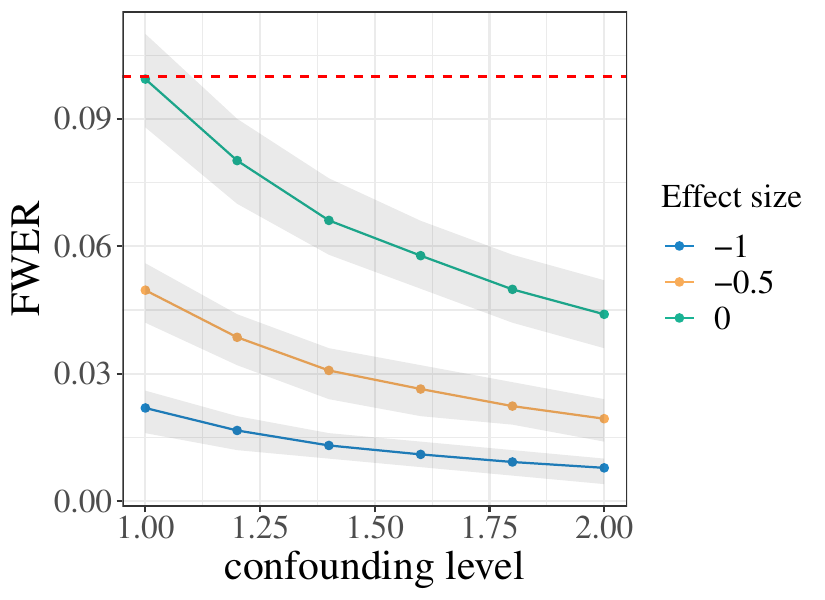}
\end{subfigure}
\begin{subfigure}[t]{0.46\linewidth}
    \centering
    \includegraphics[width=0.8\linewidth]{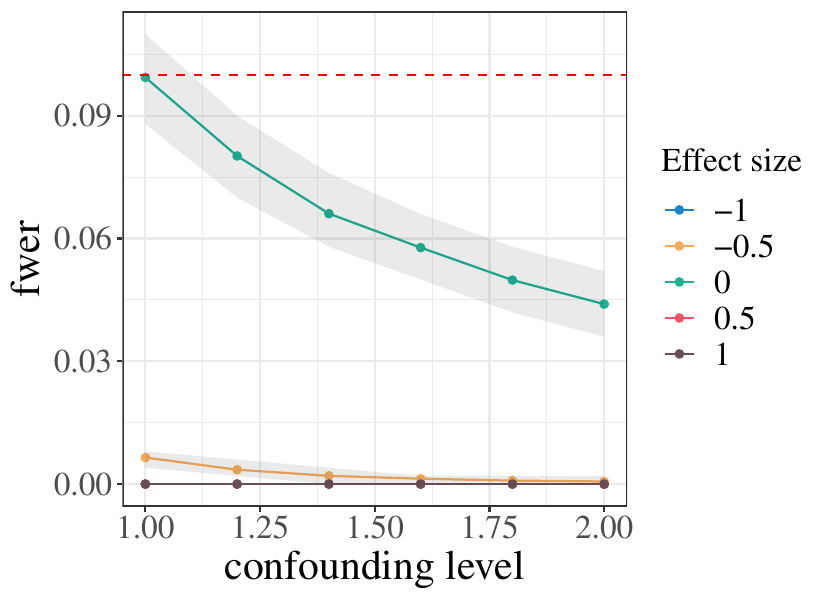}
\end{subfigure}
\caption{Empirical FWER of Algorithm~\ref{alg:mgn}. The effect size $a$ ranges in $\{-1,-0.5,0,0.5,1\}$. 
The solid lines are averaged over $N=1000$ runs, while the $0.25$-th 
and $0.75$-th quantiles form
the shaded area. 
}\label{fig:simu_ss_fwer}
\end{figure}

\begin{figure}[htbp]
\centering
\begin{subfigure}[t]{0.46\linewidth}
    \centering
    \includegraphics[width=0.8\linewidth]{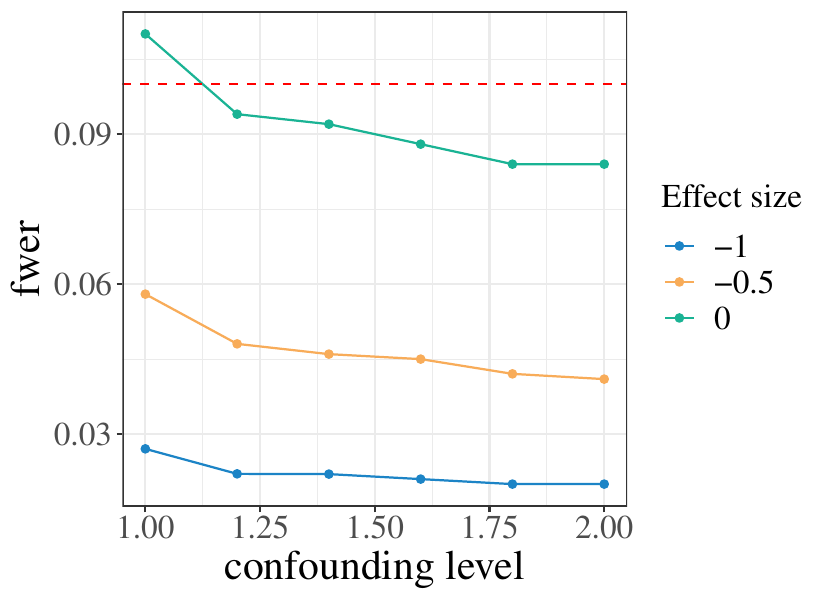}
\end{subfigure}
\begin{subfigure}[t]{0.46\linewidth}
    \centering
    \includegraphics[width=0.8\linewidth]{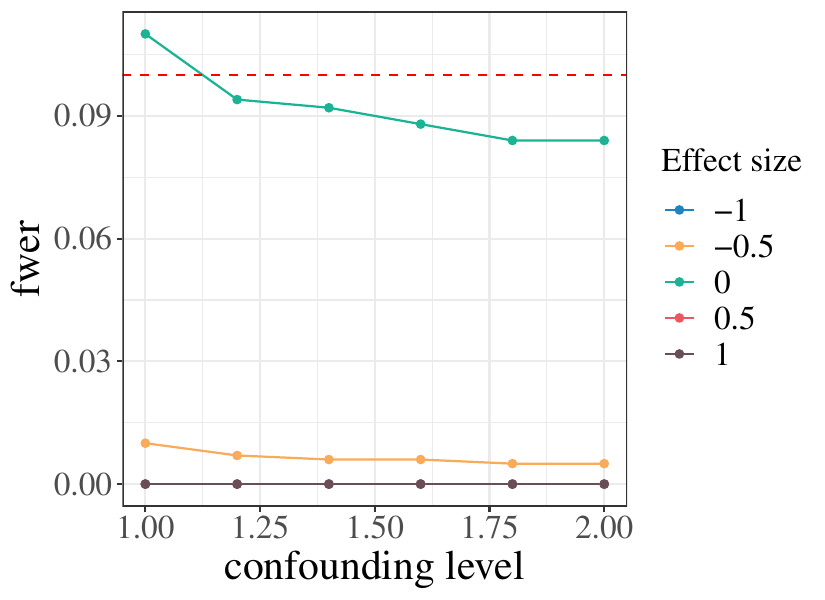}
\end{subfigure}
\caption{$0.05$-th quantile of empirical FWER using 
Algorithm~\ref{alg:whp} with estimated $\hat\ell(\cdot)$ and $\hat{u}(\cdot)$.
with the effect size $a$ ranging in $\{-1,-0.5,0,0.5,1\}$ for 
fixed ITE (left) and random ITE (right).
}\label{fig:simu_ss_fwer_pac}
\end{figure}

\subsubsection{\revise{$\Gamma$-values} } 
We plot the 
estimated survival function $\hat \Gamma$
defined as $S(\Gamma) =  \PP(\hat{\Gamma} > \Gamma)$, 
which  
characterizes the proportion of test units 
that are identified as positive ITE 
with each confounding level $\Gamma$.
Figure~\ref{fig:simu_gval_mgn} presents
the results from  one run of the procedure
with Algorithm~\ref{alg:mgn},
where we focus on the random ITE: $Y(1)-Y(0)=a\cdot U$.

\begin{figure}[ht]
\includegraphics[width=\linewidth]{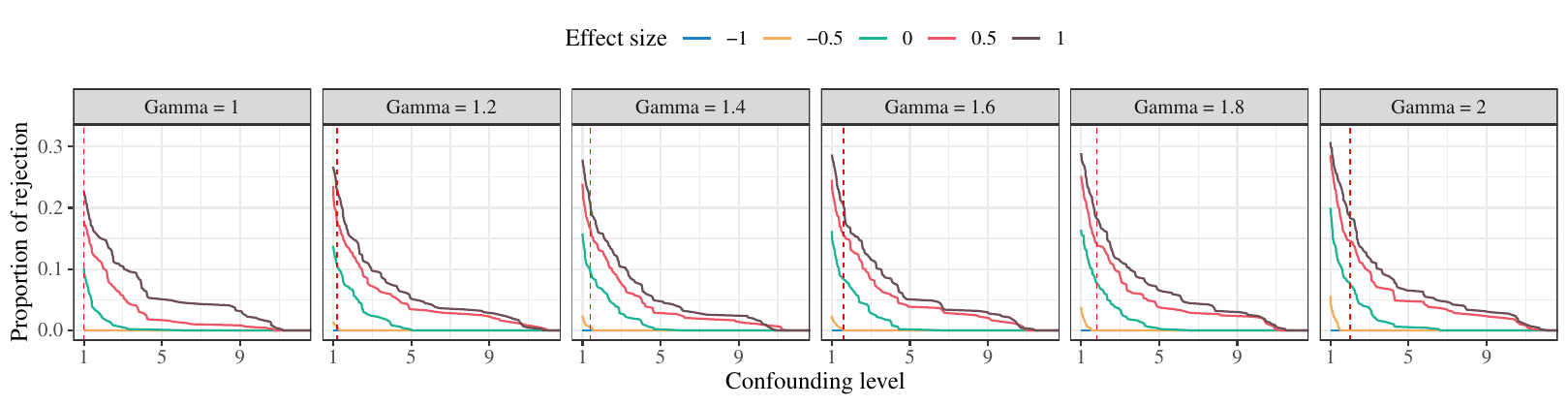}
\caption{ 
Empirical evaluation of $S(\Gamma)$
reported by (one run of) the sensitivity analysis procedure with Algorithm~\ref{alg:mgn}. The red dashed vertical lines are the true confounding levels.}
\label{fig:simu_gval_mgn}
\end{figure}

To see how to interpret these plots,
let us consider the example where
$a = 1$ and the true confounding level
is $1.6$. We can see that around $10\%$ of the 
samples have a $\Gamma$-value greater
than $2.5$, and $5\%$ have a $\Gamma$-value
greater than $5$, showing strong evidence
for positive ITEs.
Note also that the ITE is always zero when $a=0$, 
and the $\Gamma$-value should be a $90\%$ lower confidence 
bound for the true confounding level. In Figure~\ref{fig:simu_gval_mgn},
we indeed observe that for around $90\%$ of the samples,
the $\Gamma$-value is below the true confounding level (see the
green curves).

With random ITEs, 
both 
the magnitude of actual ITE  
and  
the gap between observed outcomes in treated and control groups
increase  with the effects size $a$. 
Within each subplot, 
the reported  $\Gamma$-values become larger
as the effect size increases.  
Thresholding at the true confounding level $\Gamma$, we see that 
larger magnitude of true effects also makes it easier to detect 
positive ITEs at the correct confounding level $\Gamma$.

The results
from one run of Algorithm~\ref{alg:whp} \revise{are} 
in Figure~\ref{fig:simu_gval_whp}. 
The patterns are 
similar to \revise{Figure~\ref{fig:simu_gval_mgn}}
in general, 
except that it is sometimes sharper than Algorithm~\ref{alg:mgn} 
and provides slightly stronger evidence 
against unmeasured confounding.

\begin{figure}[ht]
\includegraphics[width=\linewidth]{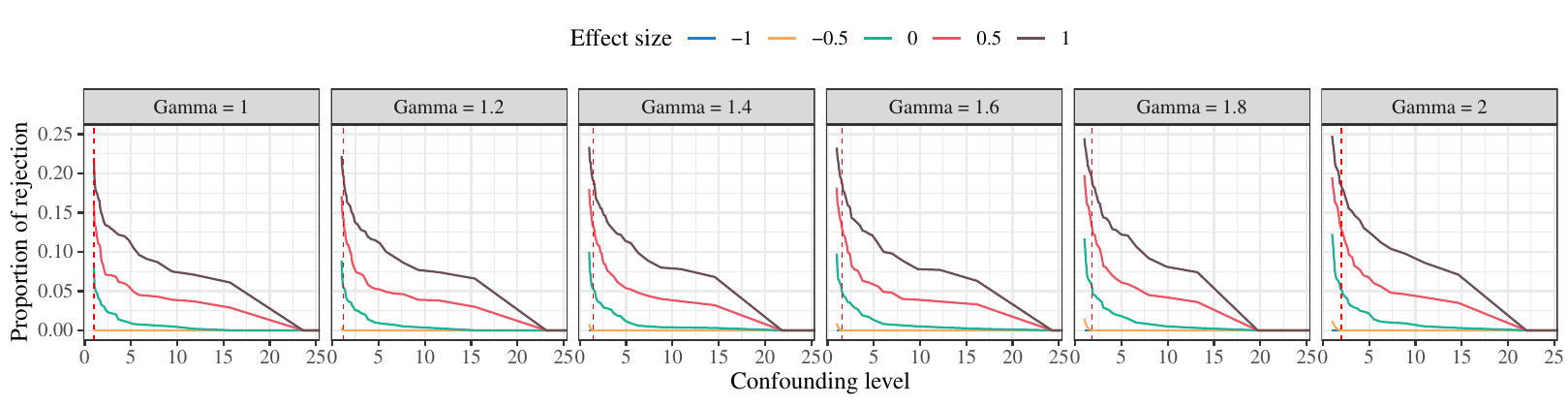}
\caption{ 
Empirical evaluation of $S(\Gamma)$ reported by 
(one run of) the sensitivity analysis procedure with Algorithm~\ref{alg:whp}. The red dashed vertical lines are the true confounding levels.}
\label{fig:simu_gval_whp}
\end{figure}

In Figure~\ref{fig:simu_gval_whp}, 
some test units have large $\Gamma$-values 
especially when the effect size is positive. 
Such strong evidence would happen if 
there is a \revise{large} gap between the observed $Y(1)$
and the typical behavior of $Y(0)$
\revise{predicted with} the training data---thus, 
the only way our procedures can 
output a prediction interval that overlaps with $Y(1)$
is  $\hat{C}(X,\Gamma) = \{V(X,y)\leq \hat{v}\} =\RR$, 
where $\hat{v}=\infty$. 
In that case, 
there is a certain value $\Gamma_\infty(X)$ 
such that $\hat{v}=\infty$ once $\Gamma \geq \Gamma_\infty(X)$. 
Since 
Algorithm~\ref{alg:whp} outputs a same $\hat{v}$ for 
all test samples, 
the value $\Gamma_\infty$ is the same 
for all such extreme individuals,
leading to steep tails 
in Figure~\ref{fig:simu_gval_whp}. 
In contrast, 
Algorithm~\ref{alg:mgn} returns 
different $\Gamma_\infty(X)$ 
for different individuals, 
\revise{and produces a smoother
curve (Figure~\ref{fig:simu_gval_mgn}).}

\subsubsection{False discovery proportions}

\revise{We also track the empirical false discovery proportion  
\$
\text{FDP}(\Gamma) = \frac{ \big|\{j\in \cD_{\test} \colon \hat\Gamma > \Gamma, ~
Y_j(1)\leq Y_j(0)\} \big| }{  \big|\{j\in \cD_{\test} \colon \hat\Gamma > \Gamma\} \big| },
\quad \Gamma\geq 1
\$
for random ITE with $a\neq 0$, 
which is the proportion of false rejections 
among test units that are rejected at confounding level $\Gamma$. 
With both the two procedures, 
we find that $\text{FDP}(\Gamma)$ is always zero---the
units that survive certain levels of adjustment for confounding 
all have positive ITE. 
It could be explained by the conservativeness of the procedure: 
to survive the adjustment, 
the observed $Y(1)$ needs to be
larger than the whole $1-\alpha$ prediction interval 
for $Y(0)$.
Therefore, a unit that survives the adjustment 
is much more likely 
to have a positive ITE, leading to vanishing FDPs. 
}

\section{Real data analysis}
\label{sec:synthetic}

\subsection{Counterfactual prediction on a semi-real dataset}\label{subsec:semi_real}
We consider an observational study dataset 
from~\citet{carvalho2019assessing}, based on which
we generate confounded
synthetic potential outcomes,  
and conduct the ATE-type prediction of $Y(1)$.  

Randomly splitting the original dataset into 
two folds of sizes $|\cZ_1|=2078$ and $|\cZ_2|=8313$ \revise{respectively}, 
we use the samples in $\cZ_1$ to fit a regression function 
$\hat\mu_0(x)$ for $\EE[Y(0)\given X=x,T=0]$ 
and a propensity score function $\hat{e}(x)$ for $\PP(T=1\given X=x)$. 
We sample $n = 10000$ covariates 
$\{X_i\}_{1\leq i\leq n}$ from $\cZ_2$ with replacement. 
Then i.i.d.~counterfactuals are generated via 
\$
Y_i(1) = \hat\mu_0(X_i) + \tau(X_i) + U_i,\quad Y_i(0) 
= \hat\mu_0(X_i) - U_i, \quad U_i \sim N(0,0.2^2),\quad i=1,\dots,n,
\$
where the conditional treatment effect function $\tau(x)$ is specified the same 
way as in equation (1) of \cite{carvalho2019assessing}. 
The propensity scores are specified as $e(X_i) := \hat{e}(X_i)$. 
For each confounding level $\Gamma \in\{1.5, 2, 2.5, 3, 5\}$, 
both $e(X_i,U_i)$ 
and $T_i$
are generated the same way as in Section~\ref{subsec:simu_mgn}.  

We then conduct counterfactual inference on the synthetic dataset 
$\cD_\obs = \{Y_i,X_i,T_i\}_{1\leq i \leq n}$. 
We randomly split $\cD_\obs$ into three folds with $1:2:1$ sizes. 
The treated samples in the first fold are used as $\cD_\train$. 
The treated samples in the second fold are used as $\cD_\calib$, 
so that $\cD=\cD_\train\cup \cD_\calib$.
All samples in the third fold (where we have ground truth even for those in the control group) 
are used as test samples. 
The process is repeated $N=1000$ times, 
where 
there are approximately $|\cD_{\calib}| = 1900$ calibration samples 
fed into the procedures and $2500$ test samples.

Figure~\ref{fig:mgn_syn} summarizes the empirical coverage  
using Algorithm~\ref{alg:mgn} 
and estimated $\hat\ell(\cdot)$, $\hat{u}(\cdot)$. 
The output of Algorithm~\ref{alg:mgn} 
always achieves valid coverage. 
The solid lines are close to the 45$^\circ$-line, 
showing the tightness of our procedure in this setting.

\begin{figure}[h]
\centering
\includegraphics[width = \linewidth]{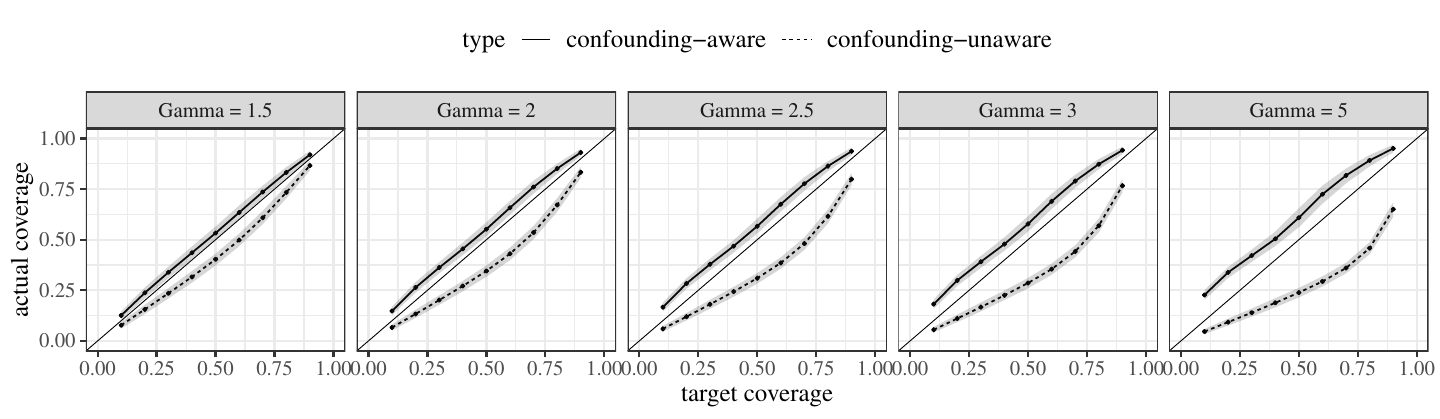}
\caption{Empirical coverage on the test sample. 
Each plot corresponds to a confounding level $\Gamma $. 
The points are average empirical coverage. 
The shaded bands corresponds to the $0.05$-th and $0.95$-th quantiles 
of coverage on test samples. The solid lines correspond 
to Algorithm~\ref{alg:mgn}. \revise{The dashed lines 
assume no confounding and are shown for comparison. In this case,
counterfactual prediction intervals are invalid.} } 
\label{fig:mgn_syn}
\end{figure}

The empirical coverage  of Algorithm~\ref{alg:whp} 
\revise{with} estimated $\hat\ell(\cdot)$, $\hat{u}(\cdot)$
is summarized in Figure~\ref{fig:whp_syn}, 
which validate the PAC-type guarantee  
(referring to the lower boundary of the shaded area in each plot, 
which is the $0.05$-th quantile of empirical coverage). 
Due to the conservativeness of the confidence lower bound 
constructed by the WSR inequality, 
the average coverage of Algorithm~\ref{alg:whp} is a bit higher than the target 
for targets around $0.5$. 
However, the 0.05-th quantile  
(lower boundary of the shaded area) 
is still very close to the target.

\begin{figure}[h]
\centering
\includegraphics[width = \linewidth]{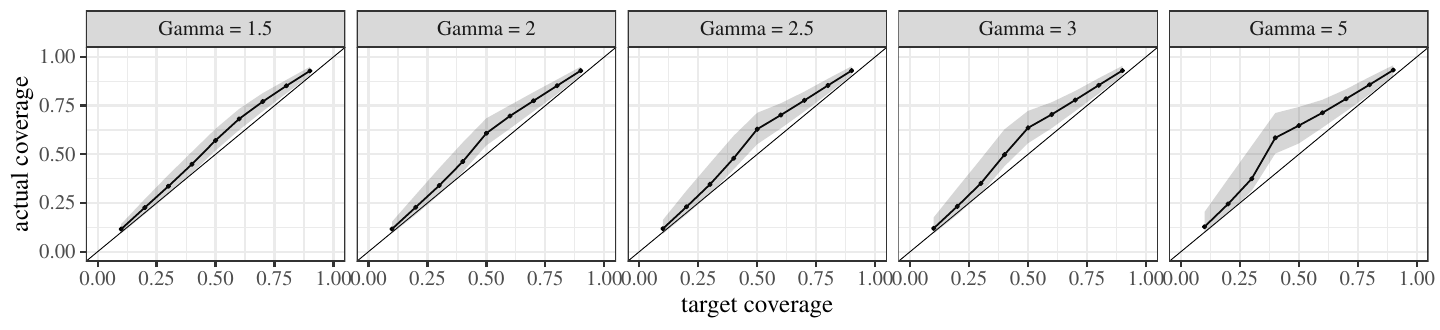}
\caption{Empirical coverage  of Algorithm~\ref{alg:whp}. 
Each plot corresponds to a confounding level $\Gamma $. 
The shaded bands corresponds to the $0.05$-th and $0.95$-th quantiles 
of coverage on test samples.} 
\label{fig:whp_syn}
\end{figure}

It is also worth pointing out that 
the shaded bands indicate the  quantiles 
of the empirical coverage on test \revise{samples}, 
which 
could be understood as an estimate of 
$\hat{c}(\cD):=\PP(Y_{n+1}\in \hat{C}(X_{n+1)})\given \cD)$. 
Although PAC-type guarantee is not theoretically provided by Algorithm~\ref{alg:mgn}, 
most of the time, 
$\hat{c}(\cD)$ 
is above the target in this example (referring to the lower boundary of the shaded area, 
which is the $0.05$-th quantile of $\hat{c}(\cD)$). 
The widths of the $0.05$-th and the $0.95$-th quantiles also provide 
empirical evidence that 
$\hat{c}(\cD)$ 
from Algorithm~\ref{alg:mgn} might be less variable 
than Algorithm~\ref{alg:whp}.
The theoretical analysis of this phenomenon 
might deserve future investigation.

\subsection{Sensitivity analysis on a real dataset}\label{subsec:real}
We finally apply the sensitivity analysis procedure to 
the treated units in the same dataset \citep{carvalho2019assessing}.
Two types of null hypotheses are considered: 
\begin{enumerate}[(1)]
    \item $H_0^-(\Gamma )\colon Y(1)-Y(0)\leq 0~\text{and}~ \Gamma ^* \leq \Gamma $.
    Rejecting the null suggests a positive ITE.
    \item $H_0^+(\Gamma )\colon Y(1)-Y(0)\geq 0~\text{and}~ \Gamma ^* \leq \Gamma $.
    Rejecting this null suggests a negative ITE. 
\end{enumerate}

We randomly subsample $1/3$ 
of the original data as $\cD_\train$. 
Among the remaining,
those with $T=0$ are used as $\cD_\calib$, 
and those with $T=1$ as the test sample $|\cD_\test|$.  
On average, we have 
$|\cD_\train|\approx 2329$, 
$|\cD_\calib|\approx 4650$
and $|\cD_\test|\approx 2250$. 
We  conduct sensitivity analysis on the dataset 
with $\alpha=0.1$ and $\delta=0.05$ (the details
are as in Section~\ref{subsec:simu_sens}).

The empirical survival function of the 
$\Gamma$-values resulting from testing
$H_0^-(\Gamma)$ 
is plotted in Figure~\ref{fig:real_positive}, 
which shows the evidence for positive ITEs. 
To illustrate the variability 
of the procedures, 
we present the estimated functions 
for all $N=10$ independent \revise{runs}, 
hence the multiple curves in each subplot.  

\begin{figure}[htbp]
\centering
\begin{subfigure}[t]{0.44\linewidth}
    \centering
    \includegraphics[width=0.9\linewidth]{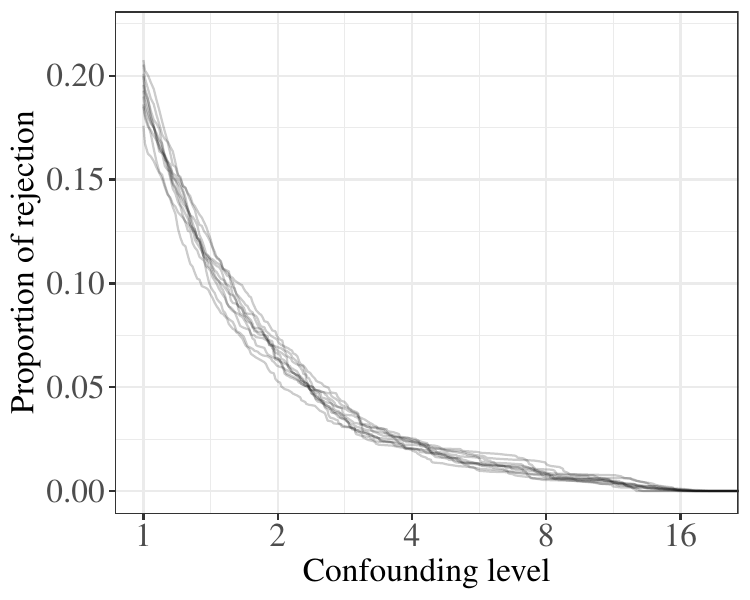} 
\end{subfigure}
\begin{subfigure}[t]{0.44\linewidth}
    \centering
    \includegraphics[width=0.9\linewidth]{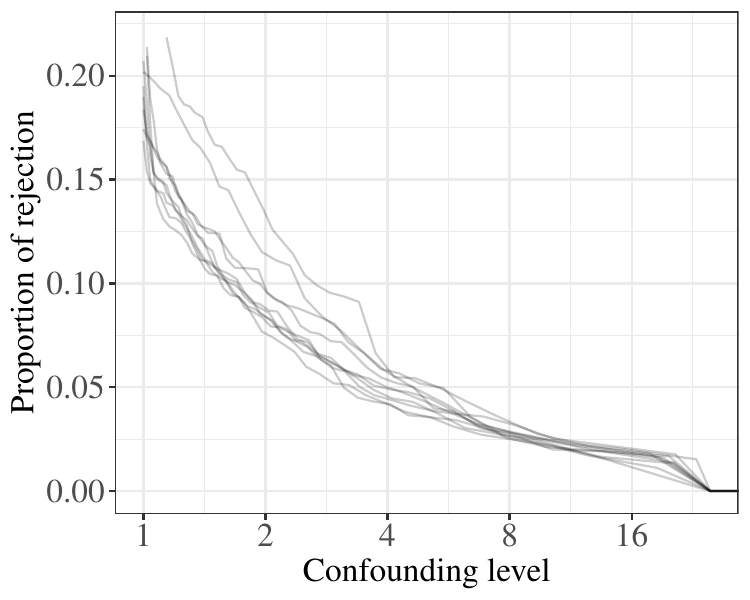} 
\end{subfigure} 
\caption{Empirical evaluation of $S(\Gamma)$ resulting from testing 
$\{H_0^-(\Gamma)\}$  
with Algorithm~\ref{alg:mgn} (left) and~\ref{alg:whp} (right).}
\label{fig:real_positive}
\end{figure}

Averaged over multiple independent runs, 
there are $19.60\%$ 
or $20.46\%$ of the treated test \revise{samples} 
that we find at $\Gamma=1$ as positive ITEs, 
using Algorithm~\ref{alg:mgn} and~\ref{alg:whp}, respectively. 
There are $6.80\%$ or $9.65\%$ 
of the test sample that we find at $\Gamma=2$ as positive ITEs. 
At $\Gamma=3$, the proportion is
$3.54\%$ with Algorithm~\ref{alg:mgn} 
and $6.95\%$ with Algorithm~\ref{alg:whp}. 
With Algorithm~\ref{alg:mgn},
around $2\%$ of the test sample have a $\Gamma$-value 
greater than $5$. 
With Algorithm~\ref{alg:whp},
about $2.5\%$ of the test sample have $\Gamma$-values
greater than $10$, and some test samples have a $\Gamma$-value
as large as $25$, 
showing  robust evidence of a positive ITE. 
Our framework guarantees that 
the mistake (i.e., 
rejecting an actually negative ITE 
at a too large confounding level) 
we make on all units 
is bounded by $\alpha=0.1$ on average.

On the other hand, we report 
the proportion of test units such that $H_0^+(\Gamma)$ 
is rejected for all $\Gamma\geq 1$ in Figure~\ref{fig:real_negative}, 
which shows the evidence for negative ITEs. 

\begin{figure}[H]
\centering
\begin{subfigure}[t]{0.44\linewidth}
    \centering
    \includegraphics[width=0.9\linewidth]{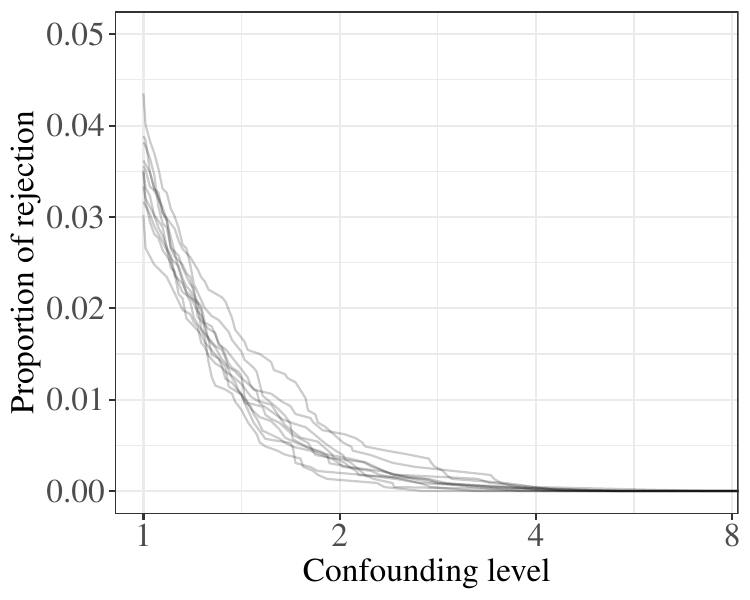} 
\end{subfigure}
\begin{subfigure}[t]{0.44\linewidth}
    \centering
    \includegraphics[width=0.9\linewidth]{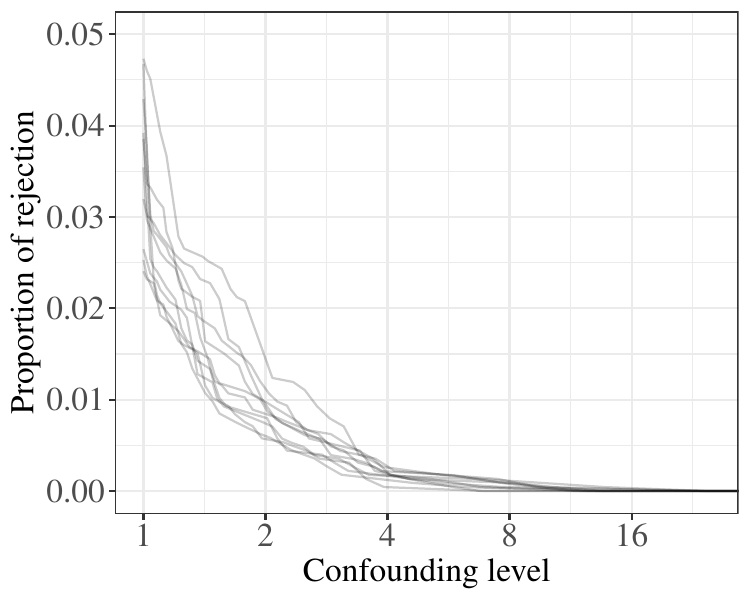} 
\end{subfigure}
\caption{Empirical evaluation of $S(\Gamma)$ resulting from testing 
$\{H_0^+(\Gamma)\}$  
with Algorithm~\ref{alg:mgn} (left) and~\ref{alg:whp} (right).}
\label{fig:real_negative}
\end{figure}

Averaged over multiple runs, 
there are $3.58\%$ and $3.58\%$ 
of ITEs of the treated test sample 
that we find as negative at $\Gamma=1$, 
using Algorithm~\ref{alg:mgn} and~\ref{alg:whp}, respectively. 
At $\Gamma=2$, these proportions are 
$0.38\%$ with Algorithm~\ref{alg:mgn} and 
$1.01\%$ with Algorithm~\ref{alg:whp}. 
Algorithm~\ref{alg:whp} produces a little stronger evidence 
against unmeasured confounding, but 
it is slightly less \revise{stable}. 
In general, very few of the samples have 
$\Gamma$-values larger than $2$ using both algorithms.

\section{Discussion}

We proposed a \revise{model-free} framework for sensitivity
analysis of ITEs,
building upon reliable counterfactual inference 
with potentially confounded observational data. 
We close the paper by discussing possible extensions of the current work.

One extension is to test the confounding level with 
a little more information. 
If we have some
budget to obtain a small amount of 
experimental data, 
it would be interesting to use 
these experimental data to test the existence or level of 
unmeasured confounding in the observational data. 
Another interesting question 
is whether the experimental data can be used to calibrate the value 
of $\Gamma$. 
Taking a step further, the calibrated $\Gamma$  
might also be used to provide valid counterfactual inference 
for units in the observational study. 

Other extensions might be based on the 
identification of distributional shifts 
in~\eqref{eq:id_set_causal} and~\eqref{eq:causal_worst},
which  
can be of interest for other tasks. 
For example, the recent work of~\cite{dorn2021sharp} 
studies sharp bounds on average treatment effect (ATE) under
the marginal $\Gamma$-selection condition~\eqref{eq:gamma_sel_mgn}, 
in which the identification set is the same as~\eqref{eq:data_comp_idset}. 
\revise{
The worst-case distribution function in~\eqref{eq:causal_worst} 
achieves the sharp lower bound on the (conditional) expectation of 
$V(X,Y)$ when $(X,Y)\sim \TPP$. 
Therefore, setting $V(x,y)= y$ (resp. $V(x,y)=-y$),
one could identify the 
super-population $\TPP$ 
that achieves the lower (resp. upper) bound on 
$\tilde{\EE}[Y(t)\given X]$, $t\in\{0,1\}$. 
In policy evaluation problems,  
the task is to estimate $V(\pi) = \tilde\EE[\pi(X)Y(1)+(1-\pi(X))Y(0)]$
for some policy $\pi\colon \cX\to [0,1]$ 
with possibly confounded observational data. 
Utilizing~\eqref{eq:id_set_causal}
and~\eqref{eq:causal_worst} 
along with a coupling argument, 
one could construct a super-population $\TPP$ 
that achieves the sharp lower bounds on $\tilde{\EE}[Y(1)\given X]$ 
and $\tilde{\EE}[Y(0)\given X]$ simultaneously, 
which leads to the sharp lower bound on $\tilde\EE[\pi(X)Y(1)+(1-\pi(X))Y(0)\given X]$,
hence on $V(\pi)$. 
In fact, it can also be shown that 
all policies attain
their worst-case performance
under one single super-population.  
Such result might be used  
to learn a policy 
whose worst-case performance 
is most favorable.}

\subsection*{Code availability and reproducibility}
The code for reproducing the simulations in 
Section~\ref{subsec:simu_mgn},~\ref{subsec:simu_whp}
and real data analysis in Section~\ref{sec:synthetic} is publicly available at \url{https://github.com/ying531/cfsensitivity_paper}.
An R-package to implement the procedures proposed in this paper
can be found at \url{https://github.com/zhimeir/cfsensitivity}.

\subsection*{Acknowledgement}
The authors thank Isaac Gibbs, Kevin Guo, Suyash Gupta, Jayoon Jang,  Lihua Lei, Shuangning Li and Dominik
Rothenh\"ausler for helpful discussions. E.~C.~was supported by the
Office of Naval Research grant N00014-20-12157, the National Science
Foundation grant DMS 2032014, the Simons Foundation under award
814641, and the ARO grant 2003514594. Y.~J.~was partially supported by
ARO grant 2003514594. Z.~R.~was partially supported by ONR grant
N00014-20-1-2337, and NIH grants R56HG010812, R01MH113078 and R01MH123157.

\bibliographystyle{apalike}
\bibliography{reference}

\newpage 
\appendix 

\section{Additional results}

\subsection{Constructing $\hat{G}_n(\cdot)$ in the PAC-type procedure} \label{app:pac}

In this part, we provide the construction of $\hat{G}_n(t)$ in 
Section~\ref{subsec:whp_procedure} with 
Waudby-Smith--Ramdas bound~\citep{waudbysmith2021estimating}.
The 
proof is a slight modification of~\cite{waudbysmith2021estimating} and 
included in Appendix~\ref{appendix:proof_wsr_bound} for completeness.

\begin{prop}[Waudby-Smith–Ramdas lower confidence bound for c.d.f.s]\label{prop:wsr_bound}
Suppose 
$\sup_x \hat u(x)\leq M$ for some constant $M>0$. 
For $t\in \RR$ being any constant or any random variable in $\sigma(\cD_\train)$,
and any $\delta\in(0,1)$, 
we define $\hat{G}_n(t) = \max\big\{\hat{G}_n^L(t),\hat{G}_n^U(t)\big\}$, where 
\$
&\hat{G}_n^L(t) = M\cdot  \inf\big\{ g\geq 0\colon \max_{1\leq i\leq \nc} \cK_i^L(g) \leq 2/\delta  \big\},\\
&\hat{G}_n^U(t) =  1-M+M\cdot\inf\big\{ g\geq 0\colon \max_{1\leq i\leq \nc} \cK_i^U(g) \leq 2/\delta  \big\}.
\$
For any $g\geq0,$ the thresholding functions 
for $i=1,\dots,\nc$ 
are defined as 
\$
&\cK_i^L(g) = \prod_{j=1}^i \Big(1+ \nu_j^L\cdot\big[ \ind_{\{V_j\leq t\}}\hat{\ell}(X_j)/M - g\big]\Big),~~\cK_i^U(g) = \prod_{j=1}^i \Big(1+ \nu_j^U\cdot\big[ 1- \ind_{\{V_j> t\}}\hat{u}(X_j)/M - g\big]\Big),
\$
where $\nu_j^L = \min\big\{1, \sqrt{ 2\log(2/\delta)/[n(\hat{\sigma}_{j-1}^L)^2]}\big\}$, $\nu_j^U = \min\big\{1, \sqrt{ 2\log(2/\delta)/[n(\hat{\sigma}_{j-1}^U)^2]}\big\}$, and 
\$
& (\hat\sigma_{i}^L)^2 = \frac{\frac{1}{4}+\sum_{j=1}^i\big(\ind_{\{V_j\leq t\}}\frac{\hat{\ell}(X_j)}{M} - \hat\mu_j^L\big)^2 }{1+i},~~ \hat\mu_i^L = \frac{\frac{1}{2}+\sum_{j=1}^i \ind_{\{V_i\leq t\}}\frac{\hat{\ell}(X_i)}{M}}{1+i},\\
&  (\hat\sigma_{i}^U)^2 = \frac{\frac{1}{4}+\sum_{j=1}^i\big(1-\ind_{\{V_j> t\}}\frac{\hat{u}(X_j)}{M} - \hat\mu_j^U\big)^2 }{1+i},~~ \hat\mu_i^U = \frac{\frac{1}{2}+\sum_{j=1}^i (1-\ind_{\{V_i> t\}}\frac{\hat{u}(X_i)}{M})}{1+i}.
\$
Then it holds that
$
\PP_{\cD_\calib}\big(\hat{G}_n(t)\leq G(t)\big)\geq 1-\delta
$
for $G(t)$ defined in~\eqref{eq:def_G_t_formula}.
\end{prop}

\subsection{Validity of prediction intervals for ITE} \label{app:ite_valid}

In this part, we provide the coverage guarantee for 
prediction of ITEs omitted in Section~\ref{subsec:pred_ite}. 

\begin{prop} \label{prop:ite_mgn}
Consider a new test sample where we observe 
$(X_{n+1}, T_{n+1}, Y_{n+1}(T_{n+1}))$, 
with $T_{n+1} = w$ for $w\in\{0,1\}$.
Let $\cD$ be the calibration data generated from $\PP^\super$ 
under confounding level $\Gamma$. 
If
$\hat{C}_{1-w}(X_{n+1},\Gamma,1-\alpha)$ is constructed by Algorithm~\ref{alg:mgn},
then 
\$
\PP\big( Y_{n+1}(1)-Y_{n+1}(0)\in \hat{C}(X_{n+1},\Gamma) \biggiven T_{n+1} = w\big)
\geq 1-\alpha - \hat\Delta,
\$
for the prediction set $\hat{C}(X_{n+1},\Gamma,1-\alpha)$ 
in~\eqref{eq:ite_set}, where the  probability is over $\cD_{\calib}$
as well as the test point, 
and $\hat\Delta$ is the gap of coverage for $Y_{n+1}(1-w)$ in Theorem~\ref{thm:mgn}.
If
$\hat{C}_{1-w}(X_{n+1},\Gamma,1-\alpha)$ is constructed by Algorithm~\ref{alg:whp}
with confidence level $\delta\in(0,1)$, 
then with probability at least $1-\delta$ with respect to $\cD_\calib$, 
we have 
\$
\PP\big( Y_{n+1}(1)-Y_{n+1}\in \hat{C}(X_{n+1},\Gamma )  
\biggiven T_{n+1} = w,\cD_{\calib} \big)
\geq 1-\alpha -\hat\Delta,
\$
where $\hat\Delta$ is the gap of coverage for $Y_{n+1}(1-w)$ in Theorem~\ref{thm:whp}.
\end{prop}

\begin{prop}\label{prop:ite_whp}
Consider a new test sample $X_{n+1}$.
Let $\cD$ be the observations for which 
$\PP^\super$ is under confounding level $\Gamma$. 
If for $w \in\{0,1\}$,
$\hat{C}_{w}(X_{n+1},\Gamma,1-\alpha/2,\delta/2)$ 
is constructed by Algorithm~\ref{alg:mgn},
then 
\$
\PP\big( Y_{n+1}(1)-Y_{n+1}(0)\in \hat{C}(X_{n+1},\Gamma )   \big)
&\geq 1-\alpha - \hat\Delta_1 - \hat\Delta_0,
\$
where the  probability is over 
$\cD$ as well as the test point;
$\hat\Delta_0$, $\hat\Delta_1$ is the coverage gap in Theorem~\ref{thm:mgn}
for counterfactual prediction of $Y_{n+1}(1),Y_{n+1}(0)$ 
when the bound functions are estimated. 
If
$\hat{C}_{1-w}(X_{n+1},\Gamma,1-\alpha/2,\delta/2)$ is
constructed by Algorithm~\ref{alg:whp},
then with probability at least $1-\delta$ with respect to $\cD_{\calib}$, we have 
\$
\PP\big( Y_{n+1}(1)-Y_{n+1}\in \hat{C}(X_{n+1},\Gamma )  \biggiven \cD_{\calib} \big)
&\geq 1-\alpha - \hat\Delta_1 - \hat\Delta_0,
\$
where $\hat\Delta_0$ and $\hat\Delta_1$ are 
the coverage gaps in Theorem~\ref{thm:mgn}
for counterfactual prediction of $Y_{n+1}(1),Y_{n+1}(0)$ 
when the bound functions are estimated. 
\end{prop}

\section{Technical proofs}

\subsection{Proof of Lemma~\ref{lem:shift}}
\label{supp:pf_lem_shift}

\begin{proof}[Proof of Lemma~\ref{lem:shift}]
For any measurable subset $A\subset \cU$, any $u\in A$ and any $x\in \cX$,
  by the marginal $\Gamma$-selection condition \eqref{eq:gamma_sel_mgn},
  \#\label{eq:bd1}
  \frac{1}{\Gamma} \cdot \PP(T=0\given X=x,U=u) \leq \PP(T=1\given X=x,U=u)  \cdot \frac{\PP(T=0\given X=x)}{\PP(T=1\given X=x)} \leq \Gamma \cdot \PP(T=0\given X=x,U=u).
  \#
  Marginalizing over $u\in A$ yields
  \#\label{eq:bnds_meas_sets}
  \frac{1}{\Gamma} \leq \frac{\PP(U\in A\given X=x, T=1)}{\PP(U\in A\given X=x, T=0)}  \leq \Gamma
  \#
  for $\PP$-almost $x\in \cX$. Since~\eqref{eq:bnds_meas_sets} holds for any measurable set
  in $\cU$, we have 
  \$
  \frac{1}{\Gamma} \le \frac{\ud \PP_{U \given X,T = 1}}{\ud \PP_{U \given X,T=0}}(u,x) \le \Gamma,
  \$
  for $\PP$-almost all $u\in\cU$ and $x \in \cX$. Meanwhile, for any measurable set $B\subset \cY$, 
  by the tower property, we have for any $t\in \{0,1\}$ that
  \#\label{eq:bd2}
  \PP\big(Y(1)\in B,T = t\biggiven X\big) &= 
  \EE\Big[ \EE\big[ \ind_{\{Y(1)\in B\}}\ind_{\{T=t\}}\biggiven X,U]\Biggiven X\Big] \notag \\
  &=  \EE\Big[ \EE\big[ \ind_{\{Y(1)\in B\}} \biggiven X,U]\cdot \EE\big[ \ind_{\{T=t\}} \biggiven X,U]\Biggiven X\Big].
  \#
  Rewriting~\eqref{eq:bd1}, we have $\PP$-almost surely that
  \$
  \frac{1}{\Gamma} \cdot \EE\big[ \ind_{\{T=0\}} \biggiven X,U] \cdot \frac{\EE\big[ \ind_{\{T=1\}} \biggiven X ]}{\EE\big[ \ind_{\{T=0\}} \biggiven X ]} \leq \EE\big[ \ind_{\{T=1\}} \biggiven X,U] \leq \Gamma \cdot \EE\big[ \ind_{\{T=0\}} \biggiven X,U] \cdot \frac{\EE\big[ \ind_{\{T=1\}} \biggiven X ]}{\EE\big[ \ind_{\{T=0\}} \biggiven X ]}.
  \$
  Multiplying all sides by $ \EE\big[ \ind_{\{Y(1)\in B\}} \biggiven X,U]$ and using~\eqref{eq:bd2}, we know
  \$
  &\frac{1}{\Gamma} \cdot \PP\big(Y(1)\in B,T=0\biggiven X\big)\cdot \frac{\EE\big[ \ind_{\{T=1\}} \biggiven X ]}{\EE\big[ \ind_{\{T=0\}} \biggiven X ]}
  \leq \PP\big(Y(1)\in B,T=1\biggiven X\big) 
  \leq  \Gamma \cdot \PP\big(Y(1)\in B,T=0\biggiven X\big) \cdot \frac{\EE\big[ \ind_{\{T=1\}} \biggiven X ]}{\EE\big[ \ind_{\{T=0\}} \biggiven X ]} 
  \$
  holds $\PP$-almost surely, and for $\PP$-almost all $x \in \cX$,
  \$
  \frac{1}{\Gamma} \cdot  \frac{1-e(x)}{e(x)}  \leq \frac{\PP(Y(1)\in B,T=0\given X=x)}{\PP(Y(1)\in B,T=1\given X=x)}  \leq \Gamma \cdot  \frac{1-e(x)}{e(x)}.
  \$
  Note that 
  \$
  \frac{\PP(Y(1)\in B\given X=x,T=1)}{\PP(Y(1) \in B\given X=x,T=0)} 
  = \frac{\PP(Y(1)\in B, T=1 \given X=x)}{\PP(Y(1) \in B, T = 0\given X=x)} 
  \cdot \frac{1 - e(x)}{e(x)}.
  \$
  Consequently, 
  \#\label{eq:bd3}
  \frac{1}{\Gamma} \le \frac{\PP(Y(1)\in B\given X=x,T=1)}{\PP(Y(1) \in B\given X=x,T=0)} \le \Gamma
  \#
  holds for $\PP$-almost all $x\in \cX$. By the arbitrariness of $B$, we have 
  \$
    \frac{1}{\Gamma} \le \frac{\ud \PP_{Y(1) \given X,T=1}}{\ud \PP_{Y(1) \mid X,T=0}}(x,y)\le \Gamma.
  \$
  Repeating the above steps for $Y(0)$ we conclude the proof of Lemma~\ref{lem:shift}.
\end{proof}

\subsection{Proof of Theorem \ref{thm:mgn}} \label{app:proof_thm_mgn}
\begin{proof}[Proof of Theorem~\ref{thm:mgn}]
Fixing any $\TPP\in \cP(\PP,\ell,u)$, we denote the likelihood ratio 
$w(x,y) = \frac{\ud \TPP}{\ud \PP}(x,y)$. Recall that the calibration data
is $\{(X_i,Y_i)\}_{i\in \cD_\calib}$ 
with $\cD_\calib = \{1,\dots,n\}$, 
and the test data point is $(X_{n+1},Y_{n+1}) \sim \TPP$. 
We denote the random variables $Z_i=(X_i,Y_i)$ 
and realized values $z_i=(x_i,y_i)$ for $i=1,\dots,n$. 

As a starting point, 
we elaborate on the weighted conformal inference introduced 
in~\cite{TibshiraniBCR19}, which paves the way for the analysis of 
marginal coverage later on. Following~\cite{TibshiraniBCR19}, 
the random variables $\{Z_i\}_{i=1}^{n+1}$ are \emph{weighted exchangeable}, 
meaning that the density of their joint distribution can be factorized as 
\#\label{eq:weighted_exchangeable}
f(z_1,\dots,z_{n+1}) = \prod_{i=1}^{n+1} w_i(z_i)\cdot g(z_1,\dots,z_{n+1}),
\#
where $g$ is some permutation-invariant function, 
i.e., $g(z_{\sigma(1)},\dots,z_{\sigma(n+1)})=g(z_1,\dots,z_{n+1})$ 
for any permuation $\sigma$ of $1,\dots,n+1$. 
Specifically, here $w_i(z)=1$ for $1\leq i \leq n$ and $w_{n+1}(z ) = w(x,y)$.  
For a set of values $z_1,\dots,z_{n+1}$ where there may be repeated elements, 
we denote the unordered set 
$z = [z_1,\dots,z_{n+1}]$ and  
the event 
\$
\cE_z = \big\{ [Z_1,\dots,Z_{n+1}] = [z_1,\dots,z_{n+1}]  \big\}.
\$
Let $\Pi_{n+1}$ be the set of all permutations of $\{1,\dots,n+1\}$. 
Writing $v_i=V(x_i,y_i)=V(z_i)$,  for each $1\leq i\leq n+1$, it holds that
\$
\PP(V_{n+1} = v_i\given \cE_z) = 
\PP(Z_{n+1} = z_i\given \cE_z) = 
\frac{\sum_{\sigma\in \Pi:\sigma(n+1)=i}f(z_{\sigma(1)},\dots,z_{\sigma(n+1)})}{\sum_{\sigma\in\Pi} f(z_{\sigma(1)},\dots,z_{\sigma(n+1)})},
\$
where $\PP$ is induced by the joint distribution of 
$\cD_\train\cup \cD_\calib \cup Z_{n+1}$. 
By the factorization \eqref{eq:weighted_exchangeable}, we have
\$
\frac{\sum_{\sigma\in \Pi:\sigma(n+1)=i}f(z_{\sigma(1)},\dots,z_{\sigma(n+1)})}{\sum_{\sigma\in\Pi} f(z_{\sigma(1)},\dots,z_{\sigma(n+1)})} = \frac{  \sum_{ \sigma\in \Pi:\sigma(n+1)=i}  w_{n+1}(z_{i}) g(z_{\sigma(1)},\dots,z_{\sigma(n+1)})   }{\sum_{\sigma\in\Pi} w_{n+1}(z_{\sigma(n+1)})g(z_{\sigma(1)},\dots,z_{\sigma(n+1)})  } = \frac{w_{n+1}(z_i)}{\sum_{j=1}^{n+1}w_{n+1}(z_j)}.
\$
Therefore, the distribution of $V_{n+1}$ conditional on the event $\cE_z$ is 
\$
V_{n+1}\given \cE_z ~\sim~  \sum_{i=1}^{n+1} \delta_{v_i} p_i^w,\quad \text{ where } ~p_i^w =  \frac{w_{n+1}(z_{i})}{\sum_{j=1}^{n+1}w_{n+1}(z_j)} = \frac{w(x_i,y_i)}{\sum_{j=1}^{n+1}w(x_i,y_i)}.
\$
Here $\delta_v$ denotes the point mass at $\{v\}$. 
For any unordered set $z=[z_1,\dots,z_n,z_{n+1}]$ and 
the corresponding $k^*$ as defined in \eqref{eq:ci_mgn}, it holds that
\$
\PP\big(Y_{n+1}\in \hat{C}(X_{n+1})\biggiven \cE_z\big) =  \PP\big( V_{n+1} \leq v_{[ {k}^*]}\biggiven \cE_z\big) = \sum_{i=1}^{n+1} p_i^w \ind_{\{v_i\leq v_{[ {k}^*]}\}} = \frac{\sum_{i=1}^{n+1}w(x_i,y_i)\ind_{\{v_i\leq v_{[ {k}^*]}\}}}{\sum_{j=1}^{n+1}w(x_i,y_i)}.
\$
By the tower property of conditional expectations, we have
\#\label{eq:mgn_cover_formula}
\PP\big(Y_{n+1}\in \hat{C}(X_{n+1}) \big) = \EE\Big[ \PP\big(Y_{n+1}\in \hat{C}(X_{n+1})\biggiven \cE_Z\big)  \Big] = \EE\,\Bigg[\, \frac{\sum_{i=1}^{n+1}w(X_i,Y_i)\ind_{\{V_i\leq V_{[ {k}^*]}\}}}{\sum_{j=1}^{n+1}w(X_i,Y_i)} \Bigg].
\#

Equipped with the above preparations, we show 
the coverage guarantee in our setting. 
By definition \eqref{eq:ci_mgn}, we know
\$
 V_{[ {k}^*]} = \inf\bigg\{ v\colon \frac{\sum_{i=1}^n \hat{\ell}(X_i)\ind_{\{V_i\leq v\}}}{\sum_{i=1}^n \hat{\ell}(X_i)\ind_{\{V_i\leq v\}} + \sum_{i=1}^n \hat{u}(X_i)\ind_{\{V_i> v\}}+\hat{u}(X_{n+1}) }   \geq 1-\alpha \bigg\},
\$
hence  
\$
\EE\,\Bigg[\, \frac{\sum_{i=1}^{n}\hat{\ell}(X_i)\ind_{\{V_i\leq V_{[ {k}^*]}\}}
}{
\sum_{i=1}^n \hat{\ell}(X_i)\ind_{\{V_i\leq V_{[ {k}^*]}\}} 
+ \sum_{i=1}^n \hat{u}(X_i)\ind_{\{V_i> V_{[ {k}^*]}\}}
+\hat{u}(X_{n+1}) } \Bigg] \geq 1-\alpha
\$
since the inner random variable is always no smaller than $1-\alpha$. 
Combined with~\eqref{eq:mgn_cover_formula} and
by the non-negativity of $w(X_i,Y_i)$,  we have
\$
&\PP\big(Y_{n+1}\in \hat{C}(X_{n+1}) \big) - (1-\alpha) \\
&\geq   \EE\,\Bigg[\, \frac{\sum_{i=1}^{n}w(X_i,Y_i)\ind_{\{V_i\leq V_{[{k}^*]}\}}
}{
\sum_{i=1}^{n+1}w(X_i,Y_i)} \Bigg]   - 
\EE\,\Bigg[ \,\frac{\sum_{i=1}^{n}\hat{\ell}(X_i)\ind_{\{V_i\leq V_{[ {k}^*]}\}}
}{
\sum_{i=1}^n \hat{\ell}(X_i)\ind_{\{V_i\leq V_{[ {k}^*]}\}} 
+ \sum_{i=1}^n \hat{u}(X_i)\ind_{\{V_i> V_{[ {k}^*]}\}}
+\hat{u}(X_{n+1}) } \Bigg] 
:= \EE\bigg[ \frac{\textrm{(i)}}{\textrm{(ii)}} \bigg],
\$
where we denote 
\$
\textrm{(i)} &= - w(X_{n+1},Y_{n+1}) \cdot 
\sum_{i=1}^{n}\hat{\ell}(X_i)\ind_{\{V_i\leq V_{[ {k}^*]}\}}  
+ ~\hat{u}(X_{n+1}) \cdot \sum_{i=1}^{n}w(X_i,Y_i)\ind_{\{V_i\leq V_{[ {k}^*]}\}} \\
&\qquad + \bigg[ \sum_{i=1}^{n}w(X_i,Y_i)\ind_{\{V_i\leq V_{[ {k}^*]}\}}\bigg] 
\bigg[  \sum_{i=1}^n \hat{u}(X_i)\ind_{\{V_i> V_{[ {k}^*]}\}} \bigg] 
-\bigg[ \sum_{i=1}^{n}w(X_i,Y_i) \ind_{\{V_i> V_{[ {k}^*]}\}} \bigg] 
\bigg[   \sum_{i=1}^{n}\hat{\ell}(X_i)\ind_{\{V_i\leq V_{[ {k}^*]}\}} \bigg] 
\$
and
\$
\textrm{(ii)} = \bigg[\sum_{i=1}^{n+1}w(X_i,Y_i)\bigg] 
\bigg[\sum_{i=1}^n \hat{\ell}(X_i)\ind_{\{V_i\leq V_{[\hat{k}^*]}\}} + \sum_{i=1}^n \hat{u}(X_i)\ind_{\{V_i> V_{[\hat{k}^*]}\}}+\hat{u}(X_{n+1}) \bigg].
\$
We first establish a lower bound for the term (i). To this end, we define the random variables
\$
 \tilde{\ell}_i &= \max\big\{ w(X_i,Y_i), ~\hat{\ell}(X_i)\big\}\quad \text{and}\\
 \tilde{u}_i &= \min\big\{  w(X_i,Y_i),~ \hat{u}(X_i)\big\},\quad i=1,\dots,n+1.
\$
By the above definition, it holds that $0\leq \tilde{u}_i \leq \hat{u}(X_i)$ and $\tilde{\ell}_i \geq \hat{\ell}(X_i)\geq 0$.
We also define the differences
\$
\Delta\tilde{\ell}_i &= \tilde{\ell}_i - w(X_i,Y_i) =   \big[  \hat{\ell}(X_i)  - w(X_i,Y_i)\big]_+\quad \text{and}\\
\Delta \tilde{u}_i &= \tilde{u}_i - w(X_i,Y_i) = - \big[ \hat{u}(X_i) - w(X_i,Y_i) \big]_-,\quad i=1,\dots,n+1.
\$
Using the above notation, we have 
\$
&\bigg[ \sum_{i=1}^{n}w(X_i,Y_i)\ind_{ \{V_i\leq V_{[{k}^*]} \}}\bigg] 
\bigg[  \sum_{i=1}^n \hat{u}(X_i)\ind_{\{V_i> V_{[{k}^*]}\}} \bigg]  
-\bigg[ \sum_{i=1}^{n}w(X_i,Y_i) \ind_{\{V_i> V_{[{k}^*]}\}} \bigg] 
\bigg[   \sum_{i=1}^{n}\hat{\ell}(X_i)\ind_{\{V_i\leq V_{[{k}^*]}\}} \bigg] \\
&\geq \bigg[ \sum_{i=1}^{n}w(X_i,Y_i)\ind_{ \{V_i\leq V_{[{k}^*]} \}}\bigg] 
\bigg[  \sum_{i=1}^n \big( w(X_i,Y_i) +\Delta \tilde{u}_i\big) \ind_{\{V_i> V_{[{k}^*]}\}} \bigg] \\
&\qquad  -\bigg[ \sum_{i=1}^{n}w(X_i,Y_i) \ind_{\{V_i> V_{[{k}^*]}\}} \bigg] 
\bigg[ \sum_{i=1}^{n}\big( w(X_i,Y_i) +\Delta\tilde{\ell}_i \big)\ind_{\{V_i\leq V_{[{k}^*]}\}} \bigg] \\
&= \bigg[ \sum_{i=1}^{n}w(X_i,Y_i)\ind_{ \{V_i\leq V_{[{k}^*]} \}}\bigg] 
\bigg[  \sum_{i=1}^n  \Delta \tilde{u}_i  \ind_{\{V_i> V_{[{k}^*]}\}} \bigg]  -\bigg[ \sum_{i=1}^{n}w(X_i,Y_i) \ind_{\{V_i> V_{[{k}^*]}\}} \bigg] \bigg[   \sum_{i=1}^{n} \Delta\tilde{\ell}_i  \ind_{\{V_i\leq V_{[{k}^*]}\}} \bigg].
\$
Following similar arguments, we have 
\$
&- w(X_{n+1},Y_{n+1}) \cdot \sum_{i=1}^{n}\hat{\ell}(X_i)\ind_{\{V_i\leq V_{[\hat{k}^*]}\}}  
+ ~\hat{u}(X_{n+1}) \cdot \sum_{i=1}^{n}w(X_i,Y_i)\ind_{\{V_i\leq V_{[\hat{k}^*]}\}}  \\
&\geq - w(X_{n+1},Y_{n+1}) \cdot \sum_{i=1}^{n}\big( w(X_i,Y_i) +\Delta\tilde{\ell}_i \big)\ind_{\{V_i\leq V_{[\hat{k}^*]}\}}  
+ ~\big( w(X_{n+1},Y_{n+1}) +\Delta\tilde{u}_{n+1} \big) \cdot \sum_{i=1}^{n}w(X_i,Y_i)\ind_{\{V_i\leq V_{[\hat{k}^*]}\}}\\
&= - w(X_{n+1},Y_{n+1}) \cdot \sum_{i=1}^{n} \Delta\tilde{\ell}_i  \ind_{\{V_i\leq V_{[\hat{k}^*]}\}}  
+  \Delta\tilde{u}_{n+1}  \cdot \sum_{i=1}^{n}w(X_i,Y_i)\ind_{\{V_i\leq V_{[\hat{k}^*]}\}}
\$
By construction, we know that $\Delta\tilde{\ell}_i\geq 0$ and $\Delta\tilde{u}_i \leq 0$ for $i=1,\dots,n+1$. 
Putting the lower bounds together, we obtain
\$
\textrm{(i)} &\geq   \bigg[ \sum_{i=1}^{n}w(X_i,Y_i)\ind_{ \{V_i\leq V_{[\hat{k}^*]} \}}\bigg] \bigg[  \Delta\tilde{u}_{n+1} + \sum_{i=1}^n  \Delta \tilde{u}_i  \ind_{\{V_i> V_{[\hat{k}^*]}\}}   \bigg] \\
&\qquad -\bigg[w(X_{n+1},Y_{n+1}) +  \sum_{i=1}^{n}w(X_i,Y_i) \ind_{\{V_i> V_{[\hat{k}^*]}\}} \bigg] \bigg[   \sum_{i=1}^{n} \Delta\tilde{\ell}_i  \ind_{\{V_i\leq V_{[\hat{k}^*]}\}} \bigg] \\
&\geq \bigg[ \sum_{i=1}^{n}w(X_i,Y_i) \bigg] \bigg[  \sum_{i=1}^{n+1}  \Delta \tilde{u}_i   \bigg] - \bigg[ \sum_{i=1}^{n+1}w(X_i,Y_i)   \bigg] \bigg[   \sum_{i=1}^{n} \Delta\tilde{\ell}_i    \bigg].
\$
Since the term (ii) is non-negative, we have the lower bound
\$
\EE\bigg[\frac{\textrm{(i)}}{\textrm{(ii)}}\bigg] = & \EE\,\Bigg[\, \frac{\sum_{i=1}^{n}w(X_i,Y_i)\ind_{\{V_i\leq V_{[\hat{k}^*]}\}}}{\sum_{i=1}^{n+1}w(X_i,Y_i)} -  \,\frac{\sum_{i=1}^{n}\hat{\ell}(X_i)\ind_{\{V_i\leq V_{[\hat{k}^*]}\}}}{\sum_{i=1}^n \hat{\ell}(X_i)\ind_{\{V_i\leq V_{[\hat{k}^*]}\}} + \sum_{i=1}^n \hat{u}(X_i)\ind_{\{V_i> V_{[\hat{k}^*]}\}}+\hat{u}(X_{n+1}) } \Bigg] \\
&\geq \EE\,\Bigg[\, \frac{\big[ \sum_{i=1}^{n}w(X_i,Y_i) \big] \big[  \sum_{i=1}^{n+1}  \Delta \tilde{u}_i   \big] - \big[ \sum_{i=1}^{n+1}w(X_i,Y_i)   \big] \big[   \sum_{i=1}^{n} \Delta\tilde{\ell}_i    \big]}{ \big[\sum_{i=1}^{n+1}w(X_i,Y_i)\big] \big[  \sum_{i=1}^n \hat{\ell}(X_i)\ind_{\{V_i\leq V_{[\hat{k}^*]}\}} + \sum_{i=1}^n \hat{u}(X_i)\ind_{\{V_i> V_{[\hat{k}^*]}\}}+\hat{u}(X_{n+1}) \big]   }\Bigg] \\
&\stackrel{\textrm{(a)}}{\geq} \EE\,\Bigg[\, \frac{\big[ \sum_{i=1}^{n}w(X_i,Y_i) \big] \big[  \sum_{i=1}^{n+1}  \Delta \tilde{u}_i   \big] - \big[ \sum_{i=1}^{n+1}w(X_i,Y_i)   \big] \big[   \sum_{i=1}^{n} \Delta\tilde{\ell}_i    \big]}{ \big[\sum_{i=1}^{n+1}w(X_i,Y_i)\big] \big[  \sum_{i=1}^n \hat{\ell}(X_i)  \big]   }\Bigg] \\
&\stackrel{\textrm{(b)}}{\geq} - \EE\,\Bigg[\, \frac{  \sum_{i=1}^{n+1}   [ \hat{u}(X_i) - w(X_i,Y_i) ]_-   }{   \sum_{i=1}^n \hat{\ell}(X_i) }\Bigg] -  \EE\,\Bigg[\, \frac{     \sum_{i=1}^{n} [\hat{\ell}(X_i) - w(X_i,Y_i)]_+ }{   \sum_{i=1}^n \hat{\ell}(X_i)    }\Bigg].
\$
Above, step (a) follows from the fact that $\hat{u}(X_i)\geq \hat{\ell}(X_i)\geq 0$,
and step (b) is due to the non-negativity of $w(X_i,Y_i)$. By H\"{o}lder's inequality, 
\$
\EE\,\Bigg[\, \frac{  \sum_{i=1}^{n}  [\hat{\ell}(X_i)- w(X_i,Y_i)  ]_+  }{ \sum_{i=1}^n \hat{\ell}(X_i)  }\Bigg] &\leq \bigg\|\frac{1}{n}\sum_{i=1}^{n}  \big(\hat{\ell}(X_i)- w(X_i,Y_i)  \big)_+ \bigg\|_p \cdot \Bigg\|\frac{n}{\sum_{i=1}^n \hat{\ell}(X_i) } \Bigg\|_q \\
&\stackrel{\textrm{(a)}}{\leq} \Big\|  \big(\hat{\ell}(X_i)- w(X_i,Y_i) \big)_+  \Big\|_p \cdot \Bigg\|\frac{n}{\sum_{i=1}^n \hat{\ell}(X_i) } \Bigg\|_q \\
&\stackrel{\textrm{(b)}}{\leq} \Big\|  \big(\hat{\ell}(X_i)- w(X_i,Y_i)  \big)_+  \Big\|_p \cdot \Bigg\| \frac{1}{\hat{\ell}(X_i) } \Bigg\|_q,
\$
where step (a) follows from Minkowski's inequality, and the step (b) follows from
\$
 \frac{n}{\sum_{i=1}^n \hat{\ell}(X_i) }  \leq \frac{1}{n} \sum_{i=1}^n \frac{1}{\hat{\ell}(X_i) }  
\$
as implied by Cauchy-Schwarz inequality. 
Similarly,  
\$
\EE\,\Bigg[\, \frac{  \sum_{i=1}^{n+1}  [\hat{u}(X_i)-w(X_i,Y_i) ]_-  }{ \sum_{i=1}^n \hat{\ell}(X_i)  }\Bigg] & \leq \Big\|  \big(\hat{u}(X_i)-u(X_i)  \big)_-  \Big\|_p \cdot \Bigg\| \frac{1}{\hat{\ell}(X_i) } \Bigg\|_q + \frac{1}{n} \Big\|  \big(\hat{u}(X_{n+1})-u(X_{n+1})  \big)_-  \Big\|_p \cdot \Bigg\| \frac{1}{\hat{\ell}(X_i) } \Bigg\|_q,
\$
where the $L_p$ norm for $X_{n+1}$ is with respect to $\tilde\PP$, hence 
\$
\frac{1}{n}\Big\|  \big(\hat{u}(X_{n+1})-u(X_{n+1})  \big)_-  \Big\|_p = \bigg\| \frac{w(X_i,Y_i)^{1/p}}{n} \cdot \big(\hat{u}(X_{i})-u(X_{i})  \big)_-  \bigg\|_p 
\$
Combining the above results, we have 
\$
\PP\big(Y_{n+1}\in \hat{C}(X_{n+1}) \big)  \geq 1-\alpha - \hat\Delta\cdot \big\|1 / \hat{\ell}(X_i)   \big\|_q,
\$
where
\$
\hat\Delta = \Big\|  \big(\hat{\ell}(X_i)-\ell(X_i)  \big)_+  \Big\|_p +  \Big\|  \big(\hat{u}(X_i)-u(X_i)  \big)_-  \Big\|_p  + \bigg\| \frac{w(X_i,Y_i)^{1/p}}{n} \cdot \big(\hat{u}(X_{i})-u(X_{i})  \big)_-  \bigg\|_p  ,
\$
which completes the proof of Theorem \ref{thm:mgn}.
\end{proof}

\subsection{Proof of Theorem~\ref{thm:whp}} \label{app:proof_thm_whp}

\begin{proof}[Proof of Theorem~\ref{thm:whp}] 
Throughout the proof, all statements are conditional on $\cD_\train$. 
By the independence of 
$\cD_\calib\cup\{(X_{n+1},Y_{n+1})\}$, 
the scores $\{V(X_i,Y_i)\}_{i\in \cD_\calib}$
are i.i.d.~and independent of $V(X_{n+1},Y_{n+1})$.

Recall that $G(\cdot)$ is defined in~\eqref{eq:def_G_t_formula}.
To begin with, we define 
\$
\hat{q} = \inf\big\{ t\colon {G}(t)\geq 1-\alpha  \},\quad \hat{q}_n = \inf\big\{ t\colon \hat{G}_n(t)\geq 1-\alpha  \}.
\$
For any fixed $\epsilon>0$, we have 
\$
\PP(\hat{q}_n \leq \hat{q} - \epsilon)& = \PP \big( \hat{G}_n( \hat{q} - \epsilon) \geq 1-\alpha\big) \\
&\leq \PP \big(  {G}( \hat{q} - \epsilon) \geq \hat{G}_n( \hat{q} - \epsilon) \geq 1-\alpha\big) 
+ \PP \big(  {G}( \hat{q} - \epsilon) < \hat{G}_n( \hat{q} - \epsilon)  \big) \leq \delta.
\$
Here the last inequality follows from the fact that 
${G}(\hat{q}-\epsilon)<1-\alpha$ for any fixed $\epsilon>0$ and 
$\PP\big( {G}( \hat{q} - \epsilon) < \hat{G}_n( \hat{q} - \epsilon)  \big) \leq \delta$ 
by \eqref{eq:hatG_t} with $t=\hat{q}-\epsilon$. 
Therefore, by the continuity of probability measures, we have 
\$
\PP (\hat{q}_n \geq  \hat{q}\, ) = 1-\lim_{\epsilon\to 0^+}\PP(\hat{q}_n \leq \hat{q} - \epsilon)\geq 1-\delta. 
\$
Moreover, on the event $\{\hat{q}_n\geq \hat{q}\,\}$, by the definition  of $\hat{C}(X_{n+1})$,  it holds for any $\TPP \in \cP(\PP,\ell,u)$ that
\$
\tilde\PP\big( Y_{n+1}\in \hat{C}(X_{n+1})\biggiven \cD_{\calib}\big) 
&=\tilde\PP\big(V(X_{n+1},Y_{n+1})\leq \hat{q}_n \biggiven \cD_{\calib}\big) \\
&\geq  \tilde\PP\big(V(X_{n+1},Y_{n+1})\leq \hat{q} \biggiven \cD_{\calib} \big) \\
&= \EE\big[ \ind_{\{V(X,Y)\leq \hat{q}\,\}}  w(X,Y) \biggiven \cD_{\calib}  \big],
\$
where the expectation is with respect to $(X,Y)\sim \PP $ 
independent of $\cD_{\calib}$. 
By the definition of ${G}(t)$ 
\$
\hat{q} = \inf\bigg\{ t\colon \max\Big\{ \EE\big[\ind_{\{V(X,Y)\leq t\}}\hat{\ell}(X)\big],~ 1- \EE\big[ \ind_{\{V(X,Y)>t\}} \hat{u}(X) \big]  \Big\}  \geq 1-\alpha  \bigg\} = \min\{\hat{q}_1,\hat{q}_2\},
\$
where $(X,Y)\sim \PP$ is an independent copy and we define 
\$
&\hat{q}_1 = \inf\Big\{ t\colon \EE\big[\ind_{\{V(X,Y)\leq t\}}\hat{\ell}(X)\big] \geq 1-\alpha\Big\},\\
&\hat{q}_2 = \inf\Big\{ t\colon 1- \EE\big[ \ind_{\{V(X,Y)>t\}} \hat{u}(X) \big] \geq 1-\alpha\Big\}.
\$
For constants $\hat{q}_1,\hat{q}_2$, we have 
\$
&\EE\big[ \ind_{\{V(X,Y)\leq \hat{q}\,\}}  w(X,Y) \biggiven \cD_{\calib}  \big] \\
&= \min\Big\{ \EE\big[ \ind_{\{V(X,Y)\leq  \hat{q}_1 \}}  w(X,Y) \biggiven \cD_{\calib} \big],~ 
\EE\big[ \ind_{\{V(X,Y)\leq  \hat{q}_2 \}}  w(X,Y) \biggiven \cD_{\calib}\big] \Big\}.
\$
We analyze the two terms separately. Firstly,  
\$
&\EE\big[ \ind_{\{V(X,Y)\leq  \hat{q}_1 \}}  w(X,Y) \biggiven \cD_{\calib} \big] \\
&= \EE\big[ \ind_{\{V(X,Y)\leq  \hat{q}_1 \}}  \hat{\ell}(X) \biggiven \cD_{\calib}\big]  
- \EE\big[ \ind_{\{V(X,Y)\leq  \hat{q}_1 \}} \big( \hat{\ell}(X) -w(X,Y)\big)\biggiven \cD_{\calib} \big] \\
&\stackrel{\textnormal{(a)}}{\geq} 1-\alpha  - \EE\big[ \ind_{\{V(X,Y)\leq  \hat{q}_1 \}} \big( \hat{\ell}(X) -w(X,Y)\big)\biggiven \cD_{\calib} \big] \\
&\geq 1-\alpha - \EE\big[\big( \hat{\ell}(X) - w(X,Y)\big)_+\big],
\$
where the step (a) follows from the definition of $\hat{q}_1$. 
Similarly, 
\$
&\EE\big[ \ind_{\{V(X,Y)\leq  \hat{q}_2 \}}  w(X,Y) \biggiven \cD_{\calib}  \big] \\
&= 1 - \EE\big[ \ind_{\{V(X,Y)> \hat{q}_2 \}}  w(X,Y) \biggiven \cD_{\calib} \big] \\
&= 1 - \EE\big[ \ind_{\{V(X,Y)\leq  \hat{q}_2 \}}  \hat{u}(X) \biggiven \cD_{\calib} \big]  
+ \EE\big[ \ind_{\{V(X,Y)\leq  \hat{q}_2 \}} \big( \hat{u}(X) -w(X,Y)\big)\biggiven \cD_{\calib}\big] \\
&\geq 
1-\alpha  + 
\EE\big[ \ind_{\{V(X,Y)\leq  \hat{q}_2 \}} \big( \hat{u}(X) -w(X,Y)\big)\biggiven \cD_{\calib} \big] \\
&\geq 1-\alpha - \EE\big[\big( \hat{u}(X) - w(X,Y)\big)_-\big],
\$
where the first equality follows from the fact that $w(x,y)$ is a likelihood ratio, and the first inequality follows from the definition of $\hat{q}_2$. Putting them together, on the event $\{\hat{q}_n\geq \hat{q}\,\}$ which happens with probability at least $1-\delta$ with repsect to $\cD_{\calib}$, it holds that 
\$
\tilde\PP\big( Y_{n+1}\in \hat{C}(X_{n+1})\biggiven \cD_{\calib} \big)  \geq 1-\alpha - \hat\Delta,
\$
where the gap is
\$
\hat\Delta = \max \Big\{ \EE \big[\big(\, \hat{\ell}(X)-w(X,Y)\big)_+ \big] ,
~ \EE\big[\big( \hat{u}(X)-w(X,Y)\big)_{-} \big]     \Big\},
\$
and the expectations are with respect to an independent copy $(X,Y)\sim \PP$. 
Therefore, we conclude the proof of Theorem \ref{thm:whp}.
\end{proof}

\subsection{Proof of Proposition \ref{prop:wsr_bound}}\label{appendix:proof_wsr_bound}

\begin{proof}[Proof of Proposition \ref{prop:wsr_bound}]
Let $t\in \RR$ be any fixed constant 
or any random variable in $\sigma(\cD_\train)$, 
and fix any  $\delta\in(0,1)$. 
We condition on $\cD_\train$ 
throughout the proof, so that 
$V$, $\hat{\ell}$ and $\hat{u}$ can be 
viewed as fixed and $t$ can be viewed as constant.

Since $\hat{G}_n(t) = \max\big\{\hat{G}_n^\ell(t),\hat{G}_n^U(t)\big\}$, by the definition \eqref{eq:def_G_t_formula}, it suffices to show that
\$
&\PP_{\cD_{\calib}}\Big(\hat{G}_n^\ell(t)\leq \EE\big[\ind_{\{V_i\leq t\}}\hat\ell(X_i) \big]\Big)\geq 1-\delta/2\\
\text{and}\quad &\PP_{\cD_{\calib}}\Big(\hat{G}_n^U(t)\leq 1- \EE\big[\ind_{\{V_i>t\}}\hat{u}(X_i) \big]\Big)\geq 1-\delta/2,
\$
Now let 
\$
f(X_i) = \ind_{\{V_i\leq t\}}\hat\ell(X_i)/M,
\quad 
h(X_i) = 1- \ind_{\{V_i>t\}}\hat{u}(X_i)/M,
\$ 
so that 
$f(X_i)$ and $h(X_i)$, $1\leq i\leq n$, are i.i.d.\ random variables 
in $[0,1]$. 
We rescale the bounds into $\hat{f}_n =\hat{G}_n^\ell(t)/M$ 
and $\hat{h}_n = (\hat{G}_n^U(t)-1+M)/M$, hence 
by the tower property of conditional expectations, it suffices to show that
\$
\PP_{\cD_{\calib}}\Big(\hat{f}_n \leq \EE\big[f(X_i)  \big] \Big)\geq 1-\delta/2
\quad \text{and}\quad 
\PP_{\cD_{\calib}}\Big(\hat{h}_n  \leq \EE\big[h(X_i) \big] \Big)\geq 1-\delta/2.
\$
We show the desired result for $f(\cdot)$ and that for $h(\cdot)$ naturally applies. 

Consider the filtration $\{\cF_i\}_{i \ge 1}$, 
where the $\sigma$-algebra $\cF_i = \sigma\{X_1,Y_1,\dots,X_i,Y_i\}$. 
Then by definition, $\hat\sigma^L_{i}$ and 
$\hat\mu^L_{i}$ are measurable with 
respect to $\cF_i$, and $\{\nu_i^L\}_{i\ge 1}$
is a predictable sequence with respect to $\{\cF_i\}_{i\ge 1}$.
Thus letting 
\$
f_0=\EE\big[f(X_i)\big] = \EE\big[\ind_{\{V_i\leq t\}}\ell(X_i)\big]/M,
\$
we have 
\$
\EE\big[ \cK_i^\ell(f_0) \biggiven \cF_{i-1}  \big]
= \cK_{i-1}^\ell(f_0) \cdot \Big(1 + \nu_j^L
\cdot\EE\big[ f(X_i) - f_0  \biggiven \cF_{i-1}\big] \Big)
= \cK_{i-1}^\ell(f_0).
\$
Thus $\{K_i^\ell(f_0)\}_{i=1}^n$ is a martingale. Meanwhile,
since $f(X_i)\in [0,1]$, we know $f(X_i)-f_0\geq -1$.
Hence $\cK_i^\ell(f_0)\geq 0$ for all $i\in [n]$. Therefore by Ville's inequality, 
\$
\PP_{\cD_{\calib}}\Big( \max_{1\leq i \leq n} \cK^L_i(f_0)>\frac{2}{\delta}  \Big) \leq \frac{\delta}{2}.
\$
Also note that $\cK_i(g)$ is a decreasing function in $g\in \RR$, hence 
\$
\PP_{\cD_{\calib}}\Big(\hat{f}_n > \EE\big[f(X_i)\big]\Big) \leq \PP\Big( \max_{1\leq i \leq n} \cK_i(f_0)>\frac{2}{\delta}  \Big) \leq \delta/2,
\$
Therefore, 
we conclude the proof of the desired results. 
\end{proof}

\subsection{Proof of Proposition~\ref{prop:tight_mgn}} \label{app:tight_mgn}
\begin{proof}[Proof of Proposition~\ref{prop:tight_mgn}]
Throughout the proof, we denote the generic random variables $(X,Y)\sim \PP$
and $V=V(X,Y)$. 
Denoting the right-hand side of \eqref{eq:tight_mgn} as $F^*(t)$, 
We are to show that 
\begin{enumerate}[(i)]
    \item $F^*(\cdot)\colon \RR\to [0,1]$ is a distribution function.
    \item $F^*(t)$ is a lower bound for $F(t;\,\cP(\PP,\ell,u))$ for all $t\in \RR$.
    \item $F^*(t)$ can be achieved by some element in $\cP(\PP,\ell,u)$. 
\end{enumerate}

First of all, 
(i) is straightforward by noting that $F^*(\cdot)$ is right-continuous due to the 
continuity of probability measures 
and $\lim_{t\to -\infty}F^*(t)=0$, $\lim_{t\to +\infty} F^*(t) = 1$. 
Also, similar to the arguments in the proof of Theorem~\ref{thm:whp}, 
for any $\tilde\PP\in \cP(\PP,\ell,u)$, we have 
\$
&F(t\,;V,\TPP) = \EE\bigg[ \ind_{\{V(X,Y)\leq t\}}\frac{\ud \TPP}{\ud \PP}(X,Y)\bigg] \geq \EE\big[ \ind_{\{V(X,Y)\leq t\}}\ell(X)\big], \\
&F(t\,;V,\TPP) = 1- \EE\bigg[ \ind_{\{V(X,Y)> t\}}\frac{\ud \TPP}{\ud \PP}(X,Y)\bigg] \geq 1-\EE\big[ \ind_{\{V(X,Y)> t\}}u(X)\big],
\$
hence (ii) follows. For (iii), we are to construct one distribution $\PP^*\in \cP(\PP,\ell,u)$
so that $F(\cdot\,;V,\PP^*)=F^*(\cdot)$. 

If $\EE[u(X)] = 1$, then since $\EE[w(X,Y)]=1$ for the likelihood ratio function, 
the collection 
$\cP(\PP,\ell,u) = \{\PP^*\}$ is a singleton with $\ud\PP^*/\ud\PP(x,y)=u(x)$, and 
\$
F\big(t\,;\cP(\PP,\ell,u) \big) = F(t\,;\PP^*) = 
\EE\big[\ind_{\{V(X,Y)\leq t\}} {u}(X) \big] = F^*(t),
\$
where the last equality follows from $\ell(x)\leq u(x)$. 
Then $\PP^*$ satisfies (iii). 
Similarly, for the case where $\EE[\ell(X)] =1$, the collection 
$\cP(\PP,\ell,u) = \{\PP^*\}$ is a singleton with $\ud\PP^*/\ud\PP(x,y)=\ell(x)$, and
\$
F\big(t\,;\cP(\PP,\ell,u) \big) = F(t\,;\PP^*) = 
\EE\big[\ind_{\{V(X,Y)\leq t\}} {l}(X) \big] = F^*(t),
\$
hence $\PP^*$ satisfies (iii). 
In the sequel, we consider the case where $\EE[\ell(X)] < 1 < \EE[u(X)]$.
We define 
\$
H(t)\,:=\, \EE\big[\ell(X)\cdot\ind_{\{V \le t\}} + u(X)\cdot\ind_{\{V>t\}}\big].
\$
By the construction and the continuity of probability measures, we have 
\$
\lim_{t \rightarrow \infty} H(t) = \EE \big[\ell(X)\big] < 1 < 
\EE \big[u(X)\big] = \lim_{t\rightarrow -\infty} H(t).
\$
We additionally define $t^* = \inf\{t\in\RR: H(t) \le 1\}$. Note that
$t^* < \infty$ and $H(t^*) \le 1$ by the right continuity of $H(t)$. 
We also define the left limit of $H(t)$ at $t^*$ as 
\$
    H_-(t^*) = \lim_{t \,\uparrow\, t^*} H(t),
\$
and define the weight as 
\$
\gamma = \frac{1 - H(t^*)}{H_-(t^*) - H(t^*)}\ind_{\{H_-(t^*)>1\}}
\$
Here by the definition of $t^*$, we have $H_-(t^*)\geq 1 \geq H(t^*)$, and the left limit takes the form
\$
H_-(t^*) = \EE\big[\ell(X)\cdot\ind_{\{V < t^*\}} + u(X)\cdot\ind_{\{V\geq t^*\}}\big].
\$
We now construct the worst-case distribution $\PP^*$ by 
\$
\frac{\ud \PP^*}{\ud \PP}(x,y) =  \gamma \cdot 
\big[\ell(x) \ind_{\{V(x,y)<t^*\}}+u(x) \ind_{\{V(x,y)\ge t^*\}} \big]
+(1-\gamma)\cdot
\big[\ell(x) \ind_{\{V(x,y)\le t^*\}} + u(x)\ind_{\{V(x,y)>t^*\}}\big].
\$
We denote the likelihood ratio $w^*(x,y) =\ud \PP^*/\ud \PP(x,y)$ as constructed above.
Note that
\$
\PP^*(\cX\times \cY) = 
\EE\big[w^*(X,Y)\big] = \gamma \cdot H_-(t^*) +
(1-\gamma) \cdot H(t^*) = 1,
\$
and also $\ell(x) \le w^*(x,y) \le u(x)$ for all $x\in \cX$. 
Therefore $\PP^*$ is a probability measure 
and is an element of $\cP(\PP,\ell,u)$.
In the following, we check that $F(t\,;V,\PP^*) = F^*(t)$ for all $t\in \RR$. 
Recall that we work with the case $\EE[\ell(X)] < 1< \EE[u(X)]$. 
For any constant $t<t^*$, by the construction of $w^*$,
\$
F(t\,;V,\PP^*) &= \EE\big[ w^*(X,Y) \ind_{\{V(X,Y)\leq t\}}  \big] \\
&=  \gamma \cdot 
\EE\big[\ell(X) \ind_{\{V(X,Y)\leq t\}}\big]
+(1-\gamma)\cdot
\big[\ell(X) \ind_{\{V(X,Y)\le t\}} \big]\\
&= \EE\big[\ell(X) \cdot \ind_{\{V\le t\}}\big] = F^*(t),
\$
where the last equality follows from the fact that 
\$
\EE\big[\ell(X) \cdot \ind_{\{V\le t\}}\big] >  1- \EE\big[u(X) \cdot \ind_{\{V > t\}}\big]
\$
since  $H(t) > 1$ for $t<t^*$. 
When $t = t^*$, note that 
\$
    &\EE\big[\ell(X) \cdot \ind_{\{V\le t^*\}} \big] - 1 +\EE \big[  u(X) \cdot \ind_{\{V > t^*\}}\big] 
= H(t^*) - 1\le 0,
\$
hence the right-hand side of ~\eqref{eq:tight_mgn} admits the form
\$
F^*(t^*) = 1 - \EE\big[u(X) \cdot \ind_{\{V > t^* \}}\big].
\$
Meanwhile, by the construction of $w^*(x,y)$, we have 
\$
F(t^*\,;V,\PP^*) &= \EE\big[ w^*(X,Y) \ind_{\{V(X,Y)\leq t^*\}}  \big] \\
&= \gamma \cdot 
\EE\big[\ell(X) \ind_{\{V(X,Y)<t^*\}}+u(X) \ind_{\{V(X,Y) = t^*\}} \big]
+(1-\gamma)\cdot
\big[\ell(X) \ind_{\{V(X,Y)\le t^*\}}\big] \\
&= \EE\big[\ell(X) \ind_{\{V(X,Y) \le t^*\}} \big] - \gamma \cdot 
\big( \EE\big[\ell(X)  \ind_{\{V(X,Y) = t^*\}} -  u(X) \cdot\ind_{\{V(X,Y) = t^*\}}\big] \big)\\
&=  \EE\big[\ell(X) \ind_{\{V(X,Y) \le t^*\}} \big] - \gamma \cdot 
\big( H_-(t^*) - H(t^*) \big) \\
&= \EE\big[\ell(X)  \ind_{\{V(X,Y)\le t^*\}}\big] + 1 - H(t^*) 
= 1 - \EE\big[u(X)\cdot \ind_{\{V>t^*\}}\big] = F^*(t^*).
\$
Similarly,  when $t > t^*$, by the cosntruction of $w^*(x,y)$ we have 
\$
F(t^*\,;V,\PP^*) &= 1 - \EE\big[ w^*(X,Y) \ind_{\{V(X,Y)> t \}}  \big]\\
&= 1-  \gamma \cdot 
\EE\big[u(X) \ind_{\{ V(X,Y)> t\}} \big]
+(1-\gamma)\cdot
\EE\big[ u(X)\ind_{\{ V(X,Y)> t\}}\big] \\
&= 1 - \EE\big[u(X) \ind_{\{ V(X,Y)> t\}} \big] = F^*(t),
\$
where the last equality follows from the fact that $H(t)\leq 1$
thus $1 - \EE [u(X) \ind_{\{ V(X,Y)> t\}} ] \geq  
\EE [\ell(X)\ind_{\{ V(X,Y)\leq  t\}} ]$. 
Combining the three cases, we arrive at $F^*(\cdot) = F(\cdot \,;V,\PP^*)$,
hence~\eqref{eq:tight_mgn} follows and we conclude the proof of Proposition~\ref{prop:tight_mgn}.
\end{proof}

\subsection{Proof of Proposition~\ref{prop:id_causal}} \label{app:id_causal}
\begin{proof}[Proof of Proposition~\ref{prop:id_causal}]
The proof proceeds by showing that $\cP\subset \cP(\PP, f,\ell_0,u_0)$ and 
$\cP(\PP, f,\ell_0,u_0)\subset \cP$, 
which together lead to the desired result. 

\paragraph{Step 1: $\cP\subset \cP(\PP, f,\ell_0,u_0)$.}
Let $\PP^\super$ be any super-population that satisfies~\eqref{eq:gamma_sel_mgn}
and~\eqref{eq:data_compat}. 
Due to the partial observation of potential outcomes, 
the observed distribution adimits the decomposition 
\$
\PP^{\obs}_{X,Y,T} = \PP^\obs_{T=1}\times \PP^\obs_{X,Y(1)\given T=1}  + 
\PP^\obs_{T=0}\times \PP^\obs_{X,Y(0)\given T=0}.
\$
Therefore, the data-compatibility condition~\eqref{eq:data_compat} 
is equivalent to
\#\label{eq:data_compat_2}
\PP_{T}^{\sup} = \PP^\obs_T,\quad 
\PP_{X,Y(1)\given T=1}^{\sup}  = \PP^\obs_{X,Y(1)\given T=1},\quad 
\PP_{X,Y(0)\given T=0}^{\sup}  = \PP^\obs_{X,Y(0)\given T=0}.
\#
where the latter two are further equivalent to
\#\label{eq:data_compat_3}
\PP_{X \given T=w}^{\sup}  = \PP^\obs_{X \given T=w}, \quad \PP^{\sup}_{Y(1)\given X,T=w}  = \PP^\obs_{ Y(1)\given X, T=w},\quad w\in\{0,1\}.
\#
Recall that $\TPP = \PP^\super_{X,Y(1)\given T=0}$ and $\PP = \PP^\obs_{X,Y(1)\given T=1}$. 
Then we have 
\$
\frac{\ud \TPP_X}{\ud \PP_X} = 
\frac{\ud \PP^\super_{X \given T=0} }{\ud \PP^\obs_{X\given T=1}} =
\frac{\ud \PP^\obs_{X \given T=0} }{\ud \PP^\obs_{X\given T=1}} = 
\frac{\PP^\obs(T=1) \cdot \PP^\obs(T=0)}{\PP^\obs(T=0)\cdot \PP^\obs(T=1)} = f(X),
\$
where the second equality follows from~\eqref{eq:data_compat_3} and 
the third equality follows from the Bayes rule. On the other hand,
the shift of conditional distribution is 
\$
\frac{\ud\TPP_{Y(1)\given X}}{\ud \PP_{Y(1)\given X}} = 
\frac{\ud \PP^\super_{Y(1)\given X,T=0}}{\ud \PP^\obs_{Y(1)\given X,T=1}} \in [1/\Gamma,\Gamma]
\$
according to Lemma~\ref{lem:shift}. Therefore, $\PP^\super \in \cP(\PP,f,\ell_0,u_0)$ by the definition. Hence we have $\cP\subset \cP(\PP, f,\ell_0,u_0)$.

\paragraph{Step 2: $\cP(\PP, f,\ell_0,u_0)\subset \cP$.} For this part, we are to show that
for any $\TPP\in \cP(\PP, f,\ell_0,u_0)$, there exists 
some $\PP^\super$ satisfying~\eqref{eq:gamma_sel_mgn} and~\eqref{eq:data_compat}
such that $\TPP = \PP^\super_{X,Y(1)\given T=0}$. 
Fixing an arbitrary probability distribution 
$\TPP\in \cP(\PP, f,\ell_0,u_0)$, we define the function 
\$
w(y\given x) = \frac{\ud\TPP_{Y(1)\given X}}{\ud \PP_{Y(1)\given X}}(y\given x),
\$
so that $w(y\given x)\in[1/\Gamma,\Gamma]$ for $\PP^\obs$-almost all $(x,y)\in \cX\times\cY$.
Also, since $\TPP$ is a distribution, we have 
\$
\EE\big[w(Y(1)\given x)\given X=x\big] = \TPP\big(Y(1)\in \cY\biggiven X=x\big) = 1
\$
for $\PP$-almost all $x\in \cX$, and the conditional expectation is induced by 
$(X,Y(1))\sim \PP = \PP^\obs_{X,Y(1)\given T=1}$. 
Applying Lemma~\ref{lem:sharp_shift} with $r(x,y) = w(y\given x)$
and $t=1$,  
we know 
there exists a distribution $\PP^\super$ over $(X,Y(0),Y(1),U,T)$
for some confounder $U$ 
that satisfies~\eqref{eq:data_compat} and~\eqref{eq:gamma_sel_mgn}, 
and
\$
\frac{\ud \PP^\super_{Y(1),X\given T=0}}{\ud \PP^\super_{Y(1),X \given T=1}}(y\given x) = w(y\given x).
\$
Since $\PP^\super$ satisfies~\eqref{eq:data_compat}, we have 
$\PP^\super_{Y(1),X \given T=1} = \PP^\obs_{Y(1),X \given T=1} = \PP_{Y(1)\given X}$
where we recall the definition of $\PP$, the distribution at hand. Hence 
\$
w(y\given x) =   
\frac{\ud \PP^\super_{Y(1),X\given T=0}}{\ud \PP_{Y(1)\given X}}(y\given x) = \frac{\ud\TPP_{Y(1)\given X}}{\ud \PP_{Y(1)\given X}}(y\given x).
\$
Therefore, we have $\TPP_{Y(1)\given X} = \PP^\super_{Y(1),X\given T=0}$. 
Furthermore, since 
$\PP^\super$ satisfies~\eqref{eq:data_compat}, we know 
\$
\frac{\ud \PP^\super_{X\given T=0} }{\ud \PP_X} = \frac{\ud \PP^\obs_{X\given T=0} }{\ud \PP^\obs_{X\given T=1}} = f(X) = \frac{\ud\TPP_{ X}}{\ud \PP_{ X}},
\$
where the second equality follows from the Bayes rule
and the last equality follows from the fact that $\TPP\in \cP(\PP,f,\ell_0,u_0)$. 
Thus, we have $\TPP_X = \PP^\super_{X\given T=0}$. 
Putting the two parts together, we have $\TPP = \PP^\super_{X,Y(1)\given T=0}$.
By the arbitrariness of $\TPP$, we arrive at 
$\cP(\PP, f,\ell_0,u_0)\subset \cP$.

Combining the two steps, we conclude the proof of Proposition~\ref{prop:id_causal}.
\end{proof}

\subsection{Sharpness of the identification set}

\begin{lemma}[Sharpness of Lemma~\ref{lem:shift}]\label{lem:sharp_shift} 
Given $t\in\{0,1\}$, a marginal
distribution $\PP^\obs$ over $(X,Y,T)$ and
a function $r(x,y) \in [1/\Gamma,\Gamma]$
such that 
\$
\EE^\obs\big[ r\big(X,Y(t)\big)  \biggiven X,T=t \big] = 1,\quad \PP^\obs\textrm{-almost surely},
\$
there exists a distribution $\PP^{\sup}$ over
$(X,Y(0),Y(1),U,T)$ for some confounder $U$
such that 
\begin{enumerate}[(i)]
    \item $\PP^{\sup}_{X,Y,T} = \PP^\obs_{X,Y,T}$ for $Y=Y(T)$;
    \item $\PP^\super$ satisfies the marginal $\Gamma$-selection condition;
    \item the likelihood 
    ratio is exactly $r(x,y)$, so that  $r(x,y) = \frac{\ud \PP^\super_{Y(t)\given X, T=1-t}}{\ud \PP^\super_{Y(t)\given X,T=t}}(x,y)$ for $\PP^{\sup}$-almost all $x,y$. 
\end{enumerate}
hold simulatenously. 
\end{lemma}

\begin{proof}[Proof of Lemma~\ref{lem:sharp_shift}]
Fix any marginal distribution $\PP^\obs$ over $(X,Y,T)$ for $Y=Y(T)$,
and a function $r(x,y)\in[1/\Gamma,\Gamma]$ satisfying the given condition. 
We show the result for $t=1$, 
while that for $t=0$ follows exactly the same arguments. 

\paragraph{The construction of $\PP^\super$}
To begin with, we let the confounder be the counterfactual itself, 
so that $U=Y(1)$. The joint distribution is thus
\$
\PP^\super_{(X,Y(1),Y(0),T} = \PP^\super_{X,T}  \times \PP^\super_{(Y(1),Y(0))\given X,T}.
\$
We set the two parts separately. Firstly, we set
$
\PP^\super_{X,T} = \PP^\obs_{X,T}
$
for the distribution on $(X,T)$. 
The joint distribution of potential outcomes given $(X,T)$ admits
\$
\PP^\super_{(Y(1),Y(0))\given X,T} = \PP^\super_{Y(1)\given X,T} \times \PP^\super_{Y(0)\given X,T,Y(1)},
\$ 
Since our target is for $Y(1)$, we take a simple coupling where 
$Y(0)$ is independent of $(T,Y(1))$ conditional on $X$, so that we set
\#\label{eq:super_y0}
\PP^\super_{Y(0)\given X,T,Y(1)} = \PP^\super_{Y(0)\given X} = \PP^\super_{Y(0)\given X,T=0} = \PP^\obs_{Y(0)\given X,T=0}.
\#
On the other hand, we set $\PP_{Y(1)\given X,T}^\super$ for $T=0,1$ by
\#\label{eq:super_y1}
\PP^\super_{Y(1)\given X,T=1} = \PP^\obs_{Y(1)\given X,T=1}
\quad \textrm{and}\quad 
\frac{\ud \PP^\super_{Y(1)\given X,T=0} }{\ud \PP^\obs_{Y(1)\given X,T=1} }(y\given x) = r(x,y). 
\#
So far we've completed the pieces of constructing $\PP^\super$. 
It remains to check that it is indeed a probability measure. 
By construction, $\PP^\super_{X,T} = \PP^\obs_{X,T}$ is a
probability measure on $(X,T)$. Also, by the construction,
\$
\PP^\super\big( Y(1) \in \cY\given X,T = 1\big) = \PP^\obs\big( Y(1) \in \cY\given X,T = 1\big) = 1,
\$
and
\$
\PP^\super\big( Y(1) \in \cY\given X,T = 0\big) 
&= \EE^\obs\bigg[ \ind_{\{Y(1) \in \cY\}}\frac{\ud \PP^\super_{Y(1)\given X,T=0} }{\ud \PP^\obs_{Y(1)\given X,T=1} }\bigggiven X,T = 1\bigg] \\
&= \EE^\obs\Bigg[ \frac{\ud \PP^\super_{Y(1)\given X,T=0} }{\ud \PP^\obs_{Y(1)\given X,T=1} }\Bigggiven X,T = 1\Bigg] = \EE^\obs\big[ r\big(X,Y(1)\big)\biggiven X,T = 1\big]
    = 1,
\$
where the first equality is by the change-of-measure formula, 
the second equality follows from $1=\ind_{\{Y(1) \in \cY\}}$,
the third equality follows from the construction, and the last 
equality is the given condition on $r(x,y)$.
Thus, $\PP^\super_{Y(1)\given X,T}$ is a probability measure. 
Also, by the construction of
$\PP_{Y(0)\given X,T,Y(1)}$, we have 
\$
\PP^\super\big( Y(0) \in \cY\given X,T,Y(1)\big) = \PP^\obs\big( Y(0) \in \cY\given X,T = 0\big) = 1,
\$
hence $\PP^\super_{Y(0)\given X,T,Y(1)}$ is also a probability measure. 
Putting them together, we know that 
\$
\PP^\super_{X,T,Y(1),Y(0)} = \PP^\super_{X,T}\times 
\PP^\super_{Y(1)\given X,T} \times 
\PP^\super_{Y(0)\given X,T,Y(1)}
\$
is indeed a probability measure over $(X,T,Y(1),Y(0))$. 

\paragraph{Verify the properties} We now proceed to verify the three stated properties. 
For (i),
due to partial obervability, the observed distribution admits the decomposition 
\$
\PP^\obs_{X,Y,T} =  \PP^\obs_T \times \PP^\obs_{X,Y\given T} =
\PP^\obs_{T=1}\times \PP^\obs_{X,Y(1)\given T=1} + 
\PP^\obs_{T=0}\times \PP^\obs_{X,Y(0)\given T=0}.
\$
Similarly, the projection on the observable of $\PP^\super$ is 
\$
\PP^\super_{X,Y,T} =  \PP^\super_T \times \PP^\super_{X,Y\given T} =
\PP^\super_{T=1}\times \PP^\super_{X,Y(1)\given T=1} + 
\PP^\super_{T=0}\times \PP^\super_{X,Y(0)\given T=0}.
\$
Here by the construction of $\PP^\super_{X,T} = \PP^\obs_{X,T}$, we have $\PP^\obs_{T=w} = \PP^\super_{T=w}$ and $\PP^\super_{X\given T=w} = \PP^\obs_{X\given T=w}$ for $w\in\{0,1\}$. 
Also, $\PP^\super_{Y(w)\given X,T=w} = \PP^\obs_{Y(w)\given X,T=w}$ 
holds for $w\in \{0,1\}$ by~\eqref{eq:super_y0}
and~\eqref{eq:super_y1}. 
The equivalences altogether leads to $\PP^\super_{X,Y,T} = \PP^\obs_{X,Y,T}$. 
For (ii), by the Bayes rule, we have 
\$
r(x,y) = \frac{\ud \PP^\super_{Y(1)\given X,T=0} }{\ud \PP^\super_{Y(1)\given X,T=1} }(y\given x) 
&= \frac{\PP^\super(T=0\given X=x,Y(1)=y)}{\PP^\super(T=1\given X=x,Y(1)=y)} \cdot \frac{\PP^\super(T=1\given X=x)}{\PP^\super(T=0\given X=x)} \in [1/\Gamma,\Gamma],
\$
so $\PP^\super$ satisfies the marginal $\Gamma$-selection condition~\eqref{eq:gamma_sel_mgn}
hence (ii) is verified. Property (iii) has also been verified as above.  
So far, we've constructed $\PP^\super$ that satisfies all stated conditions 
and we conclude the proof of Lemma~\ref{lem:sharp_shift}.
\end{proof}

\subsection{Proof of Proposition~\ref{prop:tight_causal}} \label{app:tight_causal}
\begin{proof}[Proof of Proposition~\ref{prop:tight_causal}]
For simplicity, we denote~\eqref{eq:causal_worst} as $F^*(t)$, and aim to show that 
\begin{enumerate}[(i)]
    \item $F^*(t)$ is a distribution function;
    \item $F^*(t)$ is a lower bound for $F(t\,;V,\TPP)$ for all $t\in \RR$ and all $\TPP\in\cP(\PP,f,\ell_0,u_0)$;
    \item $F^*(t)$ can be achieved by $\PP^*\in \cP(\PP,f,\ell_0,u_0))$ where $\ud\PP^*/\ud \PP(x,y)=w^*(x,y)$ as defined in~\eqref{eq:causal_sharp_w*}. 
\end{enumerate}
To verify (i), we note that $F^*(t)$ is right continuous 
by the continuity of probability measures, as well as 
$\lim_{t\to-\infty} F^*(t)=0$
and $\lim_{t\to +\infty} F^*(t)=1$. To show (ii), we are to show that 
for any $\TPP\in \cP(\PP,f,\ell_0,u_0)$,
it holds $\PP$-almost surely that 
\$
\EE\big[  \ind_{\{V(X,Y)\leq t\}} w^*(X,Y) \biggiven X  \big] \leq 
\tilde\PP\big(  V(X,Y)\leq t \biggiven X \big) = 
\EE\bigg[  \ind_{\{V(X,Y)\leq t\}} \frac{\ud\TPP}{\ud \PP}(X,Y) \bigggiven X  \bigg].
\$
Fixing any $\TPP\in \cP(\PP,f,\ell_0,u_0)$, we denote the conditional likelihood as
$w_0(y\given x) = \ud\TPP_{Y\given X}/\ud\PP_{Y\given X}(y\given x)$,
so that $\ell_0(x)\leq w_0(y\given x) \leq u_0(x)$ for $\PP$-almost all $x\in \cX$. 
Then the marginal likelihood ratio is $w(x,y):= \ud\TPP/\ud\PP(x,y) = f(x)\cdot w_0(y\given x)$. 
Hence for any $t\in\RR$, 
\$
&\EE\big[  \ind_{\{V(X,Y)\leq t\}} w^*(X,Y) \biggiven X  \big] - \tilde\PP\big(  V(X,Y)\leq t \biggiven X \big) \\
&= \EE\big[  \ind_{\{V(X,Y)\leq t\}} w^*(X,Y) \biggiven X  \big] - 
\EE\big[  \ind_{\{V(X,Y)\leq t\}} w (X,Y) \biggiven X  \big] \\
&=\EE\big[  \ind_{\{V(X,Y)\leq t\}} \big( w^*(X,Y) - w(X,Y)\big) \biggiven X  \big] 
        \cdot \ind_{\{t < q(\tau(X)\,;X,\PP)\}} \\
&\qquad +
    \EE\big[  \ind_{\{V(X,Y)\leq t\}} \big( w^*(X,Y) - w(X,Y)\big) \biggiven X  \big] 
        \cdot \ind_{\{t \geq q(\tau(X)\,;X,\PP)\}}.
\$
We treat the two terms in the last summation separately. By the definition of $w^*(x,y)$, 
\$
&\EE\big[  \ind_{\{V(X,Y)\leq t\}} \big( w^*(X,Y) - w(X,Y)\big) \biggiven X  \big] 
        \cdot \ind_{\{t< q(\tau(X)\,;X,\PP)\}} \\
&= f(X) \cdot \EE\big[ \ind_{\{V(X,Y)\leq t\}}  \big( \ell_0(X) - w_0(Y\given X)\big) \biggiven X \big]
\cdot \ind_{\{t < q(\tau(X)\,;X,\PP)\}} \leq 0.
\$
Meanwhile, since $1=\EE[w_0(Y\given X)\given X] = \EE[w^*(X,Y)/f(X)\given X]=1$ holds $\PP$-almost surely, we have 
\$
&\EE\big[  \ind_{\{V(X,Y)\leq t\}} \big( w^*(X,Y) - w(X,Y)\big) \biggiven X  \big] 
        \cdot \ind_{\{t\geq q(\tau(X)\,;X,\PP)\}} \\
&= f(X) \cdot \EE\big[ \ind_{\{V(X,Y)> t\}}  \big( w_0(Y\given X) - w^*(X,Y)/f(X)  \big) \biggiven X \big]
\cdot \ind_{\{t \geq q(\tau(X)\,;X,\PP)\}} \\
&= f(X) \cdot \EE\big[ \ind_{\{V(X,Y)> t\}}  \big( w_0(Y\given X) - u_0(X)  \big) \biggiven X \big]
\cdot \ind_{\{t \geq q(\tau(X)\,;X,\PP)\}}  \leq 0.
\$
Summing them up and by the tower property of conditional expectations, 
it holds for any $t\in \RR$ that 
\$
F^*(t) = \EE\big[  \ind_{\{V(X,Y)\leq t\}} w^*(X,Y)  \big] \leq 
\tilde\PP\big(  V(X,Y)\leq t   \big),
\$
which verifies property (ii). Finally, we define $\PP^*$ by $\ud\PP^*/\ud\PP(x,y)=w^*(x,y)$. Then since $\EE[w^*(X,Y)/f(X)\given X=x]=1$, the marginal likelihood ratio satisfies 
$\ud\PP^*_X/\ud\PP_X = f(x)$, hence 
$\ud\PP^*_{Y\given X}/\ud\PP_{Y\given X}(x,y) = w^*(x,y)/f(x)$. 
To verify $\PP^*\in \cP(\PP,f,\ell_0,u_0)$, it remains to show 
$\ell_0(x)\leq \gamma_0(x) \leq u_0(x)$ when it is nonzero, i.e., 
$\PP(V(x,Y)=q(\tau(x)\,;x,\PP)\given X=x)>0$. In this case,
\$
\gamma_0(x) = \frac{1-\ell_0(x)\cdot \PP( V(x,Y)< q (\tau(x)\,;x,\PP )  \given X=x)
- u_0(x)\cdot \PP( V(x,Y)> q (\tau(x)\,;x,\PP )  \given X=x) }{\PP( V(x,Y) = q (\tau(x)\,;x,\PP )  \given X=x)}.
\$
Note that $\PP( V(x,Y)< q (\tau(x)\,;x,\PP )\given X=x) \leq \tau(x) = (u_0(x)-1)/(u_0(x)-\ell_0(x))$, hence 
\$
&1-\ell_0(x)\cdot \PP\big( V(x,Y)< q (\tau(x)\,;x,\PP )  \biggiven X=x\big) \\
&\leq u_0(x) - u_0(x) \cdot \PP\big( V(x,Y)< q (\tau(x)\,;x,\PP )  \biggiven X=x \big) \\
&= u_0(x) \cdot \PP\big( V(x,Y)\geq  q (\tau(x)\,;x,\PP )  \biggiven X=x \big)\\
&= u_0(x) \cdot \PP\big( V(x,Y)> q (\tau(x)\,;x,\PP )  \biggiven X=x \big) 
+ u_0(x) \cdot \PP\big( V(x,Y)= q (\tau(x)\,;x,\PP )  \biggiven X=x \big),
\$
which leads to $\gamma_0(x)\leq u_0(x)$. On the other hand, 
by the definition of quantiles, we have 
$\PP( V(x,Y)\leq q (\tau(x)\,;x,\PP )\given X=x) \geq \tau(x) = (u_0(x)-1)/(u_0(x)-\ell_0(x))$. Hence
\$
&\ell_0(x) \cdot \PP\big( V(x,Y) < q (\tau(x)\,;x,\PP )  \biggiven X=x\big) + 
\ell_0(x) \cdot \PP\big( V(x,Y) = q (\tau(x)\,;x,\PP )  \biggiven X=x\big) \\
&= \ell_0(x)\cdot \PP\big( V(x,Y) \leq q (\tau(x)\,;x,\PP )  \biggiven X=x\big) \\
&\leq 1- u_0(x) + u_0(x)\cdot \PP\big( V(x,Y) \leq q (\tau(x)\,;x,\PP )  \biggiven X=x\big) \\
&= 1- u_0(x)\cdot \PP\big( V(x,Y) > q (\tau(x)\,;x,\PP )  \biggiven X=x\big),
\$
which leads to $\gamma_0(x)\geq \ell_0(x)$. 
Therefore, we conclude the proof of 
Proposition~\ref{prop:tight_causal}.
\end{proof}

\subsection{Proofs of Proposition~\ref{prop:fwer}}
\label{appendix:proof_fwer}
\begin{proof}[Proof of Proposition~\ref{prop:fwer}]
Recall that $\Gamma^*$ is the smallest sensitivity parameter 
such that $\PP^\super \in \cP(\Gamma^*)$. 
By the definition of $\cR$ in~\eqref{eq:rej_set} and the nested property 
of the prediction sets, we know 
$
\cR = [1,\hat\Gamma)
$
if $C\cap \hat{C}(X_{n+1},1) = \varnothing$, 
where 
\$
\hat\Gamma = \sup\{ \Gamma \geq 1: C\cap \hat{C}(X_{n+1},\gamma)=\varnothing \},
\$
and $\cR = \varnothing$ otherwise. 
Also, by the nested nature of $H_0(\Gamma)$, we know 
$\cH_0 = [\Gamma^*,\infty)$ when $Y_{n+1}(1)-Y_{n+1}(0) \in C$
and $\cH_0=\varnothing$ otherwise. 
Hence 
\$
\textrm{mErr} &= \PP ( \cR\cap \cH_0 \neq \varnothing  ) \\
&= \PP \Big( \big\{Y_{n+1}(1)-Y_{n+1}(0) \in C\big\}\cap
\big\{ \exists~ \Gamma \geq \Gamma^*,~ 
C\cap \hat{C}(X_{n+1},\Gamma)=\varnothing \big\} \Big) \\
&\leq \PP \Big( \big\{Y_{n+1}(1)-Y_{n+1}(0) \in C\big\}\cap
\big\{  
C\cap \hat{C}(X_{n+1},\Gamma^*)=\varnothing \big\} \Big) \\
&\leq \PP \big( Y_{n+1}(1)-Y_{n+1}(0)  \notin \hat{C}(X_{n+1},\Gamma^*) \big).
\$
Following exactly the same arguments, we have 
\$
\textrm{dErr} &= \PP ( \cR\cap \cH_0 \neq \varnothing \given \cD_{\calib} ) \\
&= \PP \Big( \big\{Y_{n+1}(1)-Y_{n+1}(0) \in C\big\}\cap
\big\{ \exists~ \Gamma \geq \Gamma^*,~ 
C\cap \hat{C}(X_{n+1},\Gamma)=\varnothing \big\} \Biggiven \cD_{\calib} \Big) \\
&\leq \PP \Big( \big\{Y_{n+1}(1)-Y_{n+1}(0) \in C\big\}\cap
\big\{  
C\cap \hat{C}(X_{n+1}, \Gamma^*)=\varnothing \big\} \Biggiven \cD_{\calib}\Big) \\
&\leq \PP \big( Y_{n+1}(1)-Y_{n+1}(0)  \notin \hat{C}(X_{n+1},\Gamma^*)  \biggiven \cD_{\calib}\big),
\$
completing the proof of Proposition~\ref{prop:fwer}.
\end{proof}

\section{Additional simulation results} \label{app:simu}

\subsection{Additional results for Section~\ref{subsec:simu_mgn}} 
\label{app:simu_predict}

In this part, we provide additional simulation results 
on the counterfactual prediction task in Section~\ref{subsec:simu_mgn}.

\begin{figure}[H]
\centering 
\includegraphics[width=5in]{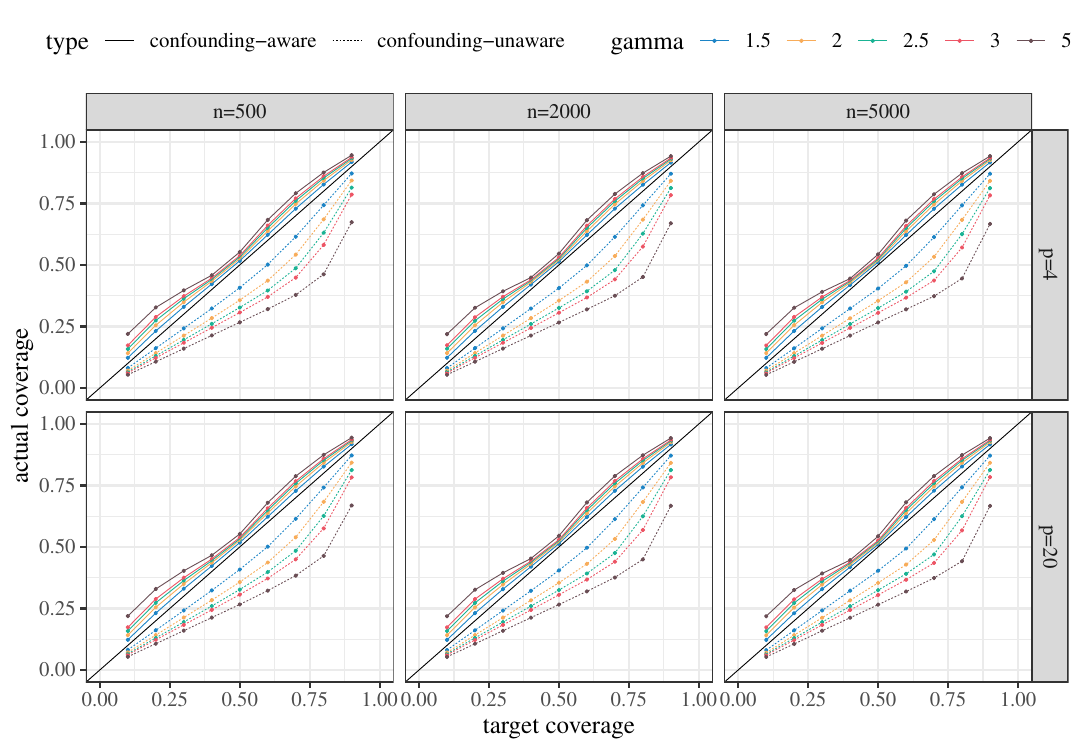}
\caption{Empirical (average) coverage of Algorithm~\ref{alg:mgn} when $\ell(\cdot)$
and $u(\cdot)$ are known. The details are otherwise the
same as in Figure~\ref{fig:mgn_cov_est}.}
\label{fig:mgn_cov_known}
\end{figure}

\subsection{Additional simulation results for Section~\ref{subsec:simu_whp}}
\label{app:simu_predict_whp}

In this part, we provide additional simulation results 
on the counterfactual prediction task in Section~\ref{subsec:simu_whp}.

\begin{figure}[H]
\centering 
\includegraphics[width=5in]{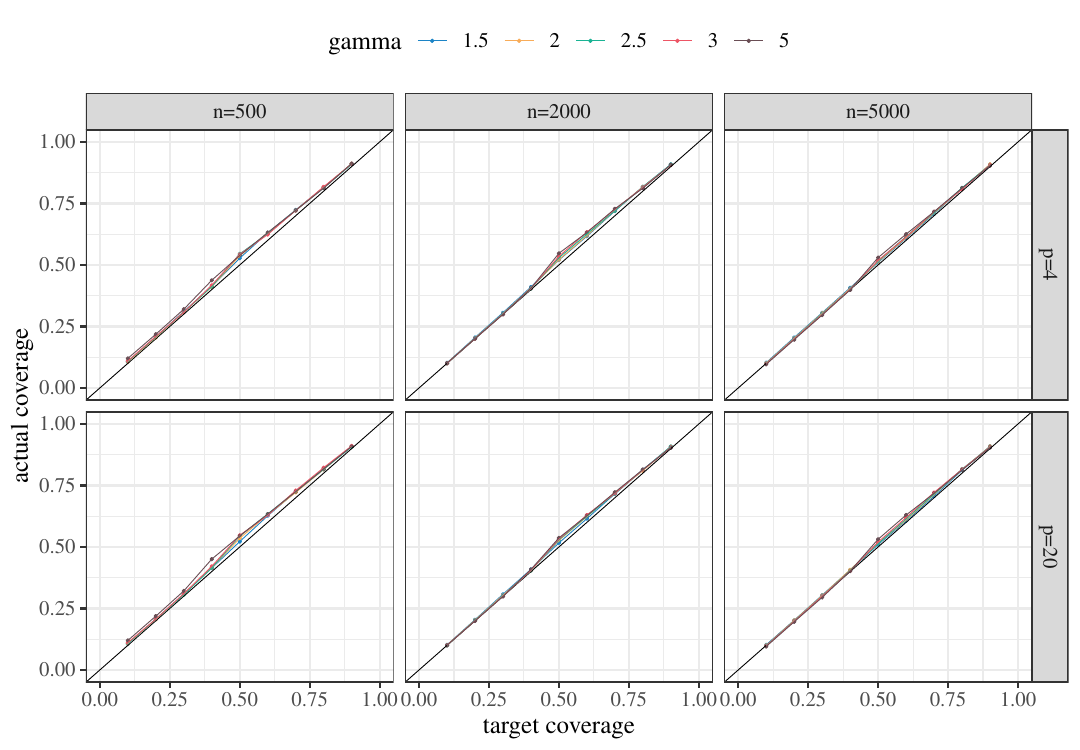}
\caption{$0.05$-th quantile of empirical coverage on test samples 
in Algorithm~\ref{alg:whp} when $\ell(\cdot)$, $u(\cdot)$ are known. 
The details are otherwise the same as in 
Figure~\ref{fig:whp_qt_est}.
}
\label{fig:whp_qt_known}
\end{figure}

\begin{figure}[H]
\centering 
\includegraphics[width=4.5in]{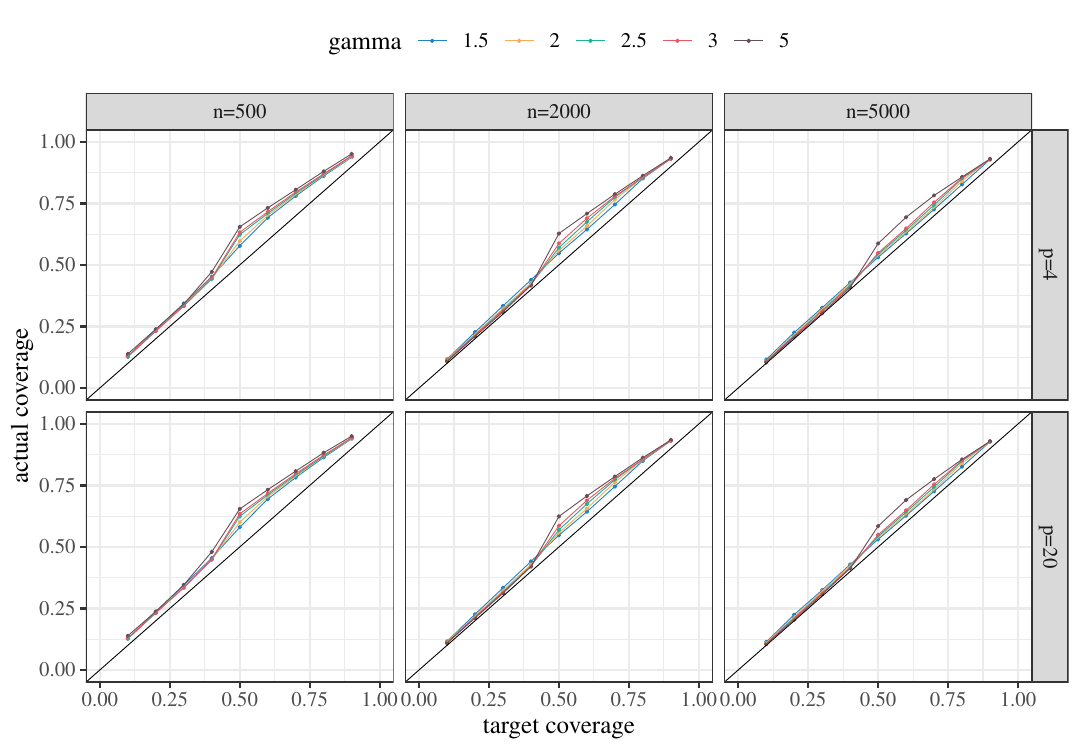}
\caption{Empirical (average) coverage of Algorithm~\ref{alg:whp} when $\ell(\cdot)$
and $u(\cdot)$ are known.
The details are otherwise the same
as in Figure~\ref{fig:whp_qt_est}.
}
\label{fig:whp_cov_known}
\end{figure}

\begin{figure}[H]
\centering 
\includegraphics[width=4.5in]{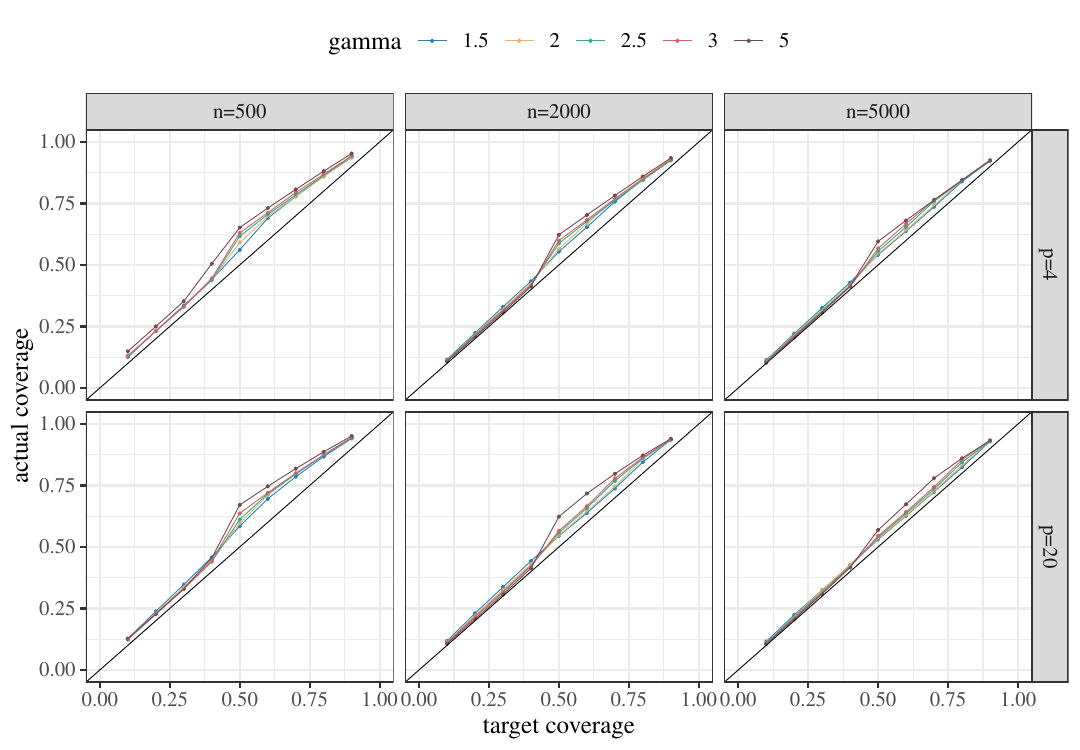}
\caption{Empirical (average) coverage of Algorithm~\ref{alg:whp} when $\hat \ell(\cdot)$
and $\hat u(\cdot)$ are estimated.
The details are otherwise the same
as in Figure~\ref{fig:whp_qt_est}.
}
\label{fig:whp_cov_est}
\end{figure}

\subsection{Additional results for Section~\ref{subsec:simu_sens}}
\label{app:simu_sens}
In this part, we collect the results 
in Section~\ref{subsec:simu_sens} 
of procedures with known bound functions.

\begin{figure}[H]
\centering
\begin{subfigure}[t]{0.43\linewidth}
    \centering
    \includegraphics[height=1.6in]{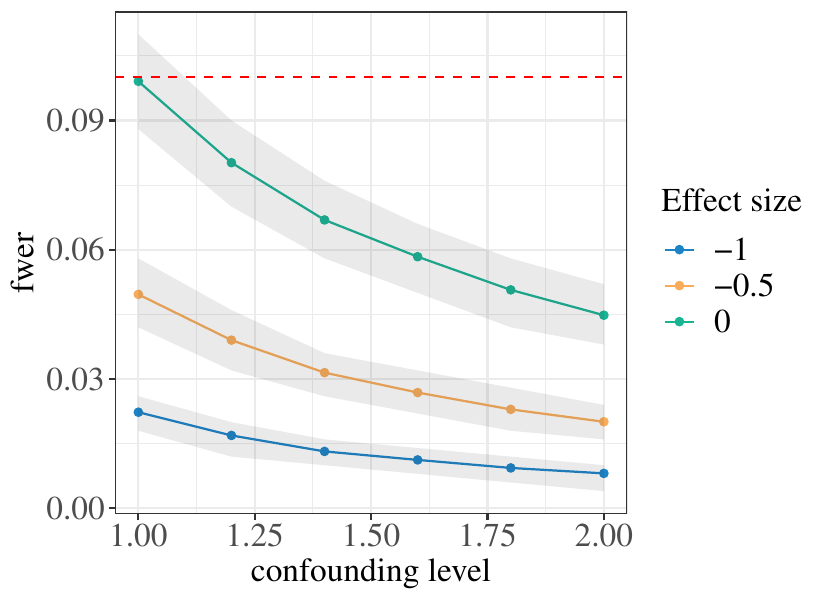}
\end{subfigure}
\begin{subfigure}[t]{0.43\linewidth}
    \centering
    \includegraphics[height=1.6in]{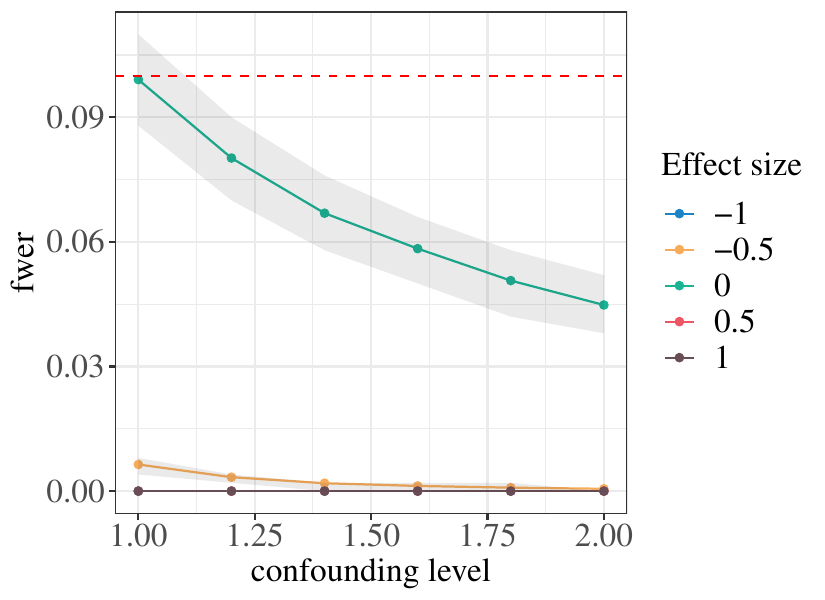}
\end{subfigure}
\caption{Empirical FWER for fixed ITE (left) 
and random ITE (right) with Algorithm~\ref{alg:mgn}.
The details are otherwise the same as in
Figure~\ref{fig:simu_ss_fwer}.
}\label{fig:simu_ss_mgn}
\end{figure}


\begin{figure}[htbp]
\centering
\begin{subfigure}[t]{0.43\linewidth}
    \centering
    \includegraphics[height=1.6in]{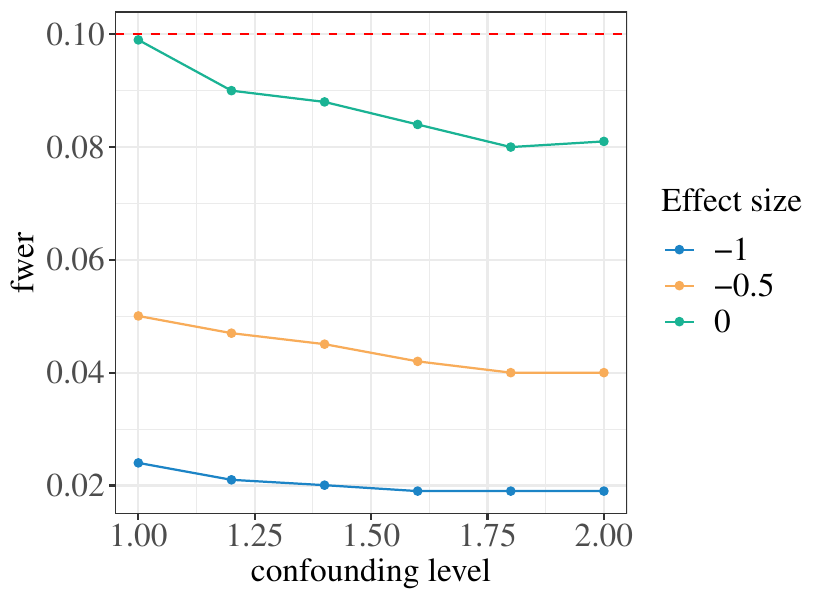}
\end{subfigure}
\begin{subfigure}[t]{0.43\linewidth}
    \centering
    \includegraphics[height=1.6in]{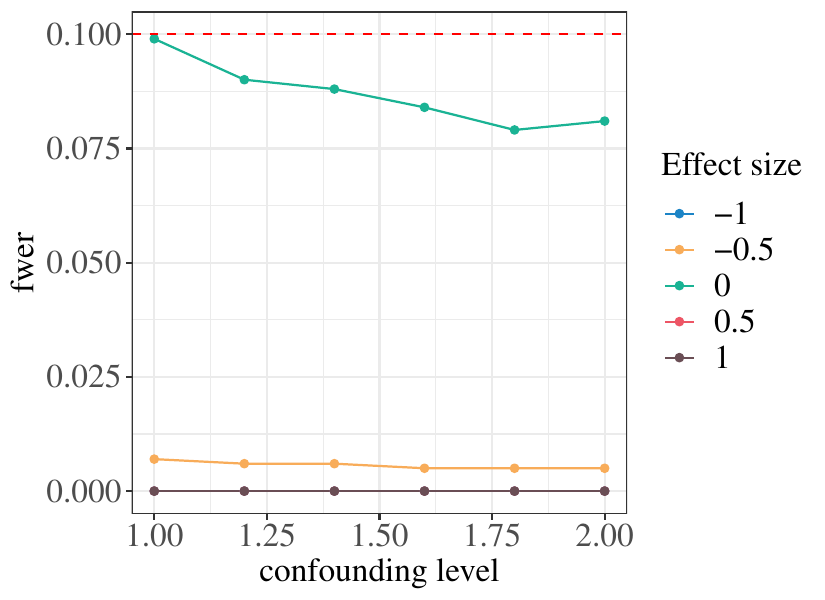}
\end{subfigure}
\caption{Empirical FWER for fixed ITE (left) and random ITE (right) with Algorithm~\ref{alg:whp}.
The details are otherwise the same
as in Figure~\ref{fig:simu_ss_fwer}.
}\label{fig:simu_ss_whp}
\end{figure}


\end{document}